%% file: main.tex
\documentclass[english,12pt]{article}

\ifx\pdftexversion\undefined\else
	\pdfoutput=1
	\usepackage[T1]{fontenc}
	\usepackage[utf8]{inputenc}
\fi

\usepackage{microtype}
\usepackage{babel}
\usepackage{lmodern}
\usepackage{csquotes}
\usepackage{fullpage}
\usepackage{mathtools}
\usepackage{amsthm}
\usepackage{amssymb}
\usepackage{hyperref}
\usepackage[style=numeric, backref, giveninits, maxbibnames=99,
	doi=true, isbn=false, url=true]{biblatex}
\usepackage{cleveref}
\usepackage{authblk}
	
\usepackage{tabu}
\usepackage{tikz-cd}
	\usetikzlibrary{matrix}
\usepackage[indent]{parskip}

\numberwithin{equation}{section}
\newtheorem{conjecture}{Conjecture}[section]
\crefname{conjecture}{conjecture}{conjectures}
\newtheorem{proposition}[conjecture]{Proposition}
\newtheorem{corollary}[conjecture]{Corollary}
\newtheorem{remark}[conjecture]{Remark}
\theoremstyle{definition}
\newtheorem{definition}[conjecture]{Definition}

\input{commands}
\title{From equivariant volumes to equivariant periods}
\author[a]{Luca Cassia}
\author[a,b]{Nicolò Piazzalunga}
\author[a]{Maxim Zabzine}
\affil[a]{Department of Physics and Astronomy, Uppsala University}
\affil[b]{NHETC and Department of Physics and Astronomy, Rutgers University}
\date{}

\begin{document}

\maketitle

\begin{abstract}
\noindent
\input{abstract}

\end{abstract}

\tableofcontents


\input{intro}

\input{setup}

\input{disk}

\input{bps}

\input{expansion}

\input{shift}

\input{quantum-coho}

\input{regularization}

\input{enumerative}

\input{examples}

\input{conclusion}

\paragraph{Acknowledgements}
We thank
Andrea Brini,
Alessandro Georgoudis,
Pietro Longhi,
Nikita Nekrasov,
Mauricio Romo,
and Johannes Walcher
for discussions.
We are also thankful to Michèle Vergne for feedback on an earlier draft of this manuscript.
The work of NP was partially supported by the US
Department of Energy under grant DE-SC0010008.
Opinions and conclusions expressed here are those of the authors
and do not necessarily reflect the views of funding agencies.

\appendix

\input{gamma-formulas}

\printbibliography
\end{document}

%% file: commands.tex
\newcommand\td[1]{\tilde{#1}}
\newcommand\ii{\mathrm{i}} 
\newcommand\dif{\mathop{}\!\mathrm{d}} 
\newcommand\eu{\mathrm{e}} 
\newcommand\pt{\mathrm{pt}} 
\newcommand\scl{\mathrm{cl}} 
\newcommand\inst{\mathrm{inst}} 
\newcommand\reg{\mathrm{reg}} 
\newcommand\cmp{\mathrm{cpt}} 
\newcommand\eq{\mathrm{eq}} 
\newcommand\sing{\mathrm{sing}} 
\newcommand\qjk{\mathrm{QJK}} 
\newcommand\jk{\mathrm{JK}} 
\newcommand\sft{\mathsf{T}} 
\newcommand\sfa{\mathsf{A}} 
\newcommand\sfp{\mathsf{P}} 
\newcommand\chamber{\mathfrak{C}} 
\newcommand\dchamber{\Lambda} 
\newcommand\e{\epsilon}
\newcommand\ve{\varepsilon}

\newcommand\al{\alpha}
\newcommand\BR{\mathbb{R}}
\newcommand\BP{\mathbb{P}}
\newcommand\BC{\mathbb{C}}
\newcommand\BZ{\mathbb{Z}}
\newcommand\BN{\mathbb{N}}
\newcommand\q{\mathfrak{q}} 

\newcommand\cN{\mathcal{N}}

\newcommand\co{\mathcal{O}} 
\newcommand\cf{\mathcal{F}} 
\newcommand\cz{\mathcal{Z}} 
\newcommand\cd{\mathsf{d}} 
\newcommand\cH{\mathcal{H}} 

\newcommand\bt{{\boldsymbol{t}}} 
\newcommand\bn{\boldsymbol{n}} 
\newcommand\bd{\boldsymbol{d}} 
\newcommand\bgamma{\boldsymbol{\gamma}}
\newcommand\bT{\boldsymbol{T}} 
\newcommand\fock{\mathrm{Fock}} 
\newcommand\fp{\mathrm{FP}} 
\newcommand\shift{T^\dagger} 

\renewcommand\hat\widehat

\DeclareMathOperator{\cD}{\mathcal{D}} 
\DeclareMathOperator{\Ga}{\Gamma} 
\DeclareMathOperator{\ch}{ch} 
\DeclareMathOperator{\pd}{PD} 
\DeclareMathOperator{\Li}{Li} 
\DeclareMathOperator{\Lie}{Lie} 
\DeclareMathOperator\Cone{Cone} 
\DeclareMathOperator\PF{\mathcal{L}} 
\newcommand\tot{} 

%% file: abstract.tex
We consider generalizations of equivariant volumes of abelian GIT quotients
obtained as partition functions of 1d, 2d, and 3d supersymmetric GLSM
on $S^1$, $D^2$ and $D^2 \times S^1$, respectively.
We define these objects and study their dependence
on equivariant parameters for non-compact toric Kähler quotients.
We generalize the finite-difference equations (shift equations)
obeyed by equivariant volumes to these partition functions. 
The partition functions are annihilated by differential/difference operators
that represent equivariant quantum cohomology/K-theory relations of the target
and the appearance of compact divisors in these relations plays a crucial role 
in the analysis of the non-equivariant limit. 
We show that the expansion in equivariant parameters contains information
about genus-zero Gromov--Witten invariants of the target.

%% file: intro.tex
\section{Introduction}

This work continues our investigation \cite{Nekrasov:2021ked} of
Duistermaat--Heckman localization formula for non-compact toric Kähler 
manifolds.
The original motivation comes from the study of higher-rank K-theoretic
Donaldson--Thomas theory on Calabi--Yau threefolds \cite{DelZotto:2021gzy}.
Let us sketch some ideas, while definitions are given in \cref{sec:setup}.
Consider the Kähler quotient
$X_\bt = \BC^N // U(1)^r$ with charge matrix $Q$. Its equivariant volume
\begin{equation} \label{eq:original-intro}
 \cf (\bt,\e) = \int_{X_\bt} \eu^{\varpi_\bt - H_\e}
\end{equation}
can be computed as a contour integral 
\begin{equation} \label{eq:cont-intro}
 \cf (\bt,\e) = \oint_\jk  \prod_{a=1}^r \frac {\dif \phi_a} {2\pi \ii} \,
 \frac {\eu^{ \sum_a \phi_a t^a }} {\prod_{i=1}^N
 \left(\e_i + \sum_a \phi_a Q^a_i\right)}
 ~.
\end{equation}
Once a chamber for $\bt$ is fixed,
the contour is given by the Jeffrey--Kirwan (JK) prescription.
In general, $\cf (\bt,\e)$ is a function of $\bt$ and the
equivariant parameters $\e$. If $X_\bt$ is compact, then $\cf (\bt,\e)$ is a
regular function around $\e = 0$ and $\cf (\bt,0)$ is a homogeneous polynomial
that encodes the intersection theory of $X_\bt$. If instead $X_\bt$ is not
compact, then $\cf (\bt,\e)$ has singular terms in $\e$ around $\e=0$,
and there is no canonical way to extract a polynomial in $\bt$ that could be
interpreted as intersection polynomial. The quantum mechanical analog of
$\cf (\bt,\e)$ is the equivariant count of states (holomorphic sections of
appropriate line bundles over $X_\bt$), which can be presented as
\begin{equation} \label{intro-KPF}
 \cz (\bT,q) = \sum_{Q \cdot \bn = \bT} \prod_{i=1}^N q_i^{n^i}
 ~,
\end{equation}
where $\bt = \hbar \bT$ and $q=\eu^{-\hbar \e}$. Here the sum is over integer
points inside the momentum polyhedron. The classical limit in $\hbar$ gives
the relation
\begin{equation}
\cf (\bt,\e) = \lim_{\hbar \to 0}~ \hbar^\cd \cz (\bT,q)
\end{equation}
with $\cd = \dim_\BC X_\bt = N-r$. If the manifold $X_\bt$ is compact, then the
sum in \cref{intro-KPF} has a finite number of terms since the momentum
polyhedron is compact. In this case $\cz (\bT,q)$ is a polynomial in $q$ and we 
can set $q=1$. Thus $\cz(\bT, 1)$ is a polynomial in $\bT$'s and its
highest-degree part is the classical intersection polynomial.
If instead $X_\bt$ is non-compact,
then $\cz (\bT,q)$ is a meromorphic function in $q$'s and there is no canonical
non-equivariant limit.
In the non-compact case the structure of $\cf (\bt,\e)$ and $\cz (\bT,q)$ is
controlled \cite{Nekrasov:2021ked} by the action of compact support cohomology
$H^\bullet_\cmp (X_\bt)$ on de Rham cohomology $H^\bullet (X_\bt)$.
If $H^2_\cmp (X_\bt)$ is non-empty, then the problem is controlled
by compact toric divisors
(which are Poincaré dual to elements of $H^2_\cmp (X_\bt)$).
This results in the shift equation
\begin{equation} \label{shift-eq-intro}
 \big( 1 - \eu^{-\sum_{i\in I_\cmp} m^i \cD_i} \big) \cf (\bt,\e) =
 \wp_\cd (\bt, m) + O(\e)
 ~,
\end{equation}
where $\cD_i= \e_i + Q^a_i \frac{\partial}{\partial t^a}$ are first-order 
differential operators in $\bt$ associated to divisors $D_i$ and $m^i$ are
auxiliary parameters. The sum runs over the set of compact toric divisors. This
equation allows us to define the intersection polynomial in a non-canonical way,
which requires a non-canonical embedding of $H^2_\cmp (X_\bt)$ into
$H^2 (X_\bt)$.
Another way to look at \cref{shift-eq-intro} is to present
$\cf (\bt, \e)$ as a sum of singular and regular terms
\begin{equation} \label{intro-reg-sing}
 \cf (\bt, \e) = \cf_\sing (\bt, \e) + p_\cd (\bt) + O(\e)
 ~,
\end{equation}
which cannot be done canonically, as there is always a trade-off between
$\cf_\sing (\bt, \e)$ and $p_\cd (\bt)$. Here $\cf_\sing (\bt, \e)$
is in the kernel of $\cD_i$ for all compact divisors.  
\Cref{shift-eq-intro} allows us to analyze
possible ambiguities in the representation via \cref{intro-reg-sing}.
A similar shift equation exists for $\cz (\bT,q)$ and can be analyzed similarly.
We can consider more general cases with the insertion of an equivariant 
cohomology class in \cref{eq:original-intro,eq:cont-intro}
\begin{equation}
 \cf_\alpha (\bt, \e)
 = \int_{X_\bt} \eu^{\varpi_\bt - H_\e} \alpha (R_\eq) = 
 \oint_\jk  \prod_{a=1}^r \frac {\dif \phi_a} {2\pi \ii} \,
 \frac {\eu^{\sum_a \phi_a t^a }} {\prod_{i=1}^N
 \left(\e_i + \sum_a \phi_a Q^a_i\right)} \alpha(\phi,\e)
\end{equation}
with $\alpha$ being a suitable function
of the equivariant curvature $R_\eq$.
The object $\cf (\bt, \e)$ can be
regarded as a generating function for such insertions, since they can be
generated by derivatives in $\bt$'s. The previous discussion of the
behavior around $\e=0$ can be extended to $\cf_\alpha (\bt, \e)$ and 
there is an analog of the shift equation for $\cf_\alpha (\bt, \e)$
on non-compact quotients.

Our goal is to extend these ideas
to more complicated objects such as the partition function on the
disk $\cf^D (\bt,\e;\lambda)$ and its K-theoretic generalization
$\cz^D (\bT,q;\q)$. 
What is the role of equivariant parameters in these generalizations? Is there 
an analog of shift equation?
How to extract a non-equivariant answer from the fully equivariant answer
and what are the possible ambiguities? What is the impact of these 
considerations on enumerative geometry of non-compact toric Kähler manifolds?

In this paper we study a generalization of the equivariant volume 
\cref{eq:cont-intro} 
\begin{equation} \label{eq:cont-gamma-intro}
 \cf^D (\bt,\e; \lambda) := \lambda^{-N} \oint_\qjk \prod_{a=1}^r \frac {\dif 
\phi_a} {2\pi \ii} \,
 \eu^{\sum_a \phi_a t^a } \prod_{i=1}^N \Ga
 \left (\frac{\e_i + \sum_a \phi_a Q^a_i }{\lambda} \right )
~,
\end{equation}
where the contour is specified by the quantum Jeffrey-Kirwan prescription,
discussed in \cref{sec:disk-functions}.
Physically, \cref{eq:cont-gamma-intro} is the partition function of a
(twisted) gauged linear sigma model (GLSM)
with worldsheet a disk and boundary condition a space-filling brane
\cite{Hori:2013ika, Sugishita:2013jca, Honda:2013uca},
based on earlier works \cite{Benini:2012ui, Doroud:2012xw} on $S^2$. 
The parameter $\lambda$ is an equivariant parameter on the disk, such that 
\begin{equation}
 \lim_{\lambda \rightarrow \infty} \cf^D (\bt,\e; \lambda) = \cf (\bt,\e)
\end{equation}
as we discuss in \cref{sec:expansions},
and the parameters $\e$'s are masses in the GLSM
(they are equivariant parameters from the target view-point).
We refer to $\cf^D (\bt,\e; \lambda)$ as the disk partition function.

In analogy with $\cf (\bt,\e)$, the disk partition function 
$\cf^D (\bt,\e;\lambda)$ has a K-theoretic lift, which we denote by
$\cz^D (\bT,q;\q)$, with $\q = \eu^{-\hbar \lambda}$.
This reduces to the known count when we collapse the disk,
$\cz^D (\bT,q;1) = \cz(\bT, q)$. In \cref{sec:BPScount} we discuss
the contour integral representation of $\cz^D (\bT,q;\q)$
and the equivalent representation given by the sum
\begin{equation} \label{intro-KDPF}
 \cz^D (\bT,q;\q) =
 \sum_{Q \cdot \bn = \bT} \prod_{i=1}^N \frac {q_i^{n^i}} {(\q;\q)_{n^i}}
 ~,
\end{equation}
which is the natural disk generalization of \cref{intro-KPF} and has a nice 
combinatorial interpretation.
By construction we have the relation (see \cref{sec:expansions})
\begin{equation}
 \lim_{\hbar \to 0} \hbar^\cd \cz^D(\bT,q;\q) =
 \cf^D (\bt, \e; \lambda)
 ~.
\end{equation}
The K-theoretic disk partition function $\cz^D (\bT,q;\q)$ is
the partition function on $D\times S^1$ of the 3d uplift of a 2d GLSM,
and it is related to holomorphic blocks \cite{Beem:2012mb}.

The function $\cf^D (\bt, \e; \lambda)$ is regular around $\e=0$
for compact quotients and singular for non-compact quotients.
The main issue is how to control the singular terms.
For every compact toric divisor, its equivariant volume $\cD_i \cf (\bt,\e)$
is regular around $\e=0$.
A priori, we cannot expect this to hold
for $\cD_i \cf^D (\bt,\e;\lambda)$,
since there is no geometric interpretation of this object. However,
we find that $\cD_i \cf^D (\bt,\e;\lambda)$ is regular at $\e=0$
for every compact divisor $\cD_i$.
Thus, we have a shift equation for the disk partition function 
\begin{equation} \label{shift-eq-disk-intro}
 \big( 1 - \eu^{-\sum_{i\in I_\cmp} m^i \cD_i} \big) \cf^D (\bt,\e;\lambda)
 = \text{regular}
\end{equation}
as well as a K-theoretic generalization of this equation. 
We explain these ideas in \cref{sec:shift}.

The disk partition function is the solution of
equivariant Picard--Fuchs (PF) equations
\begin{equation}
 \PF^\eq_\gamma \cf^D (\bt,\e;\lambda) = 0
\end{equation}
with prescribed semi-classics
\begin{equation}
 \cf^D (\bt,\e;\lambda) = \int_{X_{\bt}} \eu^{\varpi_{\bt}-H_{\e}}
 \hat \Ga_\eq + O(\eu^{-\lambda t})
 ~,
\end{equation}
where we insert the equivariant Gamma-class.
The equivariant PF differential operator $\PF^\eq_\gamma$
encodes quantum equivariant cohomology relations.
It depends on geometric data, on $\lambda$ and on $\e$'s.
If we send $\lambda \to \infty$, then $\PF^\eq_\gamma$ collapses to 
the classical equivariant cohomology relations.  
If instead we set all $\e=0$, then it becomes the standard PF operator.
(In the K-theoretic case,
quantum equivariant cohomology relations $\PF^\eq_\gamma$
are promoted to difference equations.)
The disk partition function $\cf^D (\bt,\e;\lambda)$ can be generalized
by changing the semi-classical expansion
and still requiring it to be annihilated by $\PF^\eq_\gamma$
\begin{equation}
 \PF^\eq_\gamma \cf^D_\alpha (\bt,\e;\lambda)=0~, \quad
 \cf^D_\alpha (\bt,\e;\lambda) = \cf_\alpha (\bt,\e) + O(\eu^{-\lambda t}) 
 ~.
\end{equation}
This way we can find a basis of solutions to equivariant PF 
equations (which we regard as equivariant periods).
To understand the singularities in $\e$'s we follow Givental's approach
\cite{MR1408320, Givental:9701016} to mirror symmetry and use the formalism
of Givental's equivariant $I$-function $I_{X_{\bt}}$
(and the corresponding Givental's operator $\hat I_{X_\bt}$)
to represent the disk partition function
\begin{equation}
 \cf^D (\bt, \e; \lambda)
 = \lambda^{-N}\oint_\jk \prod_{a} \frac{\dif\phi_a}{2\pi\ii}
 I_{X_{\bt}} \prod_{i}\Ga\left (\frac{\e_i + \sum_a \phi_a Q^a_i 
 }{\lambda} \right ) = \hat{I}_{X_{\bt}} \cdot \cf_{\Ga} (\bt, \e )
 ~.
\end{equation}
These ideas are discussed in \cref{sec:quantum-coh}.

In analogy with \cref{intro-reg-sing} we can represent the disk 
partition function as 
\begin{equation} \label{intro-sing-reg}
 \cf^D (\bt,\e;\lambda) =
 \cf^D_\sing (\bt,\e;\lambda) + \cf^D_\reg (\bt, \e ;\lambda)
 ~,
\end{equation}
where the singular term $\cf^D_\sing (\bt,\e;\lambda)$ is in the kernel of 
compact divisor operators $\cD_i$. 
This splitting is non-canonical and it requires some choices.
In \cref{sec:regularization} we study the relation between the shift equation 
and equivariant quantum cohomology relations encoded
in the equivariant PF equations.
The appearance of compact divisors in the equivariant Givental function is
related to the possible ways of calculating the splitting \cref{intro-sing-reg}.

Our function $\cf^D (\bt,\e;\lambda)$, being a GLSM quantity,
is related to the count of quasi-maps
\cite{Shadchin:2006yz, Bullimore:2018jlp, Ciocan_Fontanine_2010, MR3126932}
from the formal disk to a target $X_\bt$.
However, there's a difference: rather than a fixed boundary
condition at infinity for the adjoint scalar,
we sum over all possible choices, compatible with symmetries,
in a sense that is made precise in \cref{fprmk},
and the object we are computing is closer to the central charge
of a brane \cite{Aleshkin:2022eop, Knapp:2020oba}.
These are UV calculations.
After integrating out gauge fields, the theory of quasi-maps flows
in the IR to a non-linear sigma model, counting stable maps to the same target.
Turning on the $\Omega$-background $\lambda$ corresponds to
equivariant GW theory \cite{MR1408320}
on $X_\bt \times \BP^1$, counting maps of bidegree $(d,1)$,
with an $S^1$ action on $\BP^1$.  
In this work, we concentrate on structural aspects of
$\cf^D (\bt,\e;\lambda)$ and $\cz^D (\bT,q;\q)$ 
(and other generalizations, e.g.\ $\cf^D_\alpha (\bt,\e;\lambda)$)
for toric non-compact manifolds
and base our considerations on the integral representations and on
the equivariant Picard--Fuchs equation
(or its K-theoretic lift \cite{Givental:2016afst6, Givental:many, Lee:2001qkt}).

When the target is a Calabi--Yau three-fold, the RG flow corresponds to mirror
symmetry \cite{Witten:1993yc, Morrison:1994fr, Losev:1999nt}
and the semi-classical expansion of $\cf^D$
coincides with the central charge of a single D6-brane wrapping
$X_\bt \times S^1$ near large radius \cite{Hori:2013ika}, which is the natural
candidate for the classical action of DT theory \cite{DelZotto:2021gzy},
so it is natural to conjecture a relation to Gromov--Witten (GW) invariants.
In \cref{sec:Gromov-Witten} we show how to extract closed genus-zero GW
invariants from $\cf^D (\bt,\e;\lambda)$, or more precisely from 
$\cf^D_\reg (\bt, 0 ;\lambda)$, in the spirit of the relation
between GLSM localization calculations on $S^2$ 
and genus-zero closed GW invariants \cite{Jockers:2012dk,Bonelli:2013mma}.
The ambiguities in $\cf^D_\reg (\bt, 0 ;\lambda)$ translate into ambiguities
for GW invariants (but not for all spaces). 
We trace these ambiguities to some old issues for some of the 
examples in ref.~\cite{Chiang:1999tz},
where some rational Gopakumar--Vafa invariants appear.
We explain how, within our framework,
certain instanton sectors cannot be trusted when we 
take the non-equivariant limit, as certain quantum equivariant 
cohomology relations do not contain compact divisors.

After presenting the general theory, we go through a number of examples.
There are cases when all quantum equivariant cohomology relations contain
compact divisors and thus all singular terms sit within the semi-classical part,
for example local $\BP^1 \times \BP^1$ and local $\BP^2$. 
There can be other cases when some of the quantum equivariant 
cohomology relations do not contain 
compact divisors, and thus singular terms appear in specific parts of the 
instanton expansion, for example local $F_2$ and local $A_2$ spaces. 
We collect the examples with compact divisors in \cref{sec:examples1}.
In \cref{sec:examples2} we present a few examples without compact divisors.

%% file: setup.tex
\section{The setup} \label{sec:setup}

Let $\sfa = U(1)^r$ be a torus of rank $r$ acting on $\BC^N$
via an integer-valued matrix of charges $Q$
\begin{equation}
 Z^i \mapsto \eu^{\ii \sum_{a=1}^r \vartheta_a Q^a_i} Z^i,
 \quad i=1,\ldots,N
\end{equation}
for real variables $\vartheta_a$
and holomorphic coordinates $Z^i$ on $\BC^N$.
The corresponding momentum map is $\mu: \BC^N \to \BR^r = (\Lie \sfa)^*$
\begin{equation} \label{eq:moment-map}
 \mu^a (Z,\bar Z) = \sum_{i=1}^N Q_i^a |Z^i|^2, \quad a=1,\dots,r
.\end{equation}
Let $\bt = (t^1,\dots,t^r) \in \BR^r$ be a regular value for $\mu$,
and $\chamber \subseteq (\Lie \sfa)^*$ an open connected subset
of the set of regular values, containing $\bt$.
We call $\chamber$ a \emph{chamber}.
We consider toric Kähler manifolds
of complex dimension $\cd=N-r$ obtained by symplectic reduction
\begin{equation} \label{eq:symplectic-quotient}
 X_{\bt} = \mu^{-1}(\bt) / \sfa
 ~.
\end{equation}
They are equipped with a symplectic form $\varpi_\bt$.
The Kähler moduli space is partitioned into disjoint chambers,
such that two manifolds $X_{\bt}$ and $X_{\bt'}$ are
symplectomorphic iff $\bt$ and $\bt'$ are in the same chamber.
We define the dual of the cone $\chamber$ 
\begin{equation}
 \chamber^\vee := \left\{\bd\in\BR^r\,\middle|\,\sum_{a=1}^r
 d_a t^a\geq 0, \quad\forall \bt\in\chamber\right\}
 ~.
\end{equation}

We require $X_\bt$ to be smooth, which is equivalent
\cite{MR1362837,MR2091310} to the requirement that
any $r \times r$ minor of $Q$, such that $\bt$ lies
in the convex span of its columns, has determinant $\pm 1$.

On $X_\bt$ we have a non-faithful action of $\sft=U(1)^N$
inherited from the standard action on $\BC^N$,
whose matrix of charges is the $N \times N$ identity matrix.
The corresponding momentum maps are $p^i (Z,\bar Z) = |Z^i|^2$,
for $i=1,\dots,N$.
We define $\e_i \in H^2_\sft (\BC^N)$ to be the equivariant parameter
corresponding to the action of the $i$-th factor in $\sft$,
while $\phi_a \in H^2_\sfa(\BC^N)$ the one corresponding
to the action of the $a$-th factor in $\sfa$.
The variables $\phi_a$ descend to generators of $H^2 (X_{\bt})$
and they correspond to Chern roots of $r$ tautological line bundles
associated to the toric fibration $\mu^{-1}(\bt) \to X_{\bt}$.
We package momentum maps and equivariant parameters together, by writing
$\mu_\phi := \sum_{a=1}^r \phi_a \mu^a$ and $H_\e := \sum_{i=1}^N \e_i p^i$.
We introduce equivariant Chern roots
$x_i := \e_i + \sum_{a=1}^r \phi_a Q^a_i \in H^2_{\sft} (X_{\bt})$.
The Kähler moduli $t^a = \int_{C^a} \varpi_\bt$ can be obtained by
integrating the symplectic form $\varpi_\bt$
on a basis of cycles $C^a\in H_2(X_{\bt})$ dual to the classes $\phi_a$.

The equivariant cohomology\footnote
{If instead we work with the $\cd$-dimensional torus $\sft/\sfa$,
we have the isomorphism \cite{MR1057441}
\begin{equation}
 H_{\sft/\sfa}^\bullet (X_\bt) \cong \BC[x_1,\ldots,x_N] /
 I_\mathrm{SR}
.\end{equation}
This isomorphism identifies any variable $x_i$ with
the equivariant Chern class of toric divisor $D_i$.}
ring
\begin{equation} \label{eq:coho-ring}
 H_\sft^\bullet (X_\bt) \cong \BC[\phi_1,\dots,\phi_r,\e_1,\dots,\e_N]
 / I_\mathrm{SR}
\end{equation}
is isomorphic to the quotient of
the $(\sfa \times \sft)$-equivariant cohomology of $\BC^N$
by the Stanley--Reisner ideal $I_\mathrm{SR}$
generated by square-free monomials in the Chern roots
\begin{equation}
 I_\mathrm{SR} = \left\langle x_{i_1} \cdots x_{i_s}
 \,\middle| \, \Cone (u^{i_1},\dots,u^{i_s})
 \, \text{is not a cone of $\Sigma$}
 \right\rangle
 ~,
\end{equation}
where $\Sigma$ is the toric fan of $X_{\bt}$ generated by the vectors
$u^i\in\BZ^{N-r}$ defined by the property
\begin{equation}
 \sum_{i=1}^N Q^a_i u^i = 0
 ~.
\end{equation}
To each coordinate in $\BC^N$, we can associate a toric divisor
$D_i = \{ p^i = 0 \} \cap X_\bt$,
obtained as the symplectic reduction of the locus
where that coordinate is identically zero.
A toric divisor $D_i$ is compact if its corresponding
vertex $u^i$ is an interior point of the toric fan $\Sigma$.
Let us introduce the set
\begin{equation}
 I_{\cmp} := \{ i | \, D_i \text{ is compact} \}
 ~.
\end{equation}
We identify the equivariant Chern root $x_i\in H^\bullet_\sft(X_\bt)$ as the 
image of
$1 \in H^\bullet_\sft(D_i)$ under pushforward along the inclusion
$D_i \hookrightarrow X_\bt$.
In the non-equivariant setting, compact toric divisors in $H_{2d-2} (X_\bt)$
are Poincaré-dual to classes in cohomology with compact support
$H^2_\cmp (X_\bt)$, and similarly lower-dimensional compact cycles are dual
to higher-degree classes in $H^\bullet_\cmp (X_\bt)$.
In the equivariant setting, we regard $x_i$ as the equivariant upgrade
of the Poincaré dual of $D_i$, and
we use the fact that Poincaré duality send intersections to products as
$\pd(D_{i_1} \cap \dots \cap D_{i_s}) = x_{i_1} \cdots x_{i_s}$.
With a slight abuse of notation we use the same symbol for
equivariant and non-equivariant Poincaré duality.

The equivariant K-theory ring of $X_\bt$
\begin{equation}
K_\sft (X_{\bt}) \cong \BC[w_1^\pm,\dots,w_r^\pm,q_1^\pm,\dots,q_N^\pm]
 / I^K_\mathrm{SR}
\end{equation}
is described in terms of equivariant K-theoretic parameters
$w_a \in K_\sfa(\BC^N)$ and $q_i \in K_\sft(\BC^N)$.
It is isomorphic to the quotient of the $\sfa \times \sft$-equivariant
K-theory of $\BC^N$ by the ideal
\begin{equation} \label{eq:K-theory-rel}
 I^K_\mathrm{SR} = \left\langle (1-q_{i_1}\prod_{a} w_a^{Q^a_{i_1}}) \cdots
 (1-q_{i_s}\prod_{a} w_a^{Q^a_{i_s}})
 \,\middle| \, \Cone (u^{i_1},\dots,u^{i_s})
 \, \text{is not a cone of $\Sigma$}
 \right\rangle
\end{equation}
generated by polynomials in the K-theoretic Chern roots.

A toric quotient $X_\bt$ is Calabi--Yau (CY) iff
the first Chern class of its tangent bundle is zero,
which is equivalent to the requirement that the charges for each
$U(1)_a$ sum to zero,
\begin{equation}
 c_1(TX_\bt) = 0 \Longleftrightarrow \sum_{i=1}^N Q_i^a = 0~,
 \quad \forall a
 ~.
\end{equation}
From this constraint on the charges, it follows that all toric CYs are
non-compact, which implies that their volume is divergent. This forces us to
work equivariantly with respect to the torus $\sft$, so that equivariance
effectively regularizes all integrals over $X_\bt$.

\subsection{Cohomological partition function}

We compute equivariant symplectic volumes as integrals over $\sfa$-equivariant
parameters that implement the symplectic quotient
\begin{equation} \label{eq:eq-volume-informal}
 \int_{X_{\bt}} \eu^{\varpi_\bt-H_{\e}} \sim
 \int_{\BC^N} \prod_{i=1}^N \frac {\dif Z^i \dif\bar{Z}^i} {2\pi\ii}
 \int_{(\ii\BR)^r} \prod_{a=1}^r \frac {\dif \phi_a} {2\pi \ii} \,
 \exp \left[ \sum_a \phi_a t^a - H_{\e} - \mu_\phi \right]
 ~.
\end{equation}
If we perform the $Z^i$ integrals first, we can use the identity
\begin{equation}
\label{eq:coho-integral}
 \int_{\BC} \frac {\dif Z^i \dif\bar{Z}^i} {2\pi\ii}
 \exp \left[ - H_{\e} - \mu_\phi \right]
 = \int_0^\infty \dif p^i
 \eu^{- x_i p^i} = \frac{1}{x_i}
\end{equation}
and we are led to the following integral representation for the
equivariant volume
\begin{equation} \label{eq:cont-F-def}
 \cf (\bt,\e) := \oint_\jk \prod_{a=1}^r \frac {\dif \phi_a} {2\pi \ii} \,
 \frac {\eu^{\sum_{a} \phi_a t^a }} {\prod_{i=1}^N x_i}
 ~,
\end{equation}
where $(\ii \BR)^r$ is replaced by a contour defined via the Jeffrey--Kirwan
prescription \cite{MR1318878, MR1710758}.
The contour is defined in such a way that the integral can be computed
by iterated residues. The residues are specified by
arrangements of hyperplanes in $\BC^r$, i.e.\ choices of $r$-tuples of indices
$(i_1,\dots,i_r)$ that specify which of the denominators go to zero at the pole.
The JK prescription then says that the poles to be taken are those
for which the cone spanned by vectors $Q_{i_1},\dots,Q_{i_r}$
contains the chamber $\chamber$. Then we can define
\begin{equation}
\label{eq:JKpoles}
 \jk:=\left\{(i_1,\dots,i_r)\,\middle|\,
 \chamber \subseteq \Cone (Q_{i_1},\dots,Q_{i_r})\right\}
 ~.
\end{equation}
With this JK prescription for the residue computation, we can rewrite the
integral for $\cf$ via a fixed-point formula of Duistermaat--Heckman type
\begin{equation}
\label{eq:fpsum}
 \cf(\bt,\e) = \sum_{p \in \fp} \eu^{-H_\e(p)}
 \frac 1 {\prod_{j \notin p} \ve_j (p)}
 ~,
\end{equation}
where we identify JK poles with fixed points in $X_{\bt}$
\begin{equation}
 \fp \ni p = (i_1,\dots,i_r) \in \jk
 ~.
\end{equation}
The smoothness of $X_\bt$ allows us to invert the matrix
\begin{equation}
 Q_p = (Q_{i_1}|\dots|Q_{i_r}) \in \mathrm{SL} (r,\BZ)
\end{equation}
at each fixed point.\footnote
{To invert this matrix, it is sufficient that fixed points
are isolated. Smoothness implies that $\det Q_p = \pm 1$.}
At a JK pole the variables $\phi_a$ evaluate to
\begin{equation}
 \phi_a\equiv\phi_a(p) = -\sum_{b=1}^r \e_{i_b} (Q_p^{-1})^b_a
 ~.
\end{equation}
The local Hamiltonian
\begin{equation} \label{eq:local-hamiltonian}
 H_\e(p) = \sum_{a,b=1}^r \e_{i_b} (Q_p^{-1})^b_a t^a
\end{equation}
is a linear function of $\bt$ and $\e$,
obtained by evaluating $H_\e$ at the fixed point,
and the $\ve_i (p)$ are the weights of the normal bundle
to the fixed point w.r.t.\ the $\sft$-action
\begin{equation}
 \ve_{j}(p) = \e_{j}
 -\sum_{a,b=1}^r \e_{i_b} (Q_{p}^{-1})^b_a Q^a_j~,\quad
 \text{for}\, j=1,\ldots,N, \, j \notin p
 ~.
\end{equation}

The Kähler moduli $t^a$ are defined as conjugate variables to $\phi_a$'s,
therefore we can formally identify the equivariant Chern roots $x_i$
with the differential operators
\begin{equation}
\label{eq:differential-operators}
 \cD_i := \e_i + \sum_a Q^a_i \frac{\partial} {\partial t^a}
 ~.
\end{equation}
Acting with $\cD_i$ on the volume $\cf (\bt,\e)$ corresponds to
inserting $x_i$ in the integral in \cref{eq:cont-F-def}
\begin{equation} \label{eq:divisor-insertion}
 \cD_{i_1}\cdots\cD_{i_s} \cf(\bt,\e)
 = \int_{X_\bt}\eu^{\varpi_\bt-H_\e}x_{i_1}\cdots x_{i_s}
 ~,
\end{equation}
which computes the intersection number of the divisors $D_{i_1},\dots,D_{i_s}$.

Suppose $x_{i_1} \cdots x_{i_s}$ is a monomial in the ideal $I_\mathrm{SR}$ of
cohomology relations and therefore a zero element in the cohomology of $X_\bt$,
then we must have
\begin{equation}
\label{eq:classical-cohomology}
 \cD_{i_1}\cdots\cD_{i_s} \cf(\bt,\e) = 0
 ~,
\end{equation}
therefore $\cf(\bt,\e)$ is a D-module for the equivariant cohomology of $X_\bt$.

\subsection{K-theoretic partition function}

The natural generalization of the volume $\cf(\bt,\e)$ to K-theory is obtained
by computing the partition function of a supersymmetric QM on $S^1$ with target
space $X_{\bt}$. We can represent it as 1d GLSM with $N$ chiral fields
charged under the gauge symmetry $\sfa$ and flavor symmetry $\sft$.
Introduce K-theoretic equivariant parameters
\begin{equation} \label{eq:K-th-param}
 w_a = \eu^{-\hbar\phi_a} \in K_{\sfa}(\BC^N)~,
 \quad q_i = \eu^{-\hbar\e_i} \in K_{\sft} (\BC^N)
 ~,
\end{equation}
where $\hbar$ is the radius of $S^1$. The partition function of the QM is
the contour integral
\begin{equation}
 \cz (\bT,q) := \hbar^{-\cd} \oint_\jk \prod_{a=1}^r
 \frac {\dif\phi_a} {2\pi \ii} \,
 \frac {\eu^{\sum_a \phi_a t^a}} {\prod_{i=1}^N x_i}
 \prod_{i=1}^N \frac {\hbar x_i}{1 - \eu^{-\hbar x_i}}
\end{equation}
or equivalently, using the exponentiated parameters of \cref{eq:K-th-param},
\begin{equation}
 \cz (\bT,q) = (-1)^r\oint_\jk \prod_{a=1}^r
 \frac{\dif w_a} {2\pi \ii w_a} \, w_a^{-T^a} \,
 \frac{1}{\prod_{i=1}^N \left( 1 - q_i\prod_a w_a^{Q^a_i} \right)}
 ~,
\end{equation}
where $T^a = t^a / \hbar$ are rescaled Kähler moduli satisfying the 
quantization
condition $T^a \in \BZ$. The contour picks up the same poles as in the
cohomological setting, namely
\begin{equation}
 w_a\equiv w_a(p) = \prod_{b=1}^r q_{i_b}^{-(Q_p^{-1})^b_a}~,\quad
 \text{for } p=(i_1,\dots,i_r)\in\jk
\end{equation}
with JK defined as in \cref{eq:JKpoles}.

In analogy with \cref{eq:coho-integral} we can write the identity
\begin{equation}
 \sum_{n^i=0}^\infty \eu^{- \hbar x_i n^i}
 = \frac{1}{1 - \eu^{-\hbar x_i}}
 ~,
\end{equation}
so that one can interpret the infinite sum $\sum_{n^i=0}^\infty$ as the
``quantization'' of the integral over momenta $\int_0^\infty \dif p^i$ where
formally $\hbar n^i=p^i$.

By Hirzebruch--Riemann--Roch theorem, the index $\cz(\bT,q)$ is the push-forward
to the point of the K-theory class of a line bundle $L_{\bT}$ represented by
$\prod_a w_a^{-T^a}$, i.e.\ $\cz(\bT,q) = \chi (X_{\bt}, L_{\bT})$.
The K-theoretic version of Duistermaat--Heckman localization formula gives
\begin{equation} \label{eq:fpsumK-th}
 \cz(\bT,q) = \sum_{p \in \fp} \eu^{-H_\e(p)}
 \frac 1 {\prod_{j \notin p} \left(1-\eu^{-\hbar\ve_{j}(p)}\right)}
 ~.
\end{equation}

Similarly to \cref{eq:differential-operators}, we define difference operators
\begin{equation} \label{eq:delta-operator}
 \Delta_i := \eu^{-\hbar\cD_i} = q_i \prod_a (\shift_a)^{-Q^a_i}
 ~,
\end{equation}
where $\shift_a$ is the shift operator that acts by shifting $T^a$ by $1$,
\begin{equation}
 \shift_a f (T^1,\dots,T^r) = f (T^1,\dots,T^a+1,\dots,T^r)
 ~.
\end{equation}
Insertions of equivariant K-theory classes
$L_i:=\eu^{-\hbar x_i} = q_i \prod_a w_a^{Q_i^a}$,
the class of the line bundle corresponding to divisor $D_i$,
can be realized by acting with operators $\Delta_i$
\begin{equation}
 \Delta_i \cz(\bT,q) = \chi(X_{\bt},L_{\bT} \otimes L_i)
 ~.
\end{equation}
Similarly we have
\begin{equation}
 (1-\Delta_i) \cz(\bT,q) =
 \chi(X_{\bt},L_{\bT} \otimes \Lambda_{-1}^\bullet L_i)
 ~,
\end{equation}
where $\Lambda^\bullet_{-1}$ is the exterior power operator $\Lambda^\bullet_y V
:= \oplus_{n=0}^\infty y^n \Lambda^n V$, so that $\Lambda_{-1}^\bullet
L_i=(1-L_i)$. These identities are the K-theory
analogue of \cref{eq:divisor-insertion}.

To every relation in the equivariant K-theory of $X_\bt$, there corresponds an
element of the ideal $I_\mathrm{SR}^K$ defined in
\cref{eq:K-theory-rel}, to which we can associate a
finite difference equation for the partition function $\cz(\bT,q)$
\begin{equation}
\label{eq:classical-K-th}
 (1-\Delta_{i_1})\cdots(1-\Delta_{i_s}) \cz(\bT,q) = 0
\end{equation}
for $(1-L_{i_1})\cdots(1-L_{i_s})\in I^K_\mathrm{SR}$.

%% file: disk.tex
\section{The theory on the disk}
\label{sec:disk-functions}

We reviewed the construction of GLSM partition functions
on the point and on $S^1$.
In this section we uplift them to the backgrounds $D^2$ and $D^2 \times S^1$.
The space of fields now admits an additional $U(1)$ action
associated to rotations of the disk,
to which we assign an equivariant parameter $\lambda \in H^2_{U(1)} (D^2)$.
This is equivalent to an $\Omega$-background on the disk.
In the K-theoretic setup we define the variable
$\q = \eu^{-\hbar\lambda} \in K_{U(1)} (D^2)$, which acts as a fugacity for
the $U(1)$-symmetry in the counting of BPS states. The disk is fibered over
$S^1$ with holonomy $\q$, which corresponds to the $\Omega$-background
for 3d supersymmetric theories \cite{Beem:2012mb, Dimofte:2010tz}.

\subsection{Cohomological disk partition function}

We start by analyzing the 2d GLSM case.
Supersymmetric localization of $\cN=(2,2)$ theories on $D^2$
indicates that one-loop determinants of free chiral fields contribute as
$\lambda^{-1} \Ga \left( x_i / \lambda \right)$ and the partition
function of the GLSM is defined as follows.
\begin{definition}
The disk partition function is given by the integral
\begin{equation}
\label{eq:cont-F-def-D2}
 \cf^D (\bt,\e;\lambda) := \lambda^{-N}
 \oint_\qjk \prod_{a=1}^r \frac {\dif\phi_a} {2\pi\ii} \,
 \eu^{\sum_{a} \phi_a t^a } \prod_{i=1}^N
 \Ga \left( \frac{x_i}{\lambda} \right)~,
\end{equation}
where we define the \textit{Quantum Jeffrey-Kirwan} (QJK) contour via a
generalization of the JK prescription in the following way.
Every $\Ga$-function has a classical pole associated to the hyperplane
$x_i = 0$, corresponding to the same pole in \cref{eq:cont-F-def}.
To each such pole corresponds a tower of poles at $x_i + \lambda k = 0$ for
$k \in \BZ_{>0}$.
These integral shifts of the hyperplanes can be re-absorbed in a redefinition of
the corresponding $\e_i$, to which the JK prescription is blind. Hence,
if a classical pole is inside the classical JK contour,
then it is also picked up by the QJK contour and its infinite tower of higher
poles is picked up as well.
If instead a classical pole is not in the JK contour,
then that pole and its tower of higher poles are not in the QJK contour.
More concretely, we define the quantum JK poles as
\begin{equation}
 \qjk := \jk\times\BZ_{\geq0}^r
\end{equation}
so that at a QJK pole the variables $\phi_a$ evaluate to
\begin{equation}
 \phi_a\equiv\phi_a(p,k) = -\sum_{b=1}^r (\e_{i_b}+\lambda k_b) (Q_p^{-1})^b_a~.
\end{equation}
\end{definition}
\begin{remark}
From the definition of the QJK contour it follows that the disk partition
function $\cf^D(\bt,\e;\lambda)$ can be written via the fixed-point formula
\begin{equation} \label{eq:localization-Fdisk}
 \cf^D(\bt,\e;\lambda) = \sum_{k \in \BZ^r_{\geq0}}
 \frac {(-1)^{\sum_{i=1}^r k_i}} {\prod_{i=1}^r k_i!}
 \sum_{p \in \fp} \eu^{-H_\e(p,k)}
 \prod_{j \notin p} \Ga \left(\frac{\ve_{j}(p,k)}{\lambda}\right)~,
\end{equation}
where the Hamiltonian and local weights at $p \in \fp$ get shifted by $k$ as
\begin{equation}
 H_\e(p,k) := H_\e(p)
 + \lambda \sum_{a,b=1}^r k_b (Q_p^{-1})^b_a t^a~,
\end{equation}
\begin{equation}
 \ve_{j} (p,k) := \ve_{j}(p)
 - \lambda \sum_{a,b=1}^r k_b (Q_p^{-1})^b_a Q^a_j~.
\end{equation}
\end{remark}

The semi-classical part of $\cf^D (\bt,\e;\lambda)$ is the integral over the
classical JK contour
\begin{equation}
 \label{eq:Fscl-Ga-class}
 \cf_{\Ga} (\bt,\e)
 := \oint_\jk \prod_{a=1}^r \frac {\dif\phi_a} {2\pi\ii} \,
 \frac{\eu^{\sum_{a} \phi_a t^a }}{\prod_{i=1}^N x_i}
 \prod_{i=1}^N \Ga \left(1+\tfrac{x_i}{\lambda} \right)
 = \int_{X_\bt} \eu^{\varpi_\bt-H_\e} \hat \Ga (TX_\bt)
\end{equation}
so that we only pick up residues at the classical poles,
while we drop all higher poles.
Since the JK contour avoids the poles of the $\Ga$-function and picks up only
poles of the denominator, the factor $\prod_{i=1}^N \Ga(1+\tfrac{x_i}{\lambda})$
can be seen as the insertion of the $\Ga$-class of $X_{\bt}$
\cite{Iritani:20091016} in the integral for the equivariant volume $\cf$,
hence the notation $\cf_{\Ga}$.
Moreover, as $\cf_{\Ga}$ is a classical integral, it follows that it must
satisfy the same classical cohomology relations as in
\cref{eq:classical-cohomology}. This is however not true for the full disk
function $\cf^D$, which (as we show below) satisfies a quantum deformation
of cohomology relations.

\begin{remark}
We point out a few important properties of
the disk partition function $\cf^D$.
\begin{itemize}
\item The scaling property
\begin{equation}
\cf^D (\bt,\e;\lambda) = \lambda^{-\cd} \cf^D (\lambda\bt,\lambda^{-1}\e;1)~,
\end{equation}
which shows that it is a function of two dimensionless variables,
up to overall scaling.
\item The action of the differential operators $\cD_i$,
defined in \cref{eq:differential-operators},
\begin{equation}
\label{eq:disk-identity-D1}
 \left(\tfrac{\cD_i}{\lambda}\right)_n \cf^D(\bt,\e;\lambda)
 = \eu^{\lambda n\frac{\partial}{\partial\e_i}}
 \cf^D(\bt,\e;\lambda),
 \quad
 n\in\BZ_{\geq0}
\end{equation}
corresponds to shifts of equivariant parameters,
where $\eu^{\lambda\frac{\partial}{\partial\e_i}}$ is the operator
that sends $\e_i$ to $\e_i+\lambda$, and $(z)_n$ is the
Pochhammer symbol in \cref{eq:pochhammer}.
\item One can trade shifts in equivariant parameters
for shifts in Kähler parameters
\begin{equation}
\label{eq:disk-identity-D2}
 \eu^{\lambda \sum_{a,i} \gamma_a Q^a_i\tfrac{\partial}{\partial\e_i}}
 \cf^D(\bt,\e;\lambda) = \eu^{-\lambda \sum_a \gamma_a t^a}
 \cf^D(\bt,\e;\lambda), \quad \bgamma\in\BZ^r~,
\end{equation}
which follows from the change of variables
$\phi_a\mapsto\phi_a-\lambda\gamma_a$ inside the integral.
\end{itemize}
\end{remark}

From the localization formula in \cref{eq:localization-Fdisk} it is evident that
the disk function can be written as a sum of contributions weighted by the
``non-perturbative'' factors $\eu^{-\lambda k_a t^a}$.
These non-perturbative corrections are interpreted as instantonic
contributions to the 2d partition function that vanish in the large volume
limit. For later convenience we introduce the instanton counting
variables $z_a:=\eu^{-\lambda t^a}$ (not to be confused with the coordinates
on $\BC^N$) so that we can write $\cf^D$ as a power series in the $z$'s.

\subsection{K-theoretic disk partition function}

The one-loop determinant of a free chiral in a 3d $\cN=2$ supersymmetric gauge
theory on $D^2\times S^1$ \cite{Beem:2012mb,Yoshida:2014ssa} gives
$(\eu^{-\hbar x_i};\q)_\infty^{-1} = (q_i \prod_a w_a^{Q^a_i};\q)_\infty^{-1}$,
where we define the $\q$-Pochhammer symbol $(z;\q)_d$ as in
\cref{eq:q-pochhammer}.
\begin{definition}
We define the K-theoretic disk partition function
\begin{equation} \label{eq:cont-F-def-D2S1}
 \cz^D (\bT,q;\q) := (-1)^r
 \oint_\qjk \prod_{a=1}^r \frac {\dif w_a} {2\pi\ii w_a} \,
 w_a^{-T^a} \prod_{i=1}^N
 \frac {1} {\left( q_i\prod_a w_a^{Q^a_i};\q \right)_\infty}
\end{equation}
with QJK contour that selects the same poles as \cref{eq:cont-F-def-D2}.
\end{definition}

The partition function $\cz^D (\bT,q;\q)$ is the K-theoretic (3d)
refinement of the (2d) disk partition function $\cf^D (\bt,\e;\lambda)$.
One should think of it as a Witten index
on the space of holomorphic maps from $D^2$ to $X_{\bt}$,
computed via an infinite-dimensional version of Hirzebruch--Riemann--Roch.
Instead of trying to make this picture rigorous,
we use \cref{eq:cont-F-def-D2S1} as the definition of the index
and a simultaneous generalization of $\cf^D (\bt,\e;\lambda)$ and $\cz (\bT,q)$.

To make the connection to the 2d function $\cf^D (\bt,\e;\lambda)$ clear,
we rewrite the integrand in terms of Jackson $\q$-Gamma functions in
\cref{eq:q-gamma}. We then have the identity
\begin{equation} \label{eq:q-Gamma-integral}
 \cz^{D}(\bT,q;\q) =
 \frac{{\hbar^r}(1-\q)^{\sum_{i}\e_i/\lambda-N}}{(\q;\q)_\infty^N}
 \oint_\qjk \prod_{a=1}^{r} \frac{\dif{\phi}_{a}}{2\pi {\ii}} \,
 \eu^{\sum_a \phi_a\left(\hbar T^a+\lambda^{-1}\log(1-\q)\sum_i Q^a_i\right)}
 \prod_{i=1}^N \Ga_{\q} \left( \tfrac{x_i}{\lambda} \right)~,
\end{equation}
where the r.h.s.\ is a $\q$-deformation of the integral in
\cref{eq:cont-F-def-D2}.
If $X_\bt$ is a CY manifold,
as we assume in our examples,
then there is no shift in Kähler moduli
in the r.h.s.\ of \cref{eq:q-Gamma-integral}.

The semi-classical part is the contribution of the classical poles only
\begin{equation}
 \cz_{\Ga_\q} (\bT,q) := (-1)^r
 \oint_\jk \prod_{a=1}^r \frac {\dif w_a} {2\pi\ii w_a} \,
 w_a^{-T^a} \prod_{i=1}^N
 \frac {1} {\left( q_i\prod_a w_a^{Q^a_i};\q \right)_\infty}
\end{equation}
and it satisfies the relation $\cz_{\Ga_\q}(\bT,q) = \cz(\bT,q) + O(\q)$.\
Moreover, we can use the recurrence relation for the $\q$-Gamma in
\cref{eq:recurrence-q-gamma}, to write
\begin{equation}
\begin{aligned}
 \cz_{\Ga_\q} (\bT,q) &= \frac{(1-\q)^{\sum_{i}\e_i/\lambda} (-1)^r}
 {(\q;\q)_\infty^N}
 \oint_\jk \prod_{a=1}^r \frac {\dif w_a} {2\pi\ii w_a} \,
 \frac{w_a^{-T^a}}{\prod_{i=1}^N (1-\eu^{-\hbar x_i})} \prod_{i=1}^N
 \Ga_{\q} \left(1+ \frac{x_i}{\lambda} \right) \\
 &= \frac{(1-\q)^{\sum_{i}\e_i/\lambda}} {(\q;\q)_\infty^N}
 \left[\cz(\bT,q)+O\left(\lambda^{-1}\right)\right]
\end{aligned}
\end{equation}
so that, up to an overall factor, this computes the insertion of the
$\hat \Ga_{\q}$-class of $X_{\bt}$ in the 1d partition function $\cz(\bT,q)$.
From this observation it follows that the semi-classical function
$\cz_{\Ga_\q}$ satisfies the same set of K-theoretic relations as in
\cref{eq:classical-K-th}. This is however not true for the full disk
function $\cz^D$, which satisfies a quantum deformation of the
K-theory relations.

Using \cref{eq:q-pochhammer-difference} we obtain the useful identity
\begin{equation} \label{eq:disk-identity-delta}
 (\Delta_i;\q)_n \cz^D(\bT,q;\q)
 = \q^{n q_i\frac{\partial}{\partial q_i}} \cz^D(\bT,q;\q)~,
 \quad
 n\in\BZ_{\geq0}~,
\end{equation}
where $\Delta_i$ is defined in \cref{eq:delta-operator} and the operator
$\q^{q_i \frac {\partial} {\partial q_i}}$ sends $q_i$ to $\q q_i$.

%% file: bps.tex
\section{BPS states counting} \label{sec:BPScount}

We provide an interpretation of K-theoretic partition functions $\cz^D$ and 
$\cz$
as equivariant indices counting BPS states in the Hilbert space
of a certain quantum mechanics on $X_{\bt}$.
The physical theories have many $U(1)$ flavor symmetries,
whose fugacities are identified with K-theoretic equivariant parameters
\cite{Dimofte:2010tz, Nekrasov:2004vw}.

\subsection{Free theory}

We start with a quantum mechanical index on $\BC$.
Physically, this is the partition function of a free chiral field on $S^1$
charged under a flavor symmetry $\sft=U(1)$ with fugacity $q_1$.

The equivariant index is computed by the character map
$\ch: K_\sft (\BC) \xrightarrow{\cong} \BC [q_1^{\pm}]$,
applied to the Hilbert space of the QM. The computation goes as follows:
the single-particle Hilbert space $\cH$ is one-dimensional, generated by a
state of charge 1 under the $U(1)$ flavor symmetry,
hence its character is given by $\ch (\cH) = q_1$.
The full space of states of the QM is the Fock space
$\fock = S^\bullet \cH = \bigoplus_{n \geq 0} S^n \cH$, obtained
by summing over all symmetric tensor powers of $\cH$.
The index is given by the character of this space
$\cz(q_1) = \ch(S^\bullet \cH) = \frac {1} {1-q_1}$.
The index of two or more free chirals is the product of the indices of each
chiral, by the multiplicative nature of the character map.

Next we consider a 3d refinement of this counting.
Physically, we uplift the theory from the circle to $D^2\times S^1$.
The Hilbert space of this theory splits into components graded
both by the action of $\sft$ on the target
and $U(1)_\q$ on the disk.
As before we start by identifying the single-particle Hilbert space
$\cH^D \cong \bigoplus_{i\geq1} \cH^D_{(i)}$,
where $\cH^D_{(i)}$ are one-dimensional spaces
corresponding to an infinite tower of states coming from the disk.
All these components have charge one under the symmetry $\sft$
but they are distinguished by their $U(1)_\q$ charge
$\ch (\cH^D_{(i)}) = q_1 \q^{i-1}$.
The full space of states of the 3d theory is the Fock space
$\fock^D = S^\bullet \cH^D$ and its index
\begin{equation}
\label{eq:toyindex3d}
 \cz^D (q_1;\q)
 = \ch (S^\bullet \cH^D) = \frac{1}{(q_1;\q)_\infty}
\end{equation}
matches the one-loop determinant of a free chiral obtained via localization.

The states of the 1d theory are contained in the Hilbert space of the 3d theory
as those states with zero charge under $U(1)_\q$. In the limit $\q \to 0$,
all 3d states with higher $U(1)_\q$-charges decouple and the 3d index reproduces
the 1d index, $\lim_{\q \to 0} \cz^D (q_1;\q) = \cz (q_1)$.

A basis for the space $\fock^D$ is given by states of the form
\begin{equation}
\label{eq:basis-states}
 \alpha_{-i_1} \alpha_{-i_2} \cdots \alpha_{-i_n} |0\rangle~,
 \quad i_1 \geq \dots \geq i_n \geq 1~,
\end{equation}
where $\alpha_{-i}$ are mutually commuting creation operators
with charges $q_1 \q^{i-1}$.
Since the indices in \cref{eq:basis-states} are ordered,
we can label each state by an integer partition
$\mu = [i_1-1,i_2-1,\dots,i_n-1]$. So the index can be computed as
a sum of charges over the Fock space of all such states
\begin{equation}
\label{eq:exp-pochhammer}
 \cz^D(q_1;\q) = \sum_{n=0}^\infty \sum_{\ell(\mu)\leq n} q_1^n\q^{|\mu|}~,
\end{equation}
where the second sum ranges over all integer partitions $\mu$ of length less
or equal to $n$ (and arbitrary size). \Cref{eq:toyindex3d} can then be
recovered by using \cref{eq:p-len,eq:qbinomial}.

\subsection{Abelian GLSM}

We consider a toric variety $X_{\bt}$ obtained
as symplectic quotient of $\BC^N$
by the action of a torus $\sfa$ with momentum map $\mu$
as in \cref{eq:symplectic-quotient}.
The GLSM describing such quotient has $N$ chiral fields.
Each chiral field is charged both w.r.t.\ the flavor symmetry group $\sft$
and the gauge group $\sfa$, as specified by the corresponding matrix of charges.

Before looking at the gauged sigma model,
we consider the fully $(\sfa \times \sft)$-equivariant index
on the ambient space $\BC^N$,
where $\sfa$ is also regarded as a global symmetry.
This is the product of $N$ copies of the index in \cref{eq:toyindex3d},
each depending on the appropriate fugacities,
\begin{equation}
\label{eq:indexofCN}
 \prod_{i=1}^N \frac{1}{\left(q_i\prod_{a=1}^r w_a^{Q_i^a};\q\right)_\infty}
 = \sum_{\bn\in\BZ^N_+} \prod_{i=1}^N
 \frac{q_i^{n^i} \prod_{a=1}^r w_a^{Q_i^a n^i}}{(\q;\q)_{n^i}}~.
\end{equation}
We can restrict the sum over Fock space in the r.h.s.\ of \cref{eq:indexofCN}
to a given $\sfa$-charge sector $\cH_{\bT}$, $\bT=(T^1,\dots,T^r)\in\BZ^r$,
by imposing the Gauss law 
\begin{equation}
\label{eq:gauss-law}
 \sum_{i=1}^N Q_i^a n^i = T^a~,
 \quad
 a=1,\dots,r~.
\end{equation}
This can be implemented on \cref{eq:indexofCN} by the contour integral
\begin{equation} \label{eq:ZD-BPScount}
 \cz^D (\bT,q;\q) = (-1)^r \oint_\qjk \prod_{a=1}^r \frac
	{\dif w_a} {2\pi\ii w_a} w_a^{-T^a}
 \prod_{i=1}^N \frac{1}
	{\left( q_i \prod_{a=1}^r w_a^{Q_i^a};\q \right)_\infty}
 = \sum_{Q \cdot \bn = \bT} \prod_{i=1}^N \frac {q_i^{n^i}} {(\q;\q)_{n^i}}
 = \ch(\cH_{\bT})
\end{equation}
with a QJK contour defined as in \cref{sec:disk-functions}.
The Fock space of the linear sigma model splits as a sum over $\sfa$-charge
sectors, $\fock = \oplus_{\bT} \cH_{\bT}$ so that
\begin{equation}
 \ch(\fock) = \sum_{\bT} \cz^D(\bT,q;\q) \prod_{a=1}^r w_a^{T^a}
 =\prod_{i=1}^N \frac{1}{\left(q_i\prod_{a=1}^r w_a^{Q_i^a};\q\right)_\infty}~.
\end{equation}

Geometrically, the Gauss law constraint implements the restriction from $\BC^N$
to the stable locus $\mu^{-1}(\bt)$
and simultaneously the quotient w.r.t.\ the $\sfa$-action.
By comparing \cref{eq:gauss-law,eq:moment-map} we can interpret the index
as a certain graded count of integer points inside $X_{\bt}$,
where the integers $n^i$ replace the real momenta $p^i$.

%% file: expansion.tex
\section{Expansions}
\label{sec:expansions}

We study degeneration limits of the 3d partition function
$\cz^D(\bT,q;\q)$ corresponding to shrinking either the disk $D^2$,
the circle $S^1$ or both. These degenerations fit into the commutative diagram
of world-volume geometries
\begin{equation}
\label{eq:limit-geometry}
\begin{tikzcd}
 D^2\times S^1
 \arrow[d]
 \arrow[r]
 & S^1
 \arrow[d] \\
 D^2
 \arrow[r]
 & \pt
\end{tikzcd}
\end{equation}
to which we give an interpretation in terms of limits of partition functions.
For simplicity, in this section we assume that $X_\bt$ is CY.

It turns out that the limit in which the disk $D^2$ shrinks to zero-size
can be implemented by sending the equivariant disk parameter $\lambda$ to
$\infty$, so that the K-theoretic variable $\q$ goes to 0.
This limit corresponds to the horizontal arrows in \cref{eq:limit-geometry}.
Moreover, as we explain below, in this limit
the infinite towers of poles coming from the functions $\Ga$ and $\Ga_{\q}$
are sent to infinity and only classical poles survive.
For this reason the QJK contour can be shrunk back to the classical JK contour
when $\lambda$ is infinitely large.

The limit corresponding to vertical arrows in \cref{eq:limit-geometry},
in which the circle $S^1$ shrinks to zero radius,
is modulated instead by the parameter $\hbar$ going to 0.
This implies that all K-theoretic parameters go to one, as one would expect
from the reduction of K-theoretic computations to cohomology.

The two limits can be composed in two ways.
First reducing along the disk and then the circle or vice-versa.
We consider these two cases separately.
The main goal of this section is to show that these two paths
lead to the same result,
thus proving that we have a commutative diagram of partition functions
\begin{equation}
\begin{tikzcd}
 \cz^D(\bT,q;\q)
 \arrow[d,"\hbar \to 0"]
 \arrow[r,"\q \to 0"]
 & \cz(\bT,q)
 \arrow[d,"\hbar \to 0"] \\
 \cf^D (\bt,\e;\lambda)
 \arrow[r,"\lambda \to \infty"]
 & \cf(\bt,\e)
\end{tikzcd}
\end{equation}

\subsection{From 3d to 1d to 0d}

The degeneration of $\cz^D(\bT,q;\q)$ to $\cz(\bT,q)$ is rather straightforward 
to
implement. Each $\q$-Pochhammer factor in the integrand can be expanded
using \cref{eq:exp-pochhammer} and in the limit $\q\to0$ we find
\begin{equation}
 \lim_{\q\to0} \frac {1} {(\eu^{-\hbar x_i};\q)_\infty}
 = \frac {1} {1-\eu^{-\hbar x_i}}~.
\end{equation}
All the poles at $x_i+n\lambda = 0$ for $n>0$ are killed
by the $\lambda \to \infty$ limit
and one is left with the integral representation
for the 1d partition function $\cz(\bT,q)$
\begin{equation}
 \cz^D(\bT,q;\q) = \sum_{Q\cdot\bn=\bT} \prod_{i=1}^N
	\frac {q_i^{n^i}} {(\q;\q)_{n^i}}
 \xrightarrow {\q\to0}
 \sum_{Q\cdot\bn=\bT} \prod_{i=1}^N {q_i^{n^i}}
 = \cz(\bT,q)~,
\end{equation}
which agrees with previous results \cite{Nekrasov:2021ked}.

Next we reduce along the circle.
This is the cohomological limit of the Witten index $\cz(\bT,q)$
and it is known to reproduce the equivariant volume $\cf(\bt,\e)$.
We review here how the limit goes.
Using the series representation of the Todd genus
\begin{equation}
 \frac{\hbar x}{1-\eu^{-\hbar x}} =
 \sum_{n=0}^{\infty} \frac{B_n(-\hbar x)^n}{n!}
\end{equation}
we can write
\begin{equation}
 \cz(\bT,q) = \hbar^{r-N} \oint_\jk \frac {\dif\phi_i} {2\pi\ii}
 \frac {\eu^{\sum_a \phi_a t^a}} {\prod_{i=1}^N x_i} \left( 1+O(\hbar) \right)
 = \hbar^{-\cd} \cf(\bt,\e) + O(\hbar^{-\cd+1})~.
\end{equation}
Higher order corrections in $\hbar$ correspond to insertions of characteristic
classes of $X_{\bt}$.

\subsection{From 3d to 2d to 0d}

The degeneration of $\cz^D(\bT,q;\q)$ to the 2d partition function
$\cf^D(\bt,\e;\lambda)$ is slightly more involved and it requires to use
the representation in terms of Jackson $\Ga_{\q}$ function as in
\cref{eq:q-Gamma-integral}.
The function $\cz^{D} (\bT,q;\q)$ has infinitely many poles at $\q=1$,
therefore one needs to multiply it by $(\q;\q)_\infty^N$
to get a well-defined Laurent expansion.
Using the standard identity
$\lim_{\q \to 1} \Ga_{\q}(z) = \Ga(z)$, we then obtain
\begin{equation}
 (\q;\q)_\infty^N \cz^{D} (\bT,q;\q) =
	\hbar^{-\cd} \cf^{D} (\bt,\e;\lambda) + O(\hbar^{-\cd+1})~.
\end{equation}
The product is still divergent but it has a finite number
of negative powers of $\hbar$ in its Laurent series expansion.
Moreover, in the same limit we have
\begin{equation}
\label{eq:delta-limit}
 (1-\Delta_i) = (1-\eu^{-\hbar\cD_i}) = \hbar \cD_i + O (\hbar^2)
\end{equation}
so that the leading order in $\hbar$ of the difference operator $1-\Delta_i$
is the differential operator $\cD_i$.

The next limit is the zero-volume limit of the disk, $\lambda \to \infty$.
The function $\cf^D(\bt,\e;\lambda)$ depends on $\lambda$
through factors of $\Ga(\tfrac{x_i}{\lambda})$ in the integrand.
The limit of the integrand can be computed using
the series expansion of the $\Ga$-function
\begin{equation}
 \lim_{\lambda \to \infty} \frac1{\lambda}
 \Ga \left( \frac{x_i}{\lambda} \right) = \frac1{x_i}~.
\end{equation}
The QJK contour surrounds infinitely many poles of the $\Ga$-functions,
located at $x_i+\lambda k = 0$.
In the limit $\lambda \to \infty$ all of these poles run away to infinity
except for classical poles at $k=0$.
Therefore we obtain
$\lim_{\lambda \to \infty} \cf^D(\bt,\e;\lambda) = \cf(\bt,\e)$.
While the $\lambda \to \infty$ limit of $\cf^D(\bt,\e;\lambda)$ is well-defined,
its Laurent series expansion is not.
The reason is that one can expand $\cf^D(\bt,\e;\lambda)$
as a sum over infinitely many residues,
each residue at $x_i+\lambda k=0$ giving a contribution
proportional to $\eu^{- \lambda \bd\cdot\bt}$ times a power series
in $\lambda^{-1}$. Schematically,
\begin{equation}
\label{eq:lambda-expansion}
 \cf^D(\bt,\e;\lambda) = \sum_{\bd} \eu^{-\lambda\bd\cdot\bt}
 \times (\text{Laurent series in $\lambda^{-1}$})
\end{equation}
for $\bd$ a vector of integers ranging over a convex subset of $\BZ^r$ as we 
show in \cref{sec:Givental}.
In the $\lambda \to \infty$ limit, the $\eu^{-\lambda\bd\cdot\bt}$ contributions
go to zero exponentially fast (provided $\bd \cdot \bt>0$,
which we show in \cref{th:main_identity}),
therefore only the classical contributions at $\bd = 0$ survive.
The limit can then be computed by expanding $\cf_{\Ga}(\bt,\e)$ as a
Taylor series in $\lambda^{-1}$ as in \cref{eq:Fscl-Ga-class}.
Hence we can write
\begin{equation}
 \cf_{\Ga} (\bt,\e)
 = \cf (\bt,\e) - \tfrac{\gamma}{\lambda} \int_{X_{\bt}}
 \eu^{\varpi_{\bt}-H_{\e}} c_1 + O (\lambda^{-2})~,
\end{equation}
where we use the expansion of the Gamma-class of $TX_\bt$
as in \cref{eq:gamma-class}.

Due to the expansion in \cref{eq:lambda-expansion}, one should regard
contributions from higher poles as higher-order instanton corrections to
the classical partition function with instanton counting parameters
$z_a = \eu^{-\lambda t^a}$. By analogy with the genus zero Gromov--Witten theory
of the target $X_{\bt}$, one can interpret such contributions as coming from
higher-degree maps. See \cref{sec:Gromov-Witten} for a more detailed
discussion.

%% file: shift.tex
\section{Shift equations}
\label{sec:shift}

As discussed in \cref{sec:BPScount}, the disk partition function
$\cz^D (\bT,q;\q)$ is a graded count of integer points in $X_{\bt}$, or
equivalently the graded dimension of the space of sections of a certain
prequantum line bundle over the space of maps from the disk to $X_{\bt}$.
We want to know whether this function is well-defined when $\sft$-equivariant
parameters are turned off. This corresponds to the limit in which all $\e_i$
are set to zero, i.e.\ $q_i \to 1$ for $i=1,\dots,N$. As a generalization of
the volume of $X_{\bt}$, we can immediately see that this limit is not
defined if $X_{\bt}$ is non-compact, as the sum over integer points is
divergent. As a simple example consider the non-compact case of $X_{\bt} = \BC$,
then $\cz^D(\bT,q;\q) = (q_1;\q)^{-1}_\infty$, which has a simple pole at 
$q_1=1$.
On the other hand, if $X_{\bt}$ is compact, then the disk partition function
is a sum over a finite number of points and therefore it has a well-defined
limit for $q_i\to1$.

We argue that while $\cz^D (\bT,q;\q)$ does not have a
non-equivariant limit for $X_{\bt}$ non-compact, one can extract a convergent
quantity by applying a finite difference operator corresponding to a compact
toric divisor of $X_{\bt}$. This generalizes the shift equation
from ref.~\cite[Section 4]{Nekrasov:2021ked} to the disk partition functions
$\cf^D (\bt,\e;\lambda)$ and $\cz^D (\bT,q;\q)$.
The statement of regularity for $\cz^D$ requires an analysis
of the $q_i$ dependence of the disk function in the $\q$ expansion.

For simplicity, we assume that $H^2_\cmp (X_{\bt})$ is non-empty.
Let $\psi: H^2_\cmp (X_{\bt}) \to H^2(X_{\bt})$ be the map sending
cohomology classes with compact support to ordinary cohomology classes.
One can decompose its image over a basis of $H^2 (X_{\bt})$,
so that $\psi = \sum_{a=1}^r \phi_a \psi^a$.
For a toric Kähler quotient $X_{\bt}$ with charge matrix $Q^a_i$,
the map $\psi$ can be represented by a matrix of integers
\begin{equation}
\label{eq:psi-map}
 \psi^a (\pd(D_i)) = - Q^a_i~, \quad i \in I_{\cmp}~.
\end{equation}
If there are no compact divisors then the
set $I_\cmp$ is empty and the $\psi$-map is identically zero.
\begin{proposition}
\label{conj:compact-divisor-K-th}
Let $M = \sum_{i\in I_{\cmp}} M^i \pd (D_i) \in H^2_\cmp (X_{\bt})$ with
$M^i \in \BZ$. Assume that $\bT + \psi(M)$ is in the same chamber as $\bT$.
Then the difference
\begin{equation}
 \cz^D (\bT,q;\q) - \prod_{i \in I_{\cmp}} q_i^{M^i} \cz^D (\bT+\psi(M),q;\q)
	\in \BZ [q_1,\ldots,q_N] \, [[\q]]
 \label{eq:shiftZ}
\end{equation}
is a formal power series in $\q$, with \emph{polynomial} coefficients
in the variables $q_i$.
\end{proposition}
\begin{proof}
The expression in \cref{eq:shiftZ} can be rewritten as
\begin{equation}
 \left( 1-\eu^{-\hbar \sum_{i\in I_\cmp} M^i\cD_i} \right) \cz^D (\bT,q;\q)
 = \left( 1-\prod_{i \in I_\cmp}\Delta_i^{M^i} \right) \cz^D (\bT,q;\q)~.
\end{equation}
We first consider the case when $M = \pd (D_i)$ for some $i \in I_{\cmp}$.
Using \cref{eq:disk-identity-delta,eq:ZD-BPScount}, we can write
\begin{equation}
\label{eq:proof-shift-ZD}
 (1-\Delta_i) \cz^D (\bT,q;\q) = \sum_{n^i=0}^\infty
	\frac {(\q q_i)^{n^i}} {(\q;\q)_{n^i}}
 \sum_{\Lambda_i(\bT,n^i)} \prod_{j\neq i} \frac {q_j^{n^j}} {(\q;\q)_{n^j}}~,
\end{equation}
where the set
\begin{equation}
 \Lambda_i (\bT,k) := \left\{ (n^1,\dots,n^N) \in \BZ^N_{\geq0}
 \middle| \, \sum_{j=1}^N Q^a_j n^j = T^a \text{ and } n^i=k \right\}
\end{equation}
is finite by the assumption of compactness of divisor $D_i$. By repeatedly
applying \cref{eq:p-len}, we see that a given power of $\q$ in
\cref{eq:proof-shift-ZD} only receives contributions from a finite number
of $\Lambda$'s.
This shows that $(1 - \Delta_i) \cz^D$ satisfies the thesis.
Next we consider the case when $M$ is an integer multiple of a
generator $x_i$, i.e.\ $M=M^i\pd(D_i)$ with $M^i\in\BZ$ (no sum over $i$
is implied here). We then have
\begin{equation}
\label{eq:DeltaMi}
 \left( 1-\Delta_i^{M^i} \right) \cz^D (\bT,q;\q)
 = \left( 1+\Delta_i+\dots+\Delta_i^{M^i-1} \right)
	(1-\Delta_i) \cz^D (\bT,q;\q)~.
\end{equation}
Since the r.h.s.\ is a regular operator acting on the regular expression
$(1-\Delta_i) \cz^D$, we can use the previous result to the deduce
that the l.h.s.\ is also regular for any $M^i > 1$. For $M^i<0$ we use that
$(1-\Delta^{-M^i})=-\Delta^{-M^i}(1-\Delta^{M^i})$.

Given any pair of compact divisors $D_i$ and $D_j$, with $i,j \in I_{\cmp}$,
we have
\begin{equation}
 \left( 1-\Delta_i^{M^i} \Delta_j^{M^j} \right)
 = \left( 1-\Delta_i^{M^i} \right) + \left( 1-\Delta_j^{M^j} \right)
 -  \left( 1-\Delta_i^{M^i} \right) \left( 1-\Delta_j^{M^j} \right)~.
\end{equation}
Applying our previous result to terms on the right, we deduce that
$\left( 1-\Delta_i^{M^i} \Delta_j^{M^j} \right) \cz^D$
satisfies the thesis and by induction we conclude that
$\left( 1-\prod_{i\in I_{\cmp}} \Delta_{i}^{M^{i}} \right) \cz^D$
satisfies it as well.
\end{proof}

By taking the 1d limit $\q \to 0$, we find an analogous compact divisor shift
equation
\begin{equation}
 (1-\Delta_i) \cz(\bT,q) = \sum_{\Lambda_i(\bT,0)}\prod_{j\neq i}q_j^{n^j}
 \in \BZ[q_1,\dots,q_N]
\end{equation}
for the $S^1$ partition function \cite{Nekrasov:2021ked}.
In this case the quantity on the r.h.s.\ is a polynomial
in $q_i$'s with integer coefficients, hence
$\lim_{q \to 1} (1-\Delta_i) \cz(\bT,q)$ is an integer.

If instead we reduce along the circle (cohomological limit $\hbar \to 0$),
we find that
\begin{equation}
\label{eq:shift-equation-cohomology}
 \cD_i \cf^D (\bt,\e;\lambda) \quad\text{is analytic at $\e=0$}
\end{equation}
if the divisor $D_i$ is compact.

A different 2d limit is the double scaling $\hbar \to 0$ and $M \to \infty$
with $m := \hbar M$ constant, in which case \cref{eq:shiftZ} becomes
the shift equation of ref.~\cite{Nekrasov:2021ked}, namely:
\begin{proposition}
\label{conj:compact-divisor-coho}
Let $m = \sum_{i\in I_{\cmp}} m^i \pd (D_i) \in H^2_\cmp (X_{\bt})$.
Assume that $\bt + \psi(m)$ is in the same chamber as $\bt$.
Then the difference
\begin{equation}
 \label{eq:shiftF}
 \cf^D (\bt,\e;\lambda)
 - \eu^{-\sum_{i \in I_{\cmp}}m^i \e_i}
 \cf^D (\bt+\psi(m),\e;\lambda)
\end{equation}
is regular in the non-equivariant limit $\e\to0$.
\end{proposition}

If $H^2_\cmp (X_\bt)$ is empty, then we look at any
set $S \subseteq \{1,\ldots,n\}$ of divisors
such that their intersection is compact;
the action of the corresponding product of operators
makes the disk function regular in the non-equivariant limit
\begin{equation}
 \bigcap_{i \in S} D_i \text{ compact } \Longrightarrow
 \prod_{i \in S} (1-\Delta_i) \cz^D (\bT,q;\q)
 \quad\text{is analytic at $q=1$}~.
\end{equation}
The proof of this statement is a straightforward generalization
of the argument in \cref{conj:compact-divisor-K-th}.
By reducing along the circle ($\hbar \to 0$), we find that
\begin{equation}
\label{eq:generalized-shift-D}
 \bigcap_{i \in S} D_i \text{ compact } \Longrightarrow
 \prod_{i \in S} \cD_i \cf^D (\bt,\e;\lambda)
 \quad\text{is analytic at $\e=0$}~.
\end{equation}
In this case the analog of the shift equation corresponds to
some higher-order difference equation. 
We consider examples without compact divisors in \cref{sec:examples2}.

%% file: quantum-coho.tex
\section{Quantum cohomology and quantum K-theory} \label{sec:quantum-coh}

\subsection{Equivariant Picard--Fuchs equations}

Let us fix a chamber $\chamber$, and work in cohomology for simplicity
(everything can be rephrased in K-theory terms).
We define the equivariant Picard--Fuchs operator
\begin{equation}
 \PF^\eq_\gamma :=
 \prod_{\{i|\sum_a\gamma_a Q^a_i>0\}}
 \left(\tfrac{\cD_i}{\lambda}\right)_{\sum_a\gamma_a Q^a_i}
 - \eu^{-\lambda \sum_{a} \gamma_a t^a}
 \prod_{\{i|\sum_a\gamma_a Q^a_i\leq0\}}
 \left(\tfrac{\cD_i}{\lambda}\right)_{-\sum_a\gamma_a Q^a_i}~.
\end{equation}
Then, for any $\gamma \in \chamber^\vee\cap\BZ^r$,
\cref{eq:disk-identity-D2,eq:disk-identity-D1} imply the relations
\begin{equation}
\label{eq:epf}
 \PF^\eq_\gamma \cf^D (\bt,\e;\lambda) = 0~.
\end{equation}
By the formal identification of differential operators $\cD_i$ and Chern roots
$x_i$, we can interpret \cref{eq:epf} as a differential operator representation
of the Batyrev or Quantum Stanley--Reisner ideal $I_\mathrm{QSR}$
defined by products
\begin{equation} \label{eq:QSR}
 \prod_{\{i|\sum_a\gamma_a Q^a_i>0\}}
 x_i^{\sum_a\gamma_a Q^a_i}
 - \prod_a z_a^{\gamma_a}
 \prod_{\{i|\sum_a\gamma_a Q^a_i\leq0\}}
 x_i^{-\sum_a\gamma_a Q^a_i}
.\end{equation}

We argue that by \cref{eq:epf} the disk function $\cf^D$ is a D-module for
the Quantum Cohomology ring of the toric quotient $X_\bt$
\begin{equation}
 QH^\bullet_{\sft}(X_\bt) :=
 \BC[\phi_1,\dots,\phi_r,\e_1,\dots,\e_N,z_1,\dots,z_r] / I_\mathrm{QSR}~.
\end{equation}
See refs.~\cite{Batyrev:9310004,Gonzalez:2012fd} for a discussion of
Batyrev description of quantum cohomology and
quantum deformations of the Kirwan map.

Differential equations of the type of \cref{eq:epf} encode a quantum
deformation of classical cohomology and are known as
\emph{equivariant Picard--Fuchs} (PF) equations. In fact, it follows that in the
classical limit $\lambda\to\infty$ (or large volume limit $\bt\to\infty$)
the quantum deformation vanishes (by the assumption on $\gamma$)
and the operators $\PF^\eq_\gamma$ provide a realization of the classical
cohomology relations as elements of the Stanley--Reisner ideal.

The usual non-equivariant PF operators are recovered when we send all $\e_i$
to zero,
\begin{equation}
 \PF_\gamma =
 \prod_{\{i|\sum_a\gamma_a Q^a_i>0\}}
 \big(-\sum_a\theta_a Q^a_i\big)_{\sum_a\gamma_a Q^a_i}
 - \eu^{-\lambda \sum_{a} \gamma_a t^a}
 \prod_{\{i|\sum_a\gamma_a Q^a_i\leq0\}}
 \big(-\sum_a\theta_a Q^a_i\big)_{-\sum_a\gamma_a Q^a_i}
\end{equation}
with $\theta_a:=\partial/\partial\log z_a
=-\tfrac1{\lambda}\partial/\partial t^a$.
Observe that while the PF operators themselves always have a well-defined
non-equivariant limit, this might not be the case for the disk function.
In fact, we have that for any non-compact manifold $X_\bt$, the disk function is
singular at $\e=0$, and therefore \cref{eq:epf} generically does not have a
non-equivariant limit.

The equivariant K-theoretic Picard--Fuchs operators are defined as
\begin{equation}
 \PF^{K_\eq}_\gamma :=
 \prod_{\{i|\sum_a\gamma_a Q^a_i>0\}}
 \left(\Delta_i;\q\right)_{\sum_a\gamma_a Q^a_i}
 - \q^{\sum_{a} \gamma_a T^a}
 \prod_{\{i|\sum_a\gamma_a Q^a_i\leq0\}}
 \left(\Delta_i;\q\right)_{-\sum_a\gamma_a Q^a_i}
\end{equation}
and they annihilate the K-theoretic disk function,
\begin{equation}
\label{eq:Kepf}
 \PF^{K_\eq}_\gamma \cz^D (\bT,q;\q) = 0~,
\end{equation}
thus providing a representation of quantum K-theory relations.\footnote
{The correspondence between 3d $\cN=2$ gauge theories and quantum
K-theory has been previously observed in ref.~\cite{Jockers:2018sfl},
where a dictionary to match the two sides was worked out. Here we extend the
discussion to the equivariant setting for arbitrary toric CYs.
Moreover, our results follow directly from the choice of integration
contour for the integral representation of the disk function that we postulated
in \cref{eq:cont-F-def-D2,eq:cont-F-def-D2S1}.
This choice is motivated by the symplectic geometry of the target and
extends naturally to any toric example. For simplicity, in our discussion
we omit any reference to the level structure of
quantum K-theory \cite{Ruan:2018tls}, in other words we assume level 0.}
The non-equivariant K-theoretic PF operators are obtained by using the formula
\begin{equation}
 \lim_{q\to1}\Delta_i = \q^{-\sum_a \theta_a Q_i^a}~.
\end{equation}

\subsection{The Givental \texorpdfstring{$\hat{I}$}{I}-operator}
\label{sec:Givental}

\begin{definition}
Inspired by work of Givental \cite{MR1408320,Givental:9701016}, we define the
equivariant $I$-function
\begin{equation}
 I_{X_{\bt}} :=
 \eu^{\sum_{a} \phi_a t^a} \sum_{\bd \in \dchamber}
 \eu^{-\lambda\sum_{a=1}^r d_a t^a}
 \prod_{i=1}^N \left(\frac{x_i}{\lambda}\right)_{-\sum_{a} d_a Q^a_i} ~,
\end{equation}
where $\dchamber := \chamber^\vee \cap \BZ^r$ is
the intersection of the lattice $\BZ^r$ with the dual of the chamber.\footnote
{This choice guarantees that the classical cohomology limit
$\lambda \to \infty$ is well-defined.}
\end{definition}
Our considerations in this section follow from the following fact.
\begin{proposition}
\label{th:main_identity}
There is an identity
\begin{equation}
 \oint_\qjk \prod_{a} \frac{\dif\phi_a}{2\pi\ii}
 \eu^{\sum_a \phi_a t^a}\prod_{i} \Ga\left(\frac{x_i}{\lambda}\right)
 = \oint_\jk \prod_{a} \frac{\dif\phi_a}{2\pi\ii}
 I_{X_{\bt}} \prod_{i}\Ga\left(\frac{x_i}{\lambda}\right)~.
\end{equation}
\end{proposition}
\begin{proof}
Let us discuss the identity one JK pole $p$ at a time.
On the l.h.s.\ we use the definition of the QJK contour to write
\begin{equation}
\begin{aligned}
 \mathrm{LHS} &= \sum_{k\in\BZ_{\geq0}^r}
 \oint_{\phi=\phi(p)-\lambda (Q_p^{-1})^{\mathrm{t}} k}
 \prod_{a} \frac{\dif\phi_a}{2\pi\ii}
 \eu^{\sum_a \phi_a t^a}\prod_{i} \Ga\left(\tfrac{\e_i+\sum_a\phi_a Q^a_i}
 {\lambda}\right) \\
 &= \sum_{d\in(Q_p^{-1})^{\mathrm{t}}\BZ_{\geq0}^r}
 \oint_{\phi=\phi(p)-\lambda d}
 \prod_{a} \frac{\dif\phi_a}{2\pi\ii}
 \eu^{\sum_a \phi_a t^a}\prod_{i} \Ga\left(\tfrac{\e_i+\sum_a\phi_a Q^a_i}
 {\lambda}\right)~,
\end{aligned}
\end{equation}
where we relabeled the sum in terms of
$d_a = \sum_b k_b (Q_p^{-1})^b_a$.
On the r.h.s.\ we use the definition of the $I$-function and
the change of variables $\td{\phi}_a=\phi_a-\lambda d_a$,
\begin{equation}
\begin{aligned}
 \mathrm{RHS} &=
 \sum_{\bd\in \dchamber}
 \oint_{\phi=\phi(p)} \prod_{a} \frac{\dif\phi_a}{2\pi\ii}
 \eu^{\sum_{a}(\phi_a-\lambda d_a)t^a}
 \prod_{i}\Ga\left(\tfrac{\e_i+\sum_a\phi_a Q^a_i}{\lambda}
 -\sum_{a}d_a Q^a_i\right) \\
 &= \sum_{\bd\in \dchamber}
 \oint_{\td{\phi}=\phi(p)-\lambda d} \prod_{a}\frac{\dif\td{\phi}_a}{2\pi\ii}
 \eu^{\sum_{a}\td{\phi}_a t^a}
 \prod_{i}\Ga\left(\tfrac{\e_i+\sum_a\td{\phi}_a Q^a_i}{\lambda}\right) ~.
\end{aligned}
\end{equation}
The difference between the two sides of the equation is in the range of the sum
over instanton charges $\bd$. At first glance one would like to show that
the two cones $(Q_p^{-1})^{\mathrm{t}}\BZ^r_+$ and $\dchamber$ coincide for
every fixed point $p$. On closer inspection, however, we realize that a weaker
condition is sufficient, namely that
\begin{equation}
\label{eq:cone-relation}
 (Q_p^{-1})^{\mathrm{t}}\BZ^r_+\subseteq \dchamber~.
\end{equation}
This is because if $\bd\notin(Q_p^{-1})^{\mathrm{t}}\BZ^r_+$
then some of the $k_a$ become negative and the corresponding residue integral
picks up a zero of one of the $\Ga$-functions instead of a pole.
We therefore need to prove \cref{eq:cone-relation} for any fixed point $p$.
The l.h.s.\ is the cone generated by the column vectors of the matrix
$(Q_p^{-1})^{\mathrm{t}}$. For brevity we indicate this as
$\Cone((Q_p^{-1})^{\mathrm{t}})$. The cone on the r.h.s.\ is by
definition the integer cone dual to the chamber, i.e.\ $\dchamber
=\chamber^\vee\cap\BZ^r$.
Hence we need to prove the inclusion
\begin{equation}
 \Cone((Q_p^{-1})^{\mathrm{t}}) \subseteq \chamber^\vee\cap\BZ^r~.
\end{equation}
We can now use the simple fact that
\begin{equation}
 \Cone((Q_p^{-1})^{\mathrm{t}}) = \Cone(Q_p)^\vee
\end{equation}
and the fact that inclusion of cones is reversed under duality,
to rewrite \cref{eq:cone-relation} as
\begin{equation}
 \chamber\cap\BZ^r \subseteq \Cone(Q_p)~.
\end{equation}
By definition this is true for any JK pole $p$ and so the content of the
proposition is true.
\end{proof}
The argument used in the proof indicates that JK poles
are the only ones that allow for the integral over the quantum contour to be
expressed via the $I$-function. This observation then leads to the conclusion
that JK poles, together with their towers of quantum corrections, are in 
one-to-one
correspondence with solutions of equivariant PF equations, and
that there is a basis of solution labeled by fixed points of the $\sft$-action.
\begin{definition}
By replacing $x_i$ with $\cD_i$ in the $I$-function, let us define the Givental 
operator
\begin{equation}
 \hat{I}_{X_{\bt}} := \sum_{\bd \in \dchamber} 
\eu^{-\lambda\sum_{a=1}^r d_a t^a}
 \prod_{i=1}^N \left(\frac{\cD_i}{\lambda}\right)_{-\sum_{a} d_a Q^a_i}~.
\end{equation}
\end{definition}
This definition together with \cref{th:main_identity} imply the following.
\begin{corollary}
The $I$-function and the $\hat{I}$-operator are related by the identity
\begin{equation}
 I_{X_{\bt}}=\hat{I}_{X_{\bt}}\cdot\eu^{\sum_{a} \phi_a t^a}~,
\end{equation}
therefore the disk function satisfies the relation
\begin{equation}
 \cf^D
 = \hat{I}_{X_{\bt}} \cdot \cf_{\Ga}~.
\end{equation}
\end{corollary}
\begin{remark} \label{fprmk}
Any solution to the classical cohomology equations can be written
as an integral over the classical JK contour
for some cohomology class $\al(\phi,\e)\in H^\bullet_{\sft}(X_\bt)$
\begin{equation}
 \cf_\al(\bt,\e) = \oint_\jk \prod_{a} \frac{\dif\phi_a}{2\pi\ii}
 \frac{\eu^{\sum_{a} \phi_a t^a}}{\prod_{i}x_i} \al(\phi,\e)
 = \int_{X_{\bt}} \eu^{\varpi_{\bt}-H_\e} \,\al~.
\end{equation}
The semi-classical partition function $\cf_{\Ga}$ corresponds to the choice
of $\al$ equal to the $\hat \Ga$-class of the manifold $X_\bt$.
Moreover, for every fixed point $p$ there exists a class $\pd(p)$ that evaluates
to $0$ on all fixed points but $p$.
Since these classes form a basis for the (localized) equivariant cohomology,
we can then write any classical solution as a linear combination
\begin{equation}
 \cf_\al(\bt,\e) = \sum_{p\in\fp} \al_p(\e) \cf_{\pd(p)}(\bt,\e)~,
 \quad
 \cf_{\pd(p)}(\bt,\e) = \eu^{-H_\e(p)}~,
\end{equation}
where $\al_p(\e)$ are the coefficients of $\al$ in the fixed-point basis.
\end{remark}
One can then use the operator $\hat{I}_{X_{\bt}}$ to construct arbitrary
solutions to equivariant PF equations out of any solution to the classical
cohomology relations.
\begin{proposition}
For a generic solution $\cf_\al(\bt,\e)$ of classical cohomology equations, the
disk function $\hat{I}_{X_{\bt}}\cdot\cf_\al(\bt,\e)$ is a formal solution to
the equivariant PF equations.
\end{proposition}
\begin{proof}
If $\cf_\al(\bt,\e)$ solves the classical cohomology equations, then it can be 
written as a linear combination of integrals over classical JK poles.
By \cref{th:main_identity}, the function
\begin{equation}
 \hat{I}_{X_\bt} \cdot \cf_\al(\bt,\e) =
 \sum_{p\in\fp} \al_p(\e)\, \hat I_{X_\bt} \cdot \cf_{\pd(p)}(\bt,\e)
\end{equation}
can be written as a linear combination of integrals,
each of which satisfies equivariant PF equations in \cref{eq:epf}.
(In this sense, we call $\hat I_{X_\bt} \cdot \cf_\al$ an equivariant period.)
\end{proof}
In the K-theoretic case we define
\begin{equation}
 \hat{I}^K_{X_{\bt}} := \sum_{\bd \in \dchamber} 
 \q^{\sum_{a=1}^r d_a T^a}
 \prod_{i=1}^N \left(\Delta_i;\q\right)_{-\sum_{a} d_a Q^a_i}~,
 \quad
 I^K_{X_{\bt}} := \hat{I}^K_{X_{\bt}} \cdot\prod_a w_a^{-T^a}
\end{equation}
and we have the identity
\begin{equation}
 \oint_\qjk \prod_{a} \frac{\dif w_a}{2\pi\ii w_a}
 \prod_a w_a^{T^a} \prod_{i} \frac1{\left(L_i;\q\right)_\infty}
 = \oint_\jk \prod_{a} \frac{\dif w_a}{2\pi\ii w_a}
 I^K_{X_{\bt}} \prod_{i} \frac1{\left(L_i;\q\right)_\infty}~.
\end{equation}
Similarly to the cohomological case, we can generate solutions to the PF
equations by applying the $\hat{I}^K$-operator to a classical K-theory solution,
written as a linear combination of fixed point solutions.

\subsection{Non-equivariant limit, singularities and instantons}

The non-equivariant limit is defined by sending all $\sft$-equivariant
parameters $\e_i$ to zero. In this limit, the equivariant (quantum) cohomology
of $X_\bt$ reduces to ordinary (quantum) cohomology and the operators $\cD_i$
simplify to linear combinations of derivatives
\begin{equation}
 \lim_{\e\to0}\frac{\cD_i}{\lambda} = \frac1{\lambda}
 \sum_a Q_i^a\frac{\partial}{\partial t^a}
 = -\sum_a Q_i^a \theta_a~,
\end{equation}
which act as operators inserting ordinary cohomology classes
$\sum_a\phi_a Q_i^a\in H^2(X_\bt)$.
Picard--Fuchs operators $\PF^\eq_\gamma$ are analytic in the $\e_i$'s, hence
they also degenerate in this limit to the non-equivariant PF operators
$\PF_\gamma$ and similarly one can set all $\e_i$'s to zero in the
$\hat{I}$-operator.
However, the function $\cf_\al(\bt,\e)$ might have a singular behavior near
$\e=0$, and in that case the disk function $\hat{I}_{X_{\bt}}\cdot
\cf_\al(\bt,\e)$ is not analytic at $\e=0$.
This follows from the observation that the degree-zero term in the
instanton expansion of $\hat{I}_{X_{\bt}}$ is the identity operator.
Corrections at higher instanton degree might or might not cure the singularity
in $\cf_\al(\bt,\e)$, according to the details of the geometry of $X_{\bt}$.
The main result of this section is \cref{eq:singular-instantons}, which
establishes a criterion to determine whether instanton contributions to the
disk function are singular or not in the non-equivariant limit.
In our case, the semi-classical part $\cf_{\Ga}$ is indeed singular
in the non-equivariant limit for non-compact manifolds $X_{\bt}$,
as $\cf_{\Ga}$ is a deformation of the volume.
In the compact case this function is regular and so also $\cf^D$ is
regular, since instanton corrections cannot introduce singular behavior.

In order to study the behavior of the instanton corrections we introduce
\textit{instanton operators}
\begin{equation}
 \sfp_{\bd} := \eu^{-\lambda \sum_a d_a t^a}
 \prod_{i=1}^N \left(\frac{\cD_i}{\lambda}\right)_{-\sum_a d_a Q^a_i}
 \quad \text{for} \quad \bd \in \Lambda~.
\end{equation}
From the definition of the $\hat{I}$-operator it follows that we can write
\begin{equation}
 \hat{I}_{X_{\bt}} = \sum_{d\in\dchamber} \sfp_{\bd}~.
\end{equation}
\begin{proposition}
The instanton operators $\sfp_{\bd}$ form an abelian monoid isomorphic
to $\Lambda$.
\end{proposition}
\begin{proof}
The composition of instanton operators is commutative and gives:
\begin{equation}
\begin{aligned}
 \sfp_{\bd} \sfp_{\bd'} &=
 \eu^{-\lambda \sum_a d_a t^a}
 \prod_{i} \left(\tfrac{\cD_i}{\lambda}\right)_{-\sum_a d_a Q^a_i}
 \eu^{-\lambda \sum_a d'_a t^a}
 \prod_{i} \left(\tfrac{\cD_i}{\lambda}\right)_{-\sum_a d'_a Q^a_i} \\
 &=
 \eu^{-\lambda \sum_a(d_a+d'_a)t^a}
 \prod_{i} \left(\tfrac{\cD_i}{\lambda}-\sum_a d'_a Q^a_i\right)_{-\sum_a d_a 
Q^a_i}
 \left(\tfrac{\cD_i}{\lambda}\right)_{-\sum_a d'_a Q^a_i} \\
 &=
 \eu^{-\lambda \sum_a(d_a+d'_a)t^a}
 \prod_{i}
 \left(\tfrac{\cD_i}{\lambda}\right)_{-\sum_a(d_a+d'_a)Q^a_i} \\
 &= \sfp_{\bd+\bd'}
\end{aligned}
\end{equation}
for $\bd,\bd' \in \Lambda$.
This completes the proof.
\end{proof}

We can then discuss the behavior of the instanton corrections in the limit
$\e\to0$ by making use of the fact that the instanton operators are proportional
to products of divisor operators $\cD_i$ and when such products correspond
to compact intersections their action makes the integral regular as $\e \to 0$.
\begin{proposition}
\label{eq:singular-instantons}
For any fixed instanton charge $\bd\in\Lambda$, if the intersection
\begin{equation}
 \bigcap_{\{i|\sum_a d_a Q_i^a < 0 \}} D_i
\end{equation}
is compact in $X_{\bt}$, then the instanton corrections proportional to
$\eu^{-\lambda\sum_a d_a t^a}\equiv z^{\bd}$ are analytic at $\e=0$.
Conversely, if the intersection of all divisors $D_i$ with $\sum_a d_a Q_i^a<0$
is non-compact, then the instantons of degree $\bd$ are singular in the
$\e\to0$ limit. 
\end{proposition}
\begin{proof}
Using the definition of the Pochhammer symbol in \cref{eq:pochhammer} we can 
write
\begin{multline}
 \sfp_{\bd} = \eu^{-\lambda \sum_a d_a t^a}
 \left[ \frac
 {\prod_{\{i|\sum_a d_a Q_i^a<0\}}
 \left(\frac{\cD_i}{\lambda}\right)_{-\sum_a d_a Q^a_i}}
 {\prod_{\{i|\sum_a d_a Q_i^a>0\}}(-1)^{\sum_a d_a Q^a_i}
 \left(1-\frac{\cD_i}{\lambda}\right)_{\sum_a d_a Q^a_i}}
 \right] \\
 = \eu^{-\lambda \sum_a d_a t^a}
 \left[ \frac
 {\prod_{\{i|\sum_a d_a Q_i^a<0\}}
 \left(1+\frac{\cD_i}{\lambda}\right)_{-\sum_a d_a Q^a_i-1}}
 {\prod_{\{i|\sum_a d_a Q_i^a>0\}}(-1)^{\sum_a d_a Q^a_i}
 \left(1-\frac{\cD_i}{\lambda}\right)_{\sum_a d_a Q^a_i}}
 \right]
 \prod_{\{i|\sum_a d_a Q_i^a<0\}}
 \frac{\cD_i}{\lambda}~.
\end{multline}
Therefore, if $\bigcap_{\{i|\sum_a d_a Q_i^a < 0 \}} D_i$ is compact in $X_\bt$,
by the shift \cref{eq:generalized-shift-D} the function
$\sfp_{\bd}\cdot\cf_{\Ga}$ is regular in the non-equivariant limit.
All singularities of the semi-classical integral are
cured by the insertion of the compact class $\prod_{\{i|\sum_a d_aQ_i^a<0\}}
x_i$. If this class is non-compact, then the integral is still singular at
$\e=0$, which implies that this instanton is singular.
\end{proof}

The K-theoretic instanton operators are defined as
\begin{equation}
 \sfp^K_{\bd} := \q^{\sum_a d_a T^a}
 \prod_{i=1}^N \left(\Delta_i;\q\right)_{-\sum_a d_a Q^a_i}
\end{equation}
so that
\begin{equation}
 \hat{I}^K_{X_\bt} = \sum_{\bd\in\Lambda} \sfp^K_{\bd}
\end{equation}
and an analogous statement to \cref{eq:singular-instantons} holds.
One can check that $\lim_{\hbar \to 0} \hat{I}^K_{X_\bt} = \hat{I}_{X_\bt}$.

%% file: regularization.tex
\section{Regularization}
\label{sec:regularization}

In the previous section we observed that for non-compact CY manifolds both the
classical part of the disk partition functions and the instanton corrections
can have singular behavior in the non-equivariant limit.
This implies that some PF solutions, such as the disk function itself, do not
admit a limit and therefore cannot be used to extract information about
non-equivariant GW theory and other enumerative geometric invariants.
In this section we argue that one can come up with some prescription to
regularize the singular PF solutions using the shift equations in
\cref{sec:shift}. We argue that there is \emph{no canonical way} to
split the function $\cf^D$ into a regular and a singular parts. However, using
compact divisor operators $\cD_i$, with $i \in I_\cmp$, we can construct
a family of functions that are both regular and in a certain sense contain
the same amount of information as the original function.

We define a ``regularization'' of $\cf^D$ to be any function $\cf^D_\reg$
such that
\begin{equation}
\label{eq:regularizeFD}
 \cD_i\cf^D_\reg = \cD_i\cf^D~,
\quad \forall i \in I_\cmp~.
\end{equation}
It clearly follows from this definition that
$\cf^D_\reg$ differs from $\cf^D$ by some singular function
that sits in the common kernel of all compact divisor operators, and
$\cf^D_\reg$ is no longer a solution of PF equations, but it does solve
an extended set of PDEs related to the original PF equations in a specific way,
such that solutions to this system contain the original PF solutions as a
subset. The main feature of this regularization procedure is that generically
\cref{eq:regularizeFD} only defines $\cf^D_\reg$ up to arbitrary elements of the
kernel of the compact divisor operators and therefore contains an intrinsic
ambiguity corresponding to the fact that the splitting between regular and
parts of the disk function is not canonically defined.

The extended system of PDEs are sometimes known as ``modified
Picard--Fuchs equations''.
Some specific cases of extended systems of quantum equations in the context of
local mirror symmetry have previously appeared in ref.~\cite{Forbes:2005xt}
for manifolds with no compact divisors.
Here we give a systematic treatment of these equations in the toric CY case
while also working in the fully equivariant setting.

We start by defining the sub-lattice of singular instantonic contributions as
\begin{equation}
 \Lambda_\sing := \left\{\bd\in\Lambda\,\middle|\,
 \sfp_{\bd}\cdot\cf_{\Ga}\text{ is singular at }\e=0\right\}
 \subseteq\Lambda~.
\end{equation}
If the manifold $X_\bt$ is compact, then $\cf^D$ is regular and the singular
sub-lattice is empty. For non-compact $X_\bt$, the disk function is singular and
therefore $\Lambda_\sing$ contains at least the origin, i.e.\ the semi-classical
contribution. Higher degree instanton contributions could also be singular as
discussed in the previous section. Then $\Lambda_\sing$ is a sub-cone of
$\Lambda$. Similarly, let
\begin{equation}
 \cf^D_\sing := \sum_{\bd\in\Lambda_\sing}
 \sfp_{\bd}\cdot\cf_{\Ga}
\end{equation}
so that $\cf^D-\cf^D_\sing$ is regular by construction.

We give a prescription to regularize $\cf^D$ for a non-compact manifold
$X_\bt$ by making use of the shift equation. For simplicity, we consider
the case when $X_\bt$ admits at least one compact divisor.
The strategy we adopt is the following: we remove the singular part of the
disk function and add it back again after applying to it the shift operator in
\cref{eq:shiftF}. By construction, the resulting function is
regular, but we also show that it differs from the original disk function
by a term that is annihilated by all compact divisor operators.

First observe that
\begin{equation}
 \left(\cf^D-\cf^D_\sing\right)
 + \left(1-\eu^{-\sum_{i\in I_\cmp}m^i\cD_i}\right)\cf^D_\sing
\end{equation}
is a regular function in the non-equivariant limit. To define a regularized
disk function we give a prescription to fix the values of $m$'s:
we look for a matrix $R^i_a$ such that
\begin{equation}
\label{eq:left-inverse}
 \sum_{a=1}^r R^j_a Q_i^a = \delta^j_i~,
 \quad \text{for } i,j \in I_\cmp~,
\end{equation}
i.e.\ a left-inverse of (minus) the $\psi$-map in \cref{eq:psi-map}.
If it exists (it may not be unique), we let
\begin{equation}
 m^i = \sum_{a=1}^r R^i_a t^a
\end{equation}
and we define the regularized disk function
\begin{equation}
\label{eq:Freg}
 \cf^D_\reg(\bt,\e;\lambda) := \cf^D(\bt,\e;\lambda)
 - \eu^{-\sum_a\sum_{i\in I_\cmp}\e_i R^i_a t^a}
 \cf^D_\sing\left(\bt+\psi(R(\bt)),\e;\lambda\right)~.
\end{equation}
For this choice of $m$'s and for every $i \in I_\cmp$, we have
\begin{multline} \label{eq:Freg-compact-div}
 \cD_i\left(\cf^D - \cf^D_\reg\right) =
 (\e_i-\sum_a\sum_{j\in I_\cmp}\e_j R^j_a Q^a_i)
 \eu^{-\sum_a\sum_{j\in I_\cmp}\e_j R^j_a t^a}
 \cf^D_\sing\left(\bt+\psi(R(\bt)),\e;\lambda\right) \\
 + \eu^{-\sum_a\sum_{j\in I_\cmp}\e_j R^j_a t^a}
 \sum_{a,b} Q^a_i (\delta^b_a-\sum_{j\in I_\cmp}Q^b_j R^j_a)
 \frac{\partial\cf^D_\sing}{\partial t^b}
 \left(\bt+\psi(R(\bt)),\e;\lambda\right) =0~,
\end{multline}
where the last equality follows from the property in \cref{eq:left-inverse}.
\begin{proposition} \label{prop:modified-PF}
For every PF operator $\PF^\eq_\gamma$ with $\gamma \in \Lambda$
and every compact divisor $D_i$, $i \in I_\cmp$, we have the \emph{modified
Picard--Fuchs equations}
\begin{equation}
\label{eq:modifiedPF}
\begin{dcases*}
 \cD_i \PF^\eq_\gamma \cdot \cf^D_\reg =0 &
 if $\sum_a\gamma_a Q^a_i \leq 0$, \\
 (\cD_i+\lambda\sum_a\gamma_a Q^a_i) \PF^\eq_\gamma\cdot\cf^D_\reg=0 &
 if $\sum_a\gamma_a Q^a_i>0$.
\end{dcases*}
\end{equation}
\end{proposition}
\begin{proof}
Since $\cf^D$ is a solution of ordinary PF equations, it is
also a solution of modified PF equations.
We compute the commutation relation between
$\cD_i$ and the PF operator. There are two cases:
if $\sum_a\gamma_a Q^a_i \leq 0$, then
\begin{equation}
 \cD_i \PF^\eq_\gamma
 = \left[\prod_{\{j|\sum_a\gamma_a Q^a_j>0\}}
 \left(\tfrac{\cD_j}{\lambda}\right)_{\sum_a\gamma_a Q^a_j}
 -\eu^{-\lambda \sum_{a} \gamma_a t^a}
 \prod_{\{j|\sum_a\gamma_a Q^a_j\leq0\}}
 \left(\tfrac{\cD_j}{\lambda}+\delta_{i,j}\right)_{-\sum_a\gamma_a Q^a_j}
 \right]\cD_i~.
\end{equation}
If instead $\sum_a\gamma_a Q^a_i>0$, then
\begin{multline}
 (\cD_i+\lambda\sum_a\gamma_a Q^a_i) \PF^\eq_\gamma = \\
 = \left[\prod_{\{j|\sum_a\gamma_a Q^a_j>0\}}
 \left(\tfrac{\cD_j}{\lambda}+\delta_{i,j}\right)_{\sum_a\gamma_a Q^a_j}
 -\eu^{-\lambda \sum_{a} \gamma_a t^a}
 \prod_{\{j|\sum_a\gamma_a Q^a_j\leq0\}}
 \left(\tfrac{\cD_j}{\lambda}\right)_{-\sum_a\gamma_a Q^a_j}
 \right]\cD_i~.
\end{multline}
Applying this to $\cf^D_\reg$ together with \cref{eq:Freg-compact-div},
we obtain the claim.
\end{proof}
\begin{remark}
In the second case the semi-classical limit of the
modified PF equations is the same as that of the ordinary PF equations,
while in the first case the semi-classical limit
gives different classical relations. In particular, the order of the PDEs is
increased by one. This implies that in the non-equivariant limit there are
logarithmic solutions of degree higher than those of the ordinary
non-equivariant PF equations.
\end{remark}

We argue that $\cf^D_\reg$ is obtained as a sum of two solutions of the
modified PF equations in such a way that the singularities in the two
cancel out and give a regular solution. While this is somewhat nice, we remark
that $\cf^D_\reg$ is not itself a solution of the ordinary PF equations.
This follows from the fact that its semi-classical part does not
satisfy the classical cohomology relations. However, the
Givental operator associated to the modified PF equations is the same as the
operator associated to the ordinary PF equations.

%% file: enumerative.tex
\section{Enumerative geometry} \label{sec:Gromov-Witten}

We elucidate the relation of our disk partition functions
to Gromov-Witten theory and related computations in the enumerative geometry
of the target $X_\bt$. The discussion focuses mostly on the cohomological
version of the story, as the K-theoretic version is less understood
\cite{Givental:2016afst6, Givental:many, Lee:2001qkt,
Jockers:2018sfl, Jockers:2019wjh, Garoufalidis:2021cha, Chou:2022qki}.
While a connection to genus-zero GW theory is
expected on general grounds, the details of how to match the disk function
$\cf^D$ with counts of stable maps to $X_\bt$ from first principles are still
to be worked out. Nevertheless, we are able to make some speculations deriving
from explicit analysis of the disk function in various examples.

First, we review the connection to enumerative
geometry for compact CY manifolds.
Next we discuss the generalization to non-compact CY manifolds with focus on
toric quotients, where the need for equivariance becomes manifest.

\subsection{Review of the compact case}

In this subsection, we consider compact CY targets $X$ to which Givental's 
formalism
can be applied, e.g.\ compact toric complete intersections 
\cite{Givental:9701016}.
The solutions to non-equivariant PF equations are obtained by acting with the
corresponding $\hat{I}$-operator on solutions to non-equivariant classical
cohomology relations.
In the compact case, these classical solutions are polynomials in the
Kähler moduli $t^a$ and there is a one-to-one map between solutions and compact
cycles in the homology lattice of $X$.
In particular, there are always the solution corresponding to the point
$\pt\in H_0(X)$ and the fundamental class $[X]$. More generally, the
mapping between solutions and cycles goes as follows. Let $C$ be a (compact)
cycle, then there is a classical solution
$\Pi^{\scl}(C)$ defined as
\begin{equation}
\label{eq:fundamental-scl-sol}
 \Pi^{\scl}(C) := (-\lambda)^{\dim_\BC C} \int_{X} \eu^{\varpi_{\bt}} \pd(C)
 = (-\lambda)^{\dim_\BC C}\int_{C} \eu^{\varpi_{\bt}}~.
\end{equation}
This solution is a polynomial in $t^a$ of degree
equal to the complex dimension of the cycle $C$. The coefficients of the
polynomial encode information about the intersection numbers of $C$ with all
other cycles.

From this definition it follows that
\begin{equation}
 \Pi^{\scl}(\pt) = 1~,
 \quad
 \Pi^{\scl}(C^a) = - \lambda t^a = \log z_a~,
 \quad
 \dots~,
 \quad
 \Pi^{\scl}(X) = (-\lambda)^\cd p_\cd (\bt)~,
\end{equation}
where we used $\int_{C^a} \varpi_\bt = t^a$ for $C^a$ a basis of $H_2(X)$
and $p_\cd (\bt)$ is the intersection polynomial of $X$.
The full non-equivariant PF solution is obtained by acting with
Givental's operator,
\begin{equation}
 \Pi(C) := \hat{I}_X \cdot \Pi^\scl(C)~.
 \end{equation}
Since $\Pi^\scl(C)$ is polynomial in $t^a$, one can compute the full solution
by expanding the $\hat{I}$-operator as a power series in the derivatives
$\tfrac{\partial}{\partial t^a}$ up to order equal to the degree of the
classical solution. All contributions of higher order annihilate the
polynomial and do not contribute to the solution. This gives an
efficient algorithm to construct PF solutions, completely equivalent to the
standard Frobenius method.

Let us consider the familiar example of the quintic $X_5$. The PF operator is
\begin{equation}
 \PF = \theta^4 - 5 z \left(1+5\theta\right)_4
 \quad\text{with}\quad
 \theta = -\tfrac1{\lambda}\tfrac{\partial}{\partial t}
 = z\tfrac{\partial}{\partial z}
\end{equation}
from which we can construct the Givental operator
\begin{equation}
 \hat{I}_X = 
 \sum_{d=0}^\infty z^d
 \frac{\left(1+5\theta\right)_{5d}}
 {\left(1+\theta\right)_d^5}~,
\end{equation}
which can be expanded as
\begin{multline}
 \hat{I}_X
 = G^{(0)} + G^{(1)} \theta
 + \left(\tfrac12 G^{(2)} - \tfrac{5\pi^2}{3}G^{(0)}\right)
 \theta^2
 + \left(\tfrac16 G^{(3)} + 40\zeta(3)G^{(0)} - \tfrac{5\pi^2}{3}G^{(1)}\right)
 \theta^3 \\
 + \left(\tfrac1{24}G^{(4)} - \tfrac{\pi^4}{3}G^{(1)} + 40\zeta(3)G^{(1)} - 
\tfrac{5\pi^2}{6}G^{(2)}\right)
 \theta^4 + \dots
\end{multline}
with
\begin{equation}
 G^{(i)} := \sum_{d=0}^\infty z^d \left(\tfrac{\partial}{\partial d}\right)^i
 \frac{\Ga(5d+1)}{\Ga(d+1)^5}~.
\end{equation}
Classical solutions are polynomials of degree not higher than 3, which are
annihilated by $\theta^4$. The homology lattice has dimension 4, hence we have
4 solutions to the PF equation, which are usually referred to as periods,
\begin{equation}
\begin{aligned}
 \Pi(\pt) &= G^{(0)} ~, \\
 \Pi(C^1) &= G^{(0)} \left[\log z + \tfrac{G^{(1)}}{G^{(0)}}\right]
 = G^{(0)} \log\td{z} ~, \\
 \Pi(C_1) &= G^{(0)} \left[\tfrac52 \log\td{z}^2
 + 5\left( \tfrac{G^{(2)}}{2G^{(0)}}
 - \tfrac{(G^{(1)})^2}{2(G^{(0)})^2} - \tfrac{5\pi^2}{3} \right)\right] ~, \\
 \Pi(X_5) &= G^{(0)} \big[\tfrac56 \log\td{z}^3
 + 5\left( \tfrac{G^{(2)}}{2G^{(0)}}
 - \tfrac{(G^{(1)})^2}{2(G^{(0)})^2}
 - \tfrac{5\pi^2}{3} \right) \log\td{z} ~,\\
 &\qquad + 5\left( 40\zeta(3) + \tfrac{G^{(3)}}{6G^{(0)}}
 - \tfrac{G^{(1)} G^{(2)}}{2(G^{(0)})^2}
 + \tfrac{(G^{(1)})^3}{3(G^{(0)})^3}
 \right) \big] ~,
\end{aligned}
\end{equation}
where $C^1$ is the generator of $H_2(X_5)$ and $C_1$ is the generator
of $H_4(X_5)$.

One can introduce flat coordinates $\td{z}_a := z_a\,
\eu^{I^a_1(z)/I_0(z)}$ defined so that
$\Pi(C^a)/\Pi(\pt)=\log\td{z}_a$,
where $I_0,I^a_1$ are the coefficients of the Givental
operator in the series expansion in $\theta_a$, i.e.\
\begin{equation}
 \hat{I}_X
 = \sum_{n=0}^\infty \sum_{a_1,\dots,a_n} I_n^{a_1,\dots,a_n}
 \theta_{a_1}\cdots\theta_{a_n} ~.
\end{equation}
Observe that for general CYs the zeroth-order term $I_0$ can be non-trivial,
but for all toric CYs this function is identically 1.
The change of coordinates $\td{z}(z)$ is known as \emph{mirror map}.

Mirror symmetry predicts that solutions to the PF equations for a compact
CY manifold encode information about its genus-zero Gromov--Witten invariants
$N^0_{\bd}$. More specifically, one can read the GW potential $\Phi^0$
from instanton corrections to the classical solutions 
\begin{equation}
 \Phi^0(\td{z}) = (-\lambda)^\cd p_\cd (\td{\bt})
 + \Phi_\inst^0(\td{z})~,
 \quad
 \Phi_\inst^0(\td{z}) = \sum_{\bd\neq 0} N^0_{\bd} \td{z}^{\bd}~,
 \quad
 \td{z}^{\bd} = \prod_{a} \td{z}_a^{d_a}~,
\end{equation}
where $\bd=(d_1,\dots,d_r)$ is a non-zero effective class in $H_2(X,\BZ)$ that
labels the degree of a non-constant map from a genus-zero surface to $X$.
The classical part of the potential is a generating function of classical
intersection numbers
\begin{equation}
 \kappa_{a_1,\dots,a_\cd}
 = \frac{\partial^\cd}
 {\partial t^{a_1} \cdots \partial t^{a_\cd}} p_\cd (\bt)~.
\end{equation}
It is then conjectured that the potential $\td{\Phi}^0(\td{z})$ can be
re-expanded over a basis of PolyLogs with integer
coefficients defining the Gopakumar--Vafa (GV) invariants $n^0_{\bd}$ that
enumerate rational embedded curves of class ${\bd}$ and genus zero.
In the following we drop the label for the genus since we are only
considering genus-zero invariants.

As all CY twofolds are Hyperkähler, their GW invariants
are trivial, so PF solutions in complex dimension two only encode classical
information (after mirror map)
\begin{equation}
\label{eq:CY2}
 \Pi(\pt) = I_0~,
 \quad
 \frac{\Pi(C^a)}{\Pi(\pt)} = \log\td{z}_a~,
 \quad
 \frac{\Pi(X)}{\Pi(\pt)}
 = \tfrac12 \sum_{a,b} \kappa_{ab} \log\td{z}_a \log\td{z}_b~.
\end{equation}

The case of CY threefolds is the most studied one.
The GV conjecture can be stated as
\begin{equation}
 \Phi_\inst^0(\td{z}) = \sum_{\bd\neq 0} n_{\bd} \Li_3(\td{z}^{\bd}) 
\end{equation}
and the GV invariants can be obtained via the Möbius inversion formula
\begin{equation}
 n_{\bd} = \sum_{k|\bd} N_{\bd/k} \frac{\mu(k)}{k^3}~,
\end{equation}
where $\mu(k)$ is the Möbius function.

The solutions to the PF equations are conjectured to be
\begin{equation}
\label{eq:CY3-t3}
\begin{aligned}
 \frac{\Pi(X)}{\Pi(\pt)}
 &= \sum_a \td{t}^a \frac{\partial\Phi^0}{\partial\td{t}^a} -2\Phi^0 \\
 &= \tfrac16 \sum_{a,b,c} \kappa_{abc} \log\td{z}_a \log\td{z}_b \log\td{z}_c
 + \sum_{\bd\neq0} n_{\bd} \log(\td{z}^{\bd}) \Li_2(\td{z}^{\bd})
 -2 \sum_{\bd\neq0} n_{\bd} \Li_3(\td{z}^{\bd})~,
\end{aligned}
\end{equation}
\begin{equation}
\label{eq:CY3-t2}
 \frac{\Pi(C_a)}{\Pi(\pt)}
 = -\frac{1}{\lambda}\frac{\partial\Phi^0}{\partial\td{t}^a}
 = \tfrac12 \sum_{b,c} \kappa_{abc} \log\td{z}_b \log\td{z}_c
 + \sum_{\bd\neq0} n_{\bd} d_a \Li_2(\td{z}^{\bd})~,
\end{equation}
\begin{equation}
\label{eq:CY3-t1}
 \frac{\Pi(C^a)}{\Pi(\pt)} = -\lambda\td{t}^a = \log\td{z}_a
\end{equation}
with $C^a\in H_2(X)$ and $C_a\in H_4(X)$ such that $C^a\cap C_b=\delta^a_b$.

In the case of a CY fourfold it is conjectured that
\begin{equation}
\label{eq:CY4-t2}
 \frac{\Pi(C_{ab})}{\Pi(\pt)}
 = \tfrac12 \sum_{c,d} \kappa_{abcd} \log\td{z}_c \log\td{z}_d
 + \sum_{\bd\neq0}n_{\bd}(C_{ab})\Li_2(\td{z}^{\bd})~,
\end{equation}
where $C_{ab}\in H_4(X)$ and
\begin{equation}
 \sum_{\bd\neq 0} N_{\bd} \td{z}^{\bd}
 = \sum_{\bd\neq 0} n_{\bd} \Li_2(\td{z}^{\bd}),
 \quad
 n_{\bd} = \sum_{k|\bd} N_{\bd/k} \frac{\mu(k)}{k^2}~.
\end{equation}

Solutions with higher order classical behavior have more complicated expansions
in GV invariants that we do not reproduce here.
See refs.~\cite{Klemm:2007in, Honma:2013hma} for explicit formulas.

For CYs of higher dimension such formulas are not known and we do not
consider such examples in this section (even though solutions to PF equations
exist in any dimension).

Let us go back to the example of the quintic $X_5$. Matching the solutions we
found to the conjectural formulas for CY$_3$, we obtain the identities
\begin{equation}
 5\left( \tfrac{G^{(2)}}{2G^{(0)}}
 - \tfrac{(G^{(1)})^2}{2(G^{(0)})^2} - \tfrac{5\pi^2}{3} \right)
 = \sum_{d=1}^\infty n_d d \Li_2(\td{z}^d)
\end{equation}
and
\begin{equation}
 5\left( 40\zeta(3) + \tfrac{G^{(3)}}{6G^{(0)}}
 - \tfrac{G^{(1)} G^{(2)}}{2(G^{(0)})^2}
 + \tfrac{(G^{(1)})^3}{3(G^{(0)})^3}
 \right)
 = -2\sum_{d=1}^\infty n_d \Li_3(\td{z}^d)~,
\end{equation}
which give the well-known GV invariants of $X_5$.

\subsection{Non-compact case}

In the non-compact case the discussion is more involved,
as the volume is only defined equivariantly and it is a
divergent quantity in the non-equivariant limit. This is the case
relevant to our story, since all toric CY quotients are non-compact.
In the following, $X_\bt$ is a toric Kähler quotient with vanishing first
Chern class as described in \cref{sec:setup}.

We consider the fully equivariant PF operators $\PF^\eq$.
The solution is obtained by acting with the $\hat{I}$-operator on a basis
of classical solutions to the equivariant cohomology relations.
These solutions are naturally labeled by fixed points of the
torus action, i.e.\ basis elements of the localized equivariant cohomology ring.
By the localization formula \cref{eq:fpsum}, we can write $\cf(\bt,\e)$ as a sum
over this basis. Generically, to each fixed point $p\in\fp$ we can
associate the classical solution
\begin{equation}
 \Pi^\scl(p,\e) := \int_{X_\bt} \eu^{\varpi_\bt-H_\e} \pd(p)
 = \eu^{-H_\e(p)}
\end{equation}
with $H_\e(p)$ as in \cref{eq:local-hamiltonian} and $\pd(p)\in
H^{2\cd}_\sft(X_\bt)$ defined as the pushforward of $1\in H^0_\sft(p)$ along
the inclusion of the fixed-point $p\hookrightarrow X_\bt$. When comparing with
the non-equivariant case, we immediately notice that each of these solutions
goes to one in the $\e\to0$ limit. A better choice of basis to perform the 
comparison is
obtained by performing the equivariant upgrade of \cref{eq:fundamental-scl-sol}.
We then define for each cycle $C$ the equivariant solution
\begin{equation}
 \Pi^\scl(C,\e) := (-\lambda)^{\dim_\BC C}
 \int_{X_\bt} \eu^{\varpi_\bt-H_\e} \pd(C)~,
\end{equation}
which expands naturally over the basis of $\Pi^\scl(p,\e)$. These are classical
solutions that give rise to full quantum solutions when we act on them with the
equivariant Givental operator
\begin{equation}
 \Pi(C,\e) := \hat{I}_{X_\bt} \int_{X_\bt} \eu^{\varpi_\bt-H_\e} \pd(C)
 \quad\Rightarrow\quad
 \PF^\eq_\gamma \Pi(C,\e) =0~.
\end{equation}
By analogy with the compact case, we call the functions $\Pi(C,\e)$
\emph{equivariant periods}, since they solve equivariant PF equations.
When $C$ is compact, the integral in $\Pi^\scl(C,\e)$ restricts to
an integral over a compact space, therefore it defines an analytic function in
the $\e_i$'s and its non-equivariant limit is a finite quantity. As the
$\hat{I}$-operator cannot introduce singularities, the same is true for the full
PF solution. Then we have
\begin{equation}
 \lim_{\e\to0} \Pi(C,\e) = \Pi(C)
 \quad
 \text{for $C$ compact}~.
\end{equation}
On the other hand, when $C$ is non-compact, the solution $\Pi(C,\e)$ is not
analytic at $\e=0$. As $X_\bt$ is non-compact, there is no fundamental
class in homology and this is reflected in the fact that $\Pi(X_\bt,\e)$
does not admit a non-equivariant limit of the form \cref{eq:CY3-t3}.
To obtain a well-defined
non-equivariant quantity, we need to perform some regularization.\footnote
{For instance, even classical intersection numbers are not uniquely defined
unless the intersection locus is compact
(see ref.~\cite{Forbes:2005wd} for earlier attempts at regularization).}

The number of independent solutions of the equivariant PF equations
is equal to the number of fixed points,
which is the same as the Euler number $\chi$.
By definition, this equals the dimension of the homology lattice, i.e.\ the
number of independent compact cycles $C$.
This implies that compact equivariant periods generate all PF solutions
and the non-equivariant limit preserves the total number of independent
solutions.

The GV expansion of the GW potential is expected to have an equivariant
generalization but these formulas have not been derived yet. Nevertheless, we
can read some non-equivariant invariants from the $\e \to 0$ limit of the
solutions $\Pi(C,\e)$ when $C$ is compact. The numerical invariants obtained
this way are well-defined and non-ambiguous.
However, not all GV invariants $n_d$ can be obtained this way.
Those that do not appear in the limit of compact solutions
are only defined equivariantly. A regularization scheme for these solutions is
necessary and we show in examples that this allows to compute the integers
$n_d$.
The result however depends on the chosen regularization scheme
and we argue that there is an intrinsic ambiguity in their definition
as non-equivariant quantities.

We argue that, for solutions with regular behavior in $\e$, the
same type of GV formulas hold once the non-equivariant limit is taken, while
for those that do not admit a limit a regularization needs to be performed
first. For the latter, GV formulas only hold up to a correction term
$\delta$ that is annihilated by all compact divisor operators. This term can
bring both classical and quantum corrections that depend on some
non-canonical choices. In particular, we argue that $\cf^D_\reg$ as defined in
\cref{sec:regularization} provides a regularization for the equivariant
solution $\Pi(X_\bt,\e)$.

For toric CYs with $H^2_\cmp(X_\bt)\neq0$,
we define a regularized volume as any function
$\cf_\reg(\bt,\e)$ that is analytic at $\e=0$ and such that
\begin{equation}
\label{eq:regFcl}
 \cD_i\cf_\reg(\bt,\e) = \cD_i\cf(\bt,\e)~,
 \quad
 \forall i \in I_\cmp~.
\end{equation}
If $H^2_\cmp(X_\bt)$ is empty but $H^4_\cmp(X_\bt) \neq 0$, then
we define a regularized volume as any regular function such that
\begin{equation}
 \cD_i\cD_j\cf_\reg(\bt,\e) = \cD_i\cD_j\cf(\bt,\e)~,
 \quad
 \forall i,j \text{ s.t. $D_i\cap D_j$ is compact}~.
\end{equation}
and similarly for higher-codimension compact intersections.
This condition guarantees that when the intersection is compact the
corresponding intersection numbers are the same before and after regularization.
From this we can define regularized intersection numbers
\begin{equation}
 \kappa^\reg_{a_1,\dots,a_\cd}
 = \frac{\partial^\cd}
 {\partial t^{a_1} \cdots \partial t^{a_\cd}}
 \cf_\reg(\bt,0) \in \mathbb Q ~.
\end{equation}
\begin{remark}
If $H^2_\cmp (X_\bt) \cong H_{2\cd-2} (X_\bt)$ is non-empty,
then there is at least one compact divisor $D_i = \sum_{a} D_i^a C_a$
and the corresponding equivariant period $\Pi(D_i,\e)$ is regular.
Then by \cref{eq:regFcl} this period is equal to its regularization,
\begin{equation}
 \Pi_\reg(D_i,\e)\equiv\Pi(D_i,\e)~.
\end{equation}
Similarly, if $H^4_\cmp (X_\bt) \cong H_{2\cd-4} (X_\bt)$ is non-empty,
we can find two divisors that intersect to a compact subspace
and the corresponding period is regular
\begin{equation}
 \Pi_\reg(D_i\cap D_j,\e)\equiv\Pi(D_i\cap D_j,\e)~.
\end{equation}
\end{remark}
From the remark it follows that for a toric CY three-fold with a
compact divisor $D_i$
\begin{multline}
\label{eq:compactDivPeriod}
 \Pi(D_i,0) = \lim_{\e\to0}\sum_a D^a_i\,\Pi(C_a,\e)
 = -\frac{1}{\lambda}\sum_a D^a_i \frac{\partial\Phi^0}{\partial\td{t}^a} \\
 = \sum_a D_i^a \left[
 \tfrac12 \sum_{b,c} \kappa^\reg_{abc} \log\td{z}_b \log\td{z}_c
 + \sum_{\bd\neq0} n_{\bd} d_a \Li_2(\td{z}^{\bd})
 \right]~.
\end{multline}
While the combination of derivatives of the GW potential
in \cref{eq:compactDivPeriod} is well-defined in the
non-equivariant limit, this is not necessarily true for each single derivative
$\partial\Phi^0/\partial\td{t}^a$ as the periods $\Pi(C_a,\e)$ might
not have a regular behavior when considered individually.
In the next sections we show this explicitly for some concrete examples
(see \cref{ex:f2,ex:localA2}).
For a toric CY four-fold with a compact intersection $D_i\cap D_j$,
we obtain
\begin{multline}
 \Pi(D_i\cap D_j,0) = \lim_{\e\to0}\sum_{a,b} D^a_i D^b_j\, \Pi(C_{ab},\e) \\
 = \sum_{a,b} D^a_i D^b_j \left[
 \tfrac12 \sum_{c,d} \kappa^\reg_{abcd} \log\td{z}_c \log\td{z}_d
 + \sum_{\bd\neq0}n_{\bd}(C_{ab})\Li_2(\td{z}^{\bd}) \right]~.
\end{multline}
While the limit of the double sum is well-defined,
each term $\Pi(C_{ab},\e)$ may be singular.

Let us define the function
\begin{equation}
 \Pi_\reg(X_\bt,\e)
 := (-\lambda)^\cd \hat{I}_{X_\bt} \cf_\reg(\bt,\e)~,
\end{equation}
which by construction satisfies the following properties:
\begin{itemize}
\item it is analytic at $\e=0$,
\item it satisfies the modified PF equations in \cref{eq:modifiedPF}
\end{itemize}
Observe, as previously pointed out, that the choice of regularization is
not unique and the prescription in \cref{sec:regularization} is different
from $\Pi_\reg(X_\bt,\e)$. It is however true that both choices
carry a certain amount of ``universal'' enumerative geometric data that is
regularization independent and leads to well-defined integer GV invariants.
The difference between the two regularization schemes is due to some
intrinsic ambiguity in the definition of the non-equivariant limit of
$\Pi(X_\bt,\e)$. The exact relation between the two regularizations
is clarified by the following.
\begin{conjecture}
Let $X_\bt$ be a smooth toric CY three-fold.
The following GV formula holds
\begin{equation} \label{eq:localCY3-t3}
 \lim_{\e\to0} \Pi_\reg(X_\bt,\e)
 = \tfrac16 \sum_{a,b,c} \kappa^\reg_{abc}
 \log\td{z}_a \log\td{z}_b \log\td{z}_c
 + \sum_{\bd\neq0} n^\reg_{\bd} \log(\td{z}^{\bd}) \Li_2(\td{z}^{\bd})
 -2 \sum_{\bd\neq0} n^\reg_{\bd} \Li_3(\td{z}^{\bd})~,
\end{equation}
where the GV invariants are also regularized.
The regularized disk function in \cref{eq:Freg} and
the regularized period $\Pi_\reg(X_\bt,\e)$ are related as
\begin{equation} \label{eq:Freg-GV-CY3}
 (-\lambda)^3\lim_{\e\to0}\cf^D_\reg(\bt,\e;\lambda)
 = \Pi_\reg(X_\bt,0)
 - \tfrac{\pi^2}{6} \sum_{a,b,c}
 \tfrac12\kappa^\reg_{abc}\,c_2^{ab}\log\td{z}_c
 + \zeta(3) \chi
 + \delta~,
\end{equation}
where $c_2 = \tfrac12 \sum_{a,b} c_2^{ab} \phi_a\phi_b$ is the second
Chern class and $\delta$ is in the kernel of all compact divisor operators.
If $X_\bt$ has a compact divisor $D_i$ and $n_{\bd}$ can be read from
\cref{eq:compactDivPeriod}, then that integer is uniquely defined, 
$n^\reg_{\bd}\equiv n_{\bd}$.
If instead $n^\reg_{\bd}$ only appears in \cref{eq:localCY3-t3}
(i.e. when $\sum_a d_a D^a_i=0$ for all $i\in I_\cmp$), then its value is
not guaranteed to be integer and it might depend on the choice of
regularization.
\end{conjecture}

We observe
that not only classical intersection numbers need regularization but in some
cases also the instantonic contributions that define the GV invariants.
As discussed in \cref{eq:singular-instantons}, this happens when $\cf^D_\sing$
contains both classical and instantonic contributions. We see two instances
of this in the examples of $K_{F_2}$ and local $A_2$ geometry.

For general toric CYs the analogous claim reads
\begin{equation}
 (-\lambda)^\cd \lim_{\e \to 0} \cf^D_\reg (\bt,\e;\lambda)
 = \Pi_\reg (X_\bt,0) +\text{sub-leading} + \delta~,
\end{equation}
where $\Pi_\reg(X_\bt,\e)$ is obtained by regularizing the classical
intersection numbers and then applying the Givental operator. The sub-leading
terms are fixed by the expansion of the Gamma-class, see \cref{eq:gamma-class}.
The presence of the correction term $\delta$ is due to the fact
that regularization and $\hat{I}_{X_\bt}$ operator do not commute, which
means that $\cf^D_\reg$ is not necessarily in the image of the Givental
operator.

%% file: examples.tex
\section{Examples with compact divisors}
\label{sec:examples1}

\subsection{\texorpdfstring{$\co(-2)$}{O(-2)} over \texorpdfstring{$\BP^1$}{P1}}
\label{sec:A1}

Consider $X_\bt = \tot K_{\BP^1}$,
the total space of the canonical bundle over $\BP^1$,
defined by charge matrix $Q = (1, -2, 1)$ and chamber $t>0$,
also known as the $A_1$ space. Its symplectic volume is
\begin{equation}
 \cf (t,\e) = \oint_\jk  \frac{\dif\phi}{2\pi\ii} \,
 \frac{\eu^{\phi t}} {(\e_1 + \phi)(\e_2 - 2 \phi)(\e_3 + \phi)}
 = \frac{\eu^{-\e_1 t}}{\left(\e_2+2\e_1\right)\left(\e_3-\e_1\right)}
 + \frac{\eu^{-\e_3 t}}{\left(\e_1-\e_3\right) \left(\e_2+2 \e_3\right)}~,
\end{equation}
where JK contour selects poles at $\phi=-\e_1$ and $\phi=-\e_3$.
We define differential operators
\begin{equation}
 \cD_1 = \e_1 +\tfrac{\partial}{\partial t}~,
 \quad
 \cD_2 = \e_2-2\tfrac{\partial}{\partial t}~,
 \quad
 \cD_3 = \e_3 +\tfrac{\partial}{\partial t}~.
\end{equation}
Acting with the operator $\cD_1 \cD_3$ we kill both poles inside of JK,
so we get the relation
\begin{equation}
\label{eq:cl-cohom-A1}
 \cD_1 \cD_3 \cf (t,\e) = 0~,
\end{equation}
which corresponds to the description of equivariant cohomology of $X_\bt$ as
\begin{equation}
 H^\bullet_{\sft}(X_{\bt}) \cong \BC [\phi, \e_1, \e_2, \e_3]
 / \langle x_1 x_3 \rangle~.
\end{equation}
The generic solution to \cref{eq:cl-cohom-A1} takes the form
\begin{equation}
\label{eq:A1classical-solution}
 \cf(t,\e) = c_1(\e)\, \eu^{-\e_1 t} + c_3(\e)\, \eu^{-\e_3 t}
\end{equation}
with $c_1$ and $c_2$ integration constants, which may depend on
$\e_i$ but not on $t$. Indeed our symplectic volume is of this form, with
\begin{equation}
 c_1 = \frac{1} {\left( \e_2+2 \e_1 \right) \left( \e_3 - \e_1 \right)}~,
 \qquad
 c_3 = \frac{1} {\left( \e_2+2 \e_3 \right) \left( \e_1 - \e_3 \right)}~.
\end{equation}

The space $X_\bt$ has a single compact divisor $D_2$ corresponding to
the $\BP^1$ base of the bundle. It follows that
\begin{equation}
 \cD_2 \cf (t,\e) \quad\text{is analytic at $\e=0$}~.
\end{equation}
The cohomological disk partition function is
\begin{equation}
 \cf^D(t,\e;\lambda) =
 \lambda^{-3} \oint_\qjk \frac{\dif\phi}{2\pi\ii} \, \eu^{\phi t}
 \Ga \left( \tfrac{\e_1 + \phi}{\lambda} \right)
 \Ga \left( \tfrac{\e_2 - 2 \phi}{\lambda} \right)
 \Ga \left( \tfrac{\e_3 + \phi}{\lambda} \right)
\end{equation}
with QJK selecting poles at $\phi = - \e_1 - k \lambda$ and
$\phi = - \e_3 - k \lambda$ for $k \in\BZ_{\geq0}$. The classical cohomology 
relation
gets deformed to the quantum cohomology relation
\begin{equation}
\label{eq:A1quantum-cohomology}
 \left[ \cD_1 \cD_3 - \eu^{-\lambda t} (\lambda + \cD_2) \cD_2 \right]
 \cf^D (t,\e;\lambda) = 0~,
\end{equation}
which we can prove as follows:
\begin{equation}
\label{eq:A1quantum-cohomology-proof}
\begin{aligned}
 {\cD}_1 {\cD}_3 \cf^{D} (t,\e;\lambda)
 &= \lambda^{-1} \oint_\qjk \frac{\dif\phi}{2\pi\ii} \,
 \eu^{\phi t}
 \Ga \left( \tfrac{\e_1 + \phi+\lambda}{\lambda} \right)
 \Ga \left( \tfrac{\e_2 - 2 \phi}{\lambda} \right)
 \Ga \left( \tfrac{\e_3 + \phi+\lambda}{\lambda} \right) \\
 &= \lambda^{-1} \oint_{\qjk'} \frac{\dif\phi'}{2\pi\ii} \,
 \eu^{(\phi'-\lambda) t}
 \Ga \left( \tfrac{\e_1 + \phi'}{\lambda} \right)
 \Ga \left( \tfrac{\e_2 - 2 \phi'+2\lambda}{\lambda} \right)
 \Ga \left( \tfrac{\e_3 + \phi'}{\lambda} \right) \\
 &= \eu^{-\lambda t}(\cD_2+\lambda)\cD_2\lambda^{-3} \oint_{\qjk'}
 \frac{\dif\phi'}{2\pi\ii} \, \eu^{\phi' t}
 \Ga \left( \tfrac{\e_1 + \phi'}{\lambda} \right)
 \Ga \left( \tfrac{\e_2 - 2 \phi'}{\lambda} \right)
 \Ga \left( \tfrac{\e_3 + \phi'}{\lambda} \right) \\
 &= \eu^{-\lambda t}({\cD}_2+\lambda) {\cD}_2 \cf^{D} (t,\e;\lambda)~.
\end{aligned}
\end{equation}
Here we repeatedly use the property $x \Ga(x) = \Ga(x+1)$ together with the
change of variable $\phi' = \phi + \lambda$. Under this change of variables,
the QJK contour goes to $\qjk'$, which picks the poles at
$\phi' = -\e_{1,2} - k'\lambda$ with $k' = k-1 \geq -1$, but
at $k' = -1$ there are no poles in the integrand, so we can use the
original contour: when we act with
$\cD_1 \cD_3$, the two classical poles at $k = 0$ are killed
and the contour retracts until the next poles at $k = 1$, i.e.\ $k' = 0$.

An explicit residue computation yields the series expansion of the disk
partition function
\begin{multline}
\label{eq:A1cohom-residue}
 \cf^{D} (t,\e;\lambda)
 = \lambda^{-2} \sum_{d=0}^\infty \eu^{-d\lambda t} \frac{(-1)^d}{d!}
 \left[ \eu^{-\e_1 t} \Ga \left( \tfrac{\e_2+2\e_1}{\lambda} + 2d \right)
 \Ga \left( \tfrac{\e_3-\e_1}{\lambda} - d \right) \right. \\
 \left. {}+ \eu^{-\e_3 t} \Ga \left( \tfrac{\e_1-\e_3}{\lambda} - d \right)
 \Ga \left( \tfrac{\e_2 + 2 \e_3}{\lambda} + 2d \right) \right]~,
\end{multline}
where $z\equiv\eu^{-\lambda t}$ can be regarded as an instanton counting
parameter that distinguishes between contributions of maps of different degree.
If we restrict to zero-instanton sector (the contribution of classical poles)
we obtain the classical part of the disk function
\begin{equation}
 \cf_{\Ga} (t,\e;\lambda)
 = \lambda^{-2} \left[ \eu^{-\e_1 t}
 \Ga \left( \tfrac{\e_2+2\e_1}{\lambda} \right)
 \Ga \left( \tfrac{\e_3- \e_1}{\lambda} \right)
 + \eu^{-\e_3 t}
 \Ga \left( \tfrac{\e_1- \e_3}{\lambda} \right)
 \Ga \left( \tfrac{\e_2+2\e_3}{\lambda} \right) \right]~.
\end{equation}
This is of the same type as the solution in \cref{eq:A1classical-solution}
and it satisfies the relation
\begin{equation}
 \cD_1 \cD_3 \cf_{\Ga} (t,\e;\lambda) = 0
\end{equation}
of classical cohomology.
In the limit $\lambda \to \infty$ both $\cf^D$ and $\cf_{\Ga}$ reduce
to the equivariant volume $\cf$.
Let us analyze \cref{eq:A1quantum-cohomology} and its solutions. Through some
formal manipulations we can re-write it as
\begin{equation}
 \left[ 1-\eu^{-\lambda t}
	\frac {(\cD_2+\lambda)\cD_2} {(\cD_1-\lambda)(\cD_3-\lambda)}
 \right]
 \cf^D(t,\e;\lambda) = \cf_{\Ga} (t,\e;\lambda)~.
\end{equation}
We can invert the operator on the LHS to obtain the solution
\begin{equation}
\begin{aligned}
 \cf^D (t,\e;\lambda)
 &= \sum_{d=0}^\infty
 \left( \eu^{-\lambda t}
 \frac {(\cD_2+\lambda) \cD_2} {(\cD_1-\lambda)(\cD_3-\lambda)}
 \right)^d
 \cf_{\Ga} (t,\e;\lambda) \\
 &= \left[
 \sum_{d=0}^\infty \eu^{-d\lambda t}
 \frac{\left(\tfrac{\cD_2}{\lambda}\right)_{2d}}
 {\left(1-\tfrac{\cD_1}{\lambda}\right)_{d}
 \left(1-\tfrac{\cD_3}{\lambda}\right)_{d}}
 \right]
 \cf_{\Ga} (t,\e;\lambda) \\
 &= {}_3F_2 \left(
 1, \tfrac{\cD_2}{2\lambda} + \tfrac{1}{2}, \tfrac{\cD_2}{2\lambda};
 1-\tfrac{\cD_1}{\lambda }, 1-\tfrac{\cD_3}{\lambda}; 4 \eu^{-\lambda t}
 \right)
 \cf_{\Ga} (t,\e;\lambda) ~,
\end{aligned}
\end{equation}
where we used the identity
\begin{equation}
 \left( \eu^{-\lambda t}f \left( \tfrac{\partial}{\partial t} \right)\right)^d
 = \eu^{-d\lambda t} \prod_{i=0}^{d-1} f
	\left( \tfrac{\partial}{\partial t}-i\lambda \right)~.
\end{equation}
Substituting as initial condition $\cf_{\Ga} (t,\e;\lambda)$ as in
\cref{eq:A1classical-solution} we obtain
\begin{multline}
 \cf^D(t,\e;\lambda)
	= c_1 \eu^{-\e_1 t} {}_2F_1 \left(
	\tfrac{\epsilon_2+2\epsilon_1}{2\lambda},
	\tfrac{\lambda+2\epsilon_1+\epsilon_2}{2\lambda};
	\tfrac{\epsilon_1-\epsilon_3}{\lambda}+1;
	4 \eu^{-\lambda t}
 \right) \\
	+ c_3 \eu^{-\e_3 t} {}_2F_1 \left(
	\tfrac {\epsilon_2+2\epsilon_3} {2\lambda},
	\tfrac{\lambda+2\epsilon_3+\epsilon_2}{2\lambda};
	\tfrac{\epsilon_3-\epsilon_1}{\lambda}+1;
	4\eu^{-\lambda t}
 \right)
\end{multline}
so that for
\begin{equation}
 c_1 = \lambda^{-2}
 \Ga\left(\tfrac{\e_2+2\e_1}{\lambda}\right)
 \Ga\left(\tfrac{\e_3-\e_1}{\lambda}\right)~,
 \qquad
 c_3 = \lambda^{-2}
 \Ga\left(\tfrac{\e_2+2\e_3}{\lambda}\right)
 \Ga\left(\tfrac{\e_1-\e_3}{\lambda}\right)
\end{equation}
we can reproduce the computation of $\cf^D$ via residues
as in \cref{eq:A1cohom-residue}.

The K-theoretic disk partition function is represented by the integral
\begin{equation}
 \cz^D(T,q;\q) =
 - \oint_\qjk \frac{\dif w}{2\pi\ii w} \,
 w^{-T} \frac{1}{(q_1 w;\q)_\infty(q_2 w^{-2};\q)_\infty(q_3 w;\q)_\infty}
\end{equation}
with poles at $w=q_1^{-1}\q^{-k}$ and $w=q_2^{-1}\q^{-k}$ for $k\in\BZ_{\geq0}$.
A residue computation gives
\begin{multline}
\label{eq:A1-residue-Kth}
 \cz^D(T,q;\q) = \sum_{d=0}^\infty \frac {\q^{d T}} {(\q;\q)_\infty 
(\q^{-d};\q)_d}
 \left[ \frac {q_1^T}
	{(q_2q_1^2\q^{2d};\q)_\infty (q_3q_1^{-1}\q^{-d};\q)_\infty}
 \right.\\
 \left. {} + \frac {q_3^T}
	{(q_1q_3^{-1}\q^{-d};\q)_\infty (q_2q_3^2\q^{2d};\q)_\infty}
 \right]~,
\end{multline}
where $\q^T$ can be regarded as instanton counting parameter.

We define the shift operators
\begin{equation}
 \Delta_1 = q_1(T^\dagger)^{-1}~,
 \quad
 \Delta_2 = q_2(T^\dagger)^{2}~,
 \quad
 \Delta_3 = q_3(T^\dagger)^{-1}~.
\end{equation}
The K-theoretic compact divisor equation is
\begin{equation}
 (1-\Delta_2) \cz^D(T,q;\q) = \cz^D(T,q;\q) - q_2 \cz^D(T+2,q;\q)
 = \sum_{n^2=0}^{\infty}\frac{(\q q_2)^{n^2}}{(\q;\q)_{n^2}}
 \sum_{\Lambda_2(T,n^2)} \frac{q_1^{n^1}q_3^{n^3}}{(\q;\q)_{n^1}(\q;\q)_{n^3}}~,
\end{equation}
where $\Lambda_2 (T,n^2) = \{ (n^1,n^3) \in \BN^2 \,|\, n^1+n^3=T+2n^2 \}$.
By the argument in \cref{conj:compact-divisor-K-th} the RHS is regular
in the $q_1, q_2, q_3 \to 0$ limit.

The classical equivariant K-theory ring
\begin{equation}
 K_{\sft}(X_{\bt}) \cong \BC [w^{\pm},q_1^{\pm},q_2^{\pm},q_3^{\pm}]
	/ \langle(1-q_1w)(1-q_3w)\rangle
\end{equation}
is defined by the relation
\begin{equation}
 (1-\Delta_1) (1-\Delta_3) \cz_{\Ga_\q} (T,q) = 0~,
\end{equation}
whose generic solution is
\begin{equation}
 \cz_{\Ga_\q} (T,q) = c_1 q_1^T + c_3 q_3^T~.
\end{equation}
The quantum K-theory ring is then defined by the relation
\begin{equation}
 \left[ (1-\Delta_1) (1-\Delta_3) - \q^T (1-\q\Delta_2) (1-\Delta_2) \right]
	\cz^D (T,q;\q) = 0~,
\end{equation}
which can be derived similarly to \cref{eq:A1quantum-cohomology-proof}
by using the property in \cref{eq:q-pochhammer-difference}.

The quantum K-theory relation can be rewritten as
\begin{equation}
 \left[ 1 -
 \q^T \frac{(1-\q\Delta_2)(1-\Delta_2)} {(1-\q^{-1}\Delta_1)(1-\q^{-1}\Delta_3)}
 \right]
 \cz^D (T,q;\q) = \cz_{\Ga_\q} (T,q)
\end{equation}
and its solution is formally given by
\begin{equation}
\begin{aligned}
 \cz^D(T,q;\q)
 &= \sum_{d=0}^\infty \left(
 \q^T \frac{(1-\q\Delta_2)(1-\Delta_2)}{(1-\q^{-1}\Delta_1)(1-\q^{-1}\Delta_3)}
 \right)^d \cz_{\Ga_\q} (T,q)  \\
 &= \left[\sum_{d=0}^\infty \q^{dT}
 \frac{ (\Delta_2;\q)_{2d}}
 {(\q^{-d}\Delta_1;\q)_{d}(\q^{-d}\Delta_3;\q)_{d}} \right]
 \cz_{\Ga_\q} (T,q)~,
\end{aligned}
\end{equation}
where we used
\begin{equation}
 \left(\q^T f\left(T^\dagger\right)\right)^d
 = \q^{dT} \prod_{i=0}^{d-1} f\left(\q^i T^\dagger\right)~.
\end{equation}
With the initial data
\begin{equation}
 c_1 = \frac{1}{(\q;\q)_\infty(q_2q_1^2;\q)_\infty(q_3q_1^{-1};\q)_\infty}~,
 \qquad
 c_3 = \frac{1}{(\q;\q)_\infty(q_2q_3^2;\q)_\infty(q_1q_3^{-1};\q)_\infty}
\end{equation}
we reproduce the function $\cz^D$ obtained by residues in 
\cref{eq:A1-residue-Kth}.

The solutions to the equivariant PF equations are
\begin{equation}
 \Pi(p_i,\e) = z^{\frac{\e_i}{\lambda}}\sum_{d=0}^\infty z^d
 \frac{\left(\frac{\e_2+2\e_i}{\lambda}\right)_{2d}}
 {\left(1-\frac{\e_1-\e_i}{\lambda}\right)_d
  \left(1-\frac{\e_3-\e_i}{\lambda}\right)_d},\quad i=1,3
\end{equation}
one for each fixed point.
The non-equivariant $\hat{I}$-operator is
\begin{equation}
\label{eq:A1-I-op}
 \lim_{\e\to0} \hat{I}_{X_\bt} = \sum_{d=0}^\infty z^d
 \frac{(2\theta)_{2d}}
 {(1+\theta)_d^2}
 = 1+2G(z)\theta
 +2G(z)^2\theta^2
 +\dots
\end{equation}
with $\theta=z\partial_z$ and
\begin{equation}
 G(z) := \sum_{d=1}^\infty z^d \frac{\Ga(2d)}{\Ga(d+1)^2}
 = -\log\left(\frac{1+\sqrt{1-4 z}}{2}\right)~.
\end{equation}
The solutions to the non-equivariant PF equations are
\begin{equation}
\begin{aligned}
 \Pi(\pt) &= 1 ~,\\
 \Pi(\BP^1) &= \log z + 2 G(z) = \log\td{z} ~,
\end{aligned}
\end{equation}
corresponding to the degree-zero and degree-two generators of the homology 
lattice.
The solution of logarithmic degree one defines the mirror map to flat 
coordinates
$\td{z}=z\,\eu^{2G(z)}$.
As $D_2\cong\BP^1$ is compact, we have the identity
\begin{equation}
 \lim_{\e\to0} (-\lambda) \hat{I}_{X_\bt}\cD_2\cf_{\Ga}
 = \Pi(\BP^1)~.
\end{equation}
The fundamental cycle of $X_\bt$ is non-compact and therefore only defined
equivariantly. Its regularization is annihilated by the modified PF operator
\begin{equation}
 \cD_1 \cD_2 \cD_3 - z \cD_2(\cD_2+\lambda)(\cD_2+2\lambda)
\end{equation}
and it can be computed using \cref{eq:A1-I-op} as
\begin{equation}
 \Pi_\reg(X_\bt) = - \tfrac14 \log^2\td{z}~.
\end{equation}
Moreover,
\begin{equation}
 \lim_{\e\to0} (-\lambda)^2 \cf^D_\reg = \Pi_\reg(X_\bt)
\end{equation}
so that in the flat coordinates $\td{z}=\frac{4z}
{\left(1+\sqrt{1-4 z}\right)^2}$ there are no instanton corrections
and the GV invariants are all vanishing,
which is compatible with the fact that $X_\bt$ is Hyperkähler.

\subsection{\texorpdfstring{$A_2$}{A2} geometry}

Consider the $A_2$ geometry defined by the charge matrix
\begin{equation}
Q =
\begin{pmatrix}
1 & -2 & 1 & 0 \\
0 & 1 & -2 & 1
\end{pmatrix}
\end{equation}
and chamber $t^1,t^2>0$.
Its symplectic volume is
\begin{equation}
 \cf(\bt,\e) = \oint_\jk  \frac{\dif\phi_1 \dif\phi_2}{(2\pi \ii)^2} \,
 \frac {\eu^{\phi_1 t_1+\phi_2 t_2}}
 {(\e_1 + \phi_1) (\e_2 - 2 \phi_1 + \phi_2)
	(\e_3 + \phi_1 - 2\phi_2) (\e_4 + \phi_2)}~,
\end{equation}
where poles are located at
$(-\e_1, -\e_2-2\e_1)$, $(-\e_1, -\e_4)$ and $(-\e_3-2\e_4, - \e_4)$.
We have the following classical cohomology relations
\begin{equation}
 {\cD}_1 {\cD}_4 \cf (\bt,\e) = 0~,
 \quad
 {\cD}_1 {\cD}_3 \cf (\bt,\e) = 0~,
 \quad
 {\cD}_2 {\cD}_4 \cf (\bt,\e) = 0~,
\end{equation}
so the equivariant cohomology ring is given by
\begin{equation}
 \BC [\phi_1, \phi_2, \e_1,\e_2,\e_3,\e_4] /
 \langle x_1 x_4, x_1 x_3, x_2 x_4 \rangle~.
\end{equation}
There are two compact divisors $D_2$ and $D_3$.

The K-theoretic disk function is
\begin{equation}
\label{eq:A2Kth-disk-func}
 \cz^{D}(\bT,q;\q) =
 \oint_\qjk \frac{\dif w_1\dif w_2}{(2\pi \ii)^2w_1w_2} \,
 \frac {w_1^{-T^1}w_2^{-T^2}}
 {(q_1w_1;\q)_\infty (q_2w_1^{-2}w_2;\q)_\infty
	(q_3w_1w_2^{-2};\q)_\infty (q_4w_2;\q)_\infty}
\end{equation}
with poles at $(\q^{-k_1}q_1^{-1},\q^{-2k_1-k_2}q_1^{-2}q_2^{-1})$,
$(\q^{-k_1}q_1^{-1},\q^{-k_2}q_4^{-1})$ and
$(\q^{-k_1-2k_2}q_3^{-1}q_4^{-2},\q^{-k_2}q_4^{-1})$ for $k_1,k_2\geq0$.
The three towers of poles correspond to instanton contributions coming from the
three fixed points in $X_{\bt}$.
We get the following quantum K-theory relations
\begin{equation}
\begin{aligned}
 \left[ (1-\Delta_1)(1-\Delta_4) - \q^{T^1+T^2} (1-\Delta_2)(1-\Delta_3) \right]
	\cz^D(\bT,q;\q) &= 0 ~,\\
 \left[ (1-\Delta_1)(1-\Delta_3) - \q^{T^1} (1-\Delta_2)(1-\q\Delta_2) \right]
	\cz^D(\bT,q;\q) &= 0 ~,\\
 \left[ (1-\Delta_2)(1-\Delta_4) - \q^{T^2} (1-\Delta_3)(1-\q\Delta_3) \right]
	\cz^D(\bT,q;\q) &= 0 ~.
\end{aligned}
\end{equation}
We define the K-theoretic Givental operator
\begin{equation}
 \hat{I}^K_{X_\bt} =
 \sum_{d_1,d_2=0}^\infty \q^{d_1T^1+d_2T^2}
 (\Delta_1;\q)_{-d_1}
 (\Delta_2;\q)_{2d_1-d_2}
 (\Delta_3;\q)_{-d_1+2d_2}
 (\Delta_4;\q)_{-d_2}
\end{equation}
so that we can write the solution as
\begin{equation}
 \cz^D(\bT,q;\q) = \hat{I}^K_{X_\bt}\cdot \cz_{\Ga_\q}(\bT,q)
\end{equation}
with
\begin{equation}
 \cz_{\Ga_\q} (\bT,q) =
 c_{1,2}\, q_1^{T^1+2T^2}q_2^{T^2} +
 c_{1,4}\, q_1^{T^1}q_4^{T^2} +
 c_{3,4}\, q_3^{T^1}q_4^{2T^1+T^2}~.
\end{equation}
The integration coefficients $c_{1,2}$, $c_{1,4}$, $c_{3,4}$ are not fixed
by the equations and they parametrize the moduli space of solutions.
The solution corresponding to the function $\cz^D$ defined by the integral
in \cref{eq:A2Kth-disk-func} is given by the choice of semi-classical data
\begin{equation}
\begin{aligned}
 c_{1,2} &= \frac{1}
 {(\q;\q)^2_\infty (q_1^3q_2^2q_3;\q)_\infty 
(q_1^{-2}q_2^{-1}q_4;\q)_\infty}~,\\
 c_{1,4} &= \frac{1}
 {(\q;\q)^2_\infty (q_1^2q_2^2q_4^{-1};\q)_\infty 
(q_1^{-1}q_3q_4^2;\q)_\infty}~,\\
 c_{3,4} &= \frac{1}
 {(\q;\q)^2_\infty (q_1q_3^{-1}q_4^{-2};\q)_\infty (q_2q_3^2q_4^3;\q)_\infty}~.
\end{aligned}
\end{equation}

In the cohomological limit we have
\begin{equation}
 \cf^{D}(\bt,\e;\lambda) = \hat{I}_{X_\bt} \cdot \cf_{\Ga} (\bt,\e)
\end{equation}
with
\begin{equation}
\begin{aligned}
 \hat{I}_{X_\bt}=
 & \sum_{d_1,d_2=0}^\infty z_1^{d_1}z_2^{d_2}
 \left(\tfrac{\cD_1}{\lambda}\right)_{-d_1}
 \left(\tfrac{\cD_2}{\lambda}\right)_{2d_1-d_2}
 \left(\tfrac{\cD_3}{\lambda}\right)_{-d_1+2d_2}
 \left(\tfrac{\cD_4}{\lambda}\right)_{-d_2} \\
 =& 1 + \sum_{\substack{2d_1-d_2>0\\-d_1+2d_2\leq0}}
 z_1^{d_1}(-z_2)^{d_2}
 \frac
 {\left(\frac{\cD_2}{\lambda}\right)_{2d_1-d_2}}
 {\left(1-\frac{\cD_1}{\lambda}\right)_{d_1}
  \left(1-\frac{\cD_3}{\lambda}\right)_{d_1-2d_2}
  \left(1-\frac{\cD_4}{\lambda}\right)_{d_2}} \\
 & +\sum_{\substack{2d_1-d_2\leq0\\-d_1+2d_2>0}}
 (-z_1)^{d_1}z_2^{d_2}
 \frac
 {\left(\frac{\cD_3}{\lambda}\right)_{-d_1+2d_2}}
 {\left(1-\frac{\cD_1}{\lambda}\right)_{d_1}
  \left(1-\frac{\cD_2}{\lambda}\right)_{-2d_1+d_2}
  \left(1-\frac{\cD_4}{\lambda}\right)_{d_2}} \\
 &+ \sum_{\substack{2d_1-d_2>0\\-d_1+2d_2>0}}
 (-z_1)^{d_1}(-z_2)^{d_2}
 \frac
 {\left(\frac{\cD_2}{\lambda}\right)_{2d_1-d_2}
  \left(\frac{\cD_3}{\lambda}\right)_{-d_1+2d_2}}
 {\left(1-\frac{\cD_1}{\lambda}\right)_{d_1}
  \left(1-\frac{\cD_4}{\lambda}\right)_{d_2}} 
\end{aligned}
\end{equation}
and initial data
\begin{equation}
 \cf_{\Ga} (\bt,\e) =
 c_{1,2}\, \eu^{-\e_1(t^1+2t^2)-\e_2t^2} +
 c_{1,4}\, \eu^{-\e_1t^1-\e_4t^2} +
 c_{3,4}\, \eu^{-\e_3t^1-\e_4(2t^1+t^2)}
\end{equation}
and
\begin{equation}
\begin{aligned}
 c_{1,2} &= \lambda^{-2} \Ga\left(\tfrac{3\e_1+2\e_2+\e_3}{\lambda}\right)
 \Ga\left(\tfrac{-2\e_1-\e_2+\e_4}{\lambda}\right) ~,\\
 c_{1,4} &= \lambda^{-2} \Ga\left(\tfrac{2\e_1+2\e_2-\e_4}{\lambda}\right)
 \Ga\left(\tfrac{-\e_1+\e_3+2\e_4}{\lambda}\right) ~, \\
 c_{3,4} &= \lambda^{-2} \Ga\left(\tfrac{\e_1-\e_3-2\e_4}{\lambda}\right)
 \Ga\left(\tfrac{\e_2+2\e_3+3\e_4}{\lambda}\right)~.
\end{aligned}
\end{equation}
The function $\cf^D$ is annihilated
by the following set of equivariant PF operators
\begin{equation}
\begin{aligned}
 \PF^\eq_{(1,1)} &= \cD_1\cD_4 - z_1 z_2  \cD_2\cD_3 ~, \\
 \PF^\eq_{(1,0)} &= \cD_1\cD_3 - z_1  \cD_2(\cD_2+\lambda) ~, \\
 \PF^\eq_{(0,1)} &= \cD_2\cD_4 - z_2  \cD_3(\cD_3+\lambda)~,
\end{aligned}
\end{equation}
which encode the quantum cohomology relations of $X_\bt$.

The non-equivariant $\hat{I}$-operator can be expanded as
\begin{equation}
 \lim_{\e\to0}\hat{I}_{X_\bt}
 = 1+(2M_1-M_2)\theta_1
 +(-M_1+2M_2)\theta_2
 +\dots~,
\end{equation}
where we define
\begin{equation}
\begin{aligned}
 M_1(z_1,z_2)&:= \sum_{\substack{2d_1-d_2>0\\-d_1+2d_2\leq0}}
 \frac{\Ga(2d_1-d_2)}
 {\Ga(d_1+1)\Ga(d_1-2d_2+1)\Ga(d_2+1)} z_1^{d_1} (-z_2)^{d_2} ~, \\
 M_2(z_1,z_2)&:= \sum_{\substack{2d_1-d_2\leq0\\-d_1+2d_2>0}}
 \frac{\Ga(-d_1+2d_2)}
 {\Ga(d_1+1)\Ga(-2d_1+d_2+1)\Ga(d_2+1)} (-z_1)^{d_1} z_2^{d_2} ~,\\
 M_3(z_1,z_2)&:= \sum_{\substack{2d_1-d_2>0\\-d_1+2d_2>0}}
 \frac{\Ga(2d_1-d_2)\Ga(-d_1+2d_2)}
 {\Ga(d_1+1)\Ga(d_2+1)} (-z_1)^{d_1} (-z_2)^{d_2} ~.
\end{aligned}
\end{equation}
The elementary solutions to the non-equivariant PF equations
\begin{equation}
\begin{aligned}
 \Pi(\pt) &= 1~, \\
 \Pi(C^1) &= \log z_1 + 2M_1 - M_2 = \log \td{z}_1 ~,\\
 \Pi(C^2) &= \log z_2 - M_1 + 2M_2 = \log \td{z}_2 ~,
\end{aligned}
\end{equation}
correspond to the class of the point and
the two generators $C^1,C^2 \in H_2(X_\bt)$. We have
\begin{equation}
\begin{aligned}
 (-\lambda)\lim_{\e\to0} \hat{I}_{X_\bt} \cD_2 \cf_{\Ga} &= \Pi(C^1)~, \\
 (-\lambda)\lim_{\e\to0} \hat{I}_{X_\bt} \cD_3 \cf_{\Ga} &= \Pi(C^2)~.
\end{aligned}
\end{equation}
The modified PF equations admit an additional quadratic solution
that corresponds to the regularized fundamental cycle
\begin{equation}
 \Pi_\reg(X_\bt)
 = -\tfrac13(\log^2\td{z}_1+\log\td{z}_1\log\td{z}_2+\log^2\td{z}_2)~,
\end{equation}
and we have the identity
\begin{equation}
 (-\lambda)^2\lim_{\e\to0} \cf^D_\reg = \Pi_\reg(X_\bt)~,
\end{equation}
where
\begin{equation}
 \cf^D_\reg(\bt,\e;\lambda) = \cf^D(\bt,\e;\lambda)
 - \eu^{\frac{\e_2}{3}(2t^1+t^2)+\frac{\e_3}{3}(t^1+2t^2)}
 \cf_{\Ga}(0,\e)
\end{equation}
for a choice of left-inverse matrix
\begin{equation}
R =
\left(
\begin{array}{ccccc}
 -2/3 & -1/3 \\
 -1/3 & -2/3 \\
\end{array}
\right)
\end{equation}
as defined in \cref{eq:Freg}.
This solution (after mirror map) has no instanton corrections, which can be
explained by the fact that $X_\bt$ is Hyperkähler.

\subsection{\texorpdfstring{$\co(-n)$}{O(-n)} over
\texorpdfstring{$\BP^{n-1}$}{Pn-1}}

The space $X_\bt = \tot K_{\BP^{n-1}}$ is a toric CY
defined by the charge matrix
\begin{equation}
 Q =
 \begin{pmatrix}
  1 & 1 & \dots & 1 & -n \\
 \end{pmatrix}
\end{equation}
and the chamber $t>0$.
The symplectic volume is
\begin{equation}
 \cf(t,\e) = \oint_\jk\frac{\dif\phi}{2\pi\ii} \,
 \frac{\eu^{\phi t}}{(\e_{n+1}-n\phi) \prod_{i=1}^n(\e_i+\phi)}~,
\end{equation}
where we take the poles $\phi=-\e_i$ for $i=1,\dots,n$.
We have the following classical relations
\begin{equation}
 \left[\prod_{i=1}^n {\cD}_i \right] \cf(t,\e) = 0
\end{equation}
providing a representation for the equivariant cohomology
\begin{equation}
 H^\bullet_{\sft} (X_{\bt}) \cong \BC [\phi, \e_1,\e_2,\dots,\e_{n+1}] /
	\langle x_1\dots x_n \rangle~.
\end{equation}
The K-theoretic disk function is
\begin{equation}
 \cz^D(T,q;\q) =
 - \oint_\qjk \frac {\dif w} {2\pi\ii w} \,
 \frac {w^{-T}} {(q_{n+1}w^{-n};\q)_\infty \prod_{i=1}^n (q_iw;\q)_\infty}~,
\end{equation}
and it satisfies the quantum K-theory relation
\begin{equation}
 \left[ \prod_{i=1}^n (1-\Delta_i) -
	\q^T \prod_{k=0}^{n-1} (1-\q^k \Delta_{n+1}) \right]
 \cz^D (T,q;\q) = 0~.
\end{equation}
This is the K-theretic PF equation and it has the solution
\begin{equation}
 \cz^D(T,q;\q) =
 \sum_{d=0}^{\infty} \q^{dT}
 \frac{(\Delta_{n+1};\q)_{nd}}
 {\prod_{i=1}^{n}(\q^{-d}\Delta_i;\q)_{d}}
 \cz_{\Ga_\q} (T,q)
\end{equation}
with
\begin{equation}
 \cz_{\Ga_\q} (T,q)
 = \sum_{i=1}^n \frac{q_i^{T}}{(\q;\q)_\infty(q_0q_i^{n};\q)_\infty
 \prod_{j\neq i}(q_jq_i^{-1};\q)_\infty}~.
\end{equation}
In the cohomological limit we have the disk function
\begin{equation}
 \cf^{D}(t,\e;\lambda) =
 \lambda^{-(n+1)} \oint_\qjk \frac{\dif\phi}{2\pi\ii} \,
 \eu^{\phi t}\, \Ga \left( \frac{\e_{n+1}-n\phi}{\lambda} \right)
 \prod_{i=1}^n \Ga \left( \frac{\e_i+\phi}{\lambda} \right)~,
\end{equation}
which satisfies the quantum cohomology relation
\begin{equation}
 \left[ \prod_{i=1}^n {\cD}_i - \eu^{-\lambda t}
 \prod_{k=0}^{n-1}(k\lambda + \cD_{n+1}) \right]
 \cf^{D}(t,\e;\lambda) = 0~.
\end{equation}
The compact divisor of $X_{\bt}$ is $D_{n+1}$ and we have
\begin{equation}
 \cD_{n+1} \cf^D \quad\text{is analytic at $\e=0$}~.
\end{equation}
The Givental $\hat{I}$-operator can be written as
\begin{equation}
 \hat{I}_{X_\bt} =
 \sum_{d=0}^{\infty} (-1)^{nd}z^d
 \frac{\left(\frac{\cD_{n+1}}{\lambda}\right)_{nd}}
 {\prod_{i=1}^{n}\left(1-\frac{\cD_i}{\lambda}\right)_{d}}~,
\end{equation}
which implies that all instanton contributions for $d>0$ are regular and the 
only
singular term comes from the semi-classical contribution at $d=0$.
In the non-equivariant limit
\begin{equation}
 \lim_{\e\to0}\hat{I}_{X_\bt} = 1 + n
 \sum_{d=1}^{\infty} (-1)^{nd}z^d
 \frac{\Ga(nd)} {\Ga(d+1)^n} \theta
 + O(\theta^2)
\end{equation}
we can read the mirror map
\begin{equation}
 \log\td{z} = \log z + n \sum_{d=0}^{\infty} (-1)^{nd}z^d
 \frac{\Ga(nd)} {\Ga(d+1)^n}~.
\end{equation}
The solutions to equivariant PF equations are labeled by fixed points
\begin{equation}
 \Pi(p_i) = z^{\frac{\e_i}{\lambda}}\sum_{d=0}^\infty ((-1)^n z)^d
 \frac{\left(\frac{\e_{n+1}+n\e_i}{\lambda}\right)_{nd}}
 {\prod_{j=1}^n\left(1-\frac{\e_j-\e_i}{\lambda}\right)_d}~,\quad i=1,\dots,n~.
\end{equation}

The case $n=2$ is discussed in \cref{sec:A1}. In the following
sub-sections we discuss the examples $n=3,4$ in more detail.

\subsubsection{\texorpdfstring{$\co(-3)$}{O(-3)} over
\texorpdfstring{$\BP^{2}$}{P2}}

For $n=3$ the non-equivariant $\hat{I}$-operator can be expanded as
\begin{equation}
 \lim_{\e\to0} \hat{I}_{X_\bt} = \sum_{d=0}^\infty (-z)^d
 \frac{(3\theta)_{3d}}{(1+\theta)_d^3}
 = 1 + 3 G^{(0)} \theta + 3 G^{(1)} \theta^2
 + \tfrac32 (G^{(2)}-\pi^2 G^{(0)}) \theta^3 + \dots~,
\end{equation}
where we define the functions
\begin{equation}
 G^{(i)}(z) := \sum_{d=1}^\infty(-z)^d
 \left(\tfrac{\partial}{\partial d}\right)^i \frac{\Ga(3d)}{\Ga(d+1)^3}~.
\end{equation}
The solutions to the non-equivariant PF equations are
\begin{equation}
\begin{aligned}
 \Pi(\pt) &= 1 ~,\\
 \Pi(\BP^1) &= \log z +3G^{(0)} = \log\td{z} ~, \\
 \Pi(\BP^2) &= \tfrac12\log^2\td{z} -3\left(\tfrac32 (G^{(0)})^2-G^{(1)}\right) 
~.
\end{aligned}
\end{equation}
The modified PF operator
\begin{equation}
 \cD_1 \cD_2 \cD_3 \cD_4 - z
 \cD_4(\cD_4+\lambda)(\cD_4+2\lambda)(\cD_4+3\lambda)
\end{equation}
admits an additional cubic solution associated to the regularized fundamental
cycle,
\begin{equation}
 \Pi_\reg(X_\bt) =
 - \tfrac{1}{18} \log^3\td{z}
 + \log\td{z} \left(\tfrac32 (G^{(0)})^2 - G^{(1)} \right)
 - 2\left(
 \tfrac32 (G^{(0)})^3
 - \tfrac32 G^{(0)} G^{(1)}
 + \tfrac14 G^{(2)}
 - \tfrac{\pi^2}{4} G^{(0)}
 \right) ~.
\end{equation}
Moreover,
\begin{equation}
\label{eq:O(-3)quadratic-sol}
 (-\lambda)^2 \lim_{\e\to0} \hat{I}_{X_\bt}\cD_4\cf_{\Ga}
 = \Pi(\BP^2) + \pi^2\Pi(\pt)~,
\end{equation}
\begin{equation}
 (-\lambda) \lim_{\e\to0} \hat{I}_{X_\bt}\cD_4\cD_i\cf_{\Ga}
 = \Pi(\BP^1)~,
 \quad
 i=1,2,3 ~,
\end{equation}
\begin{equation}
 \lim_{\e\to0} \hat{I}_{X_\bt}\cD_4\cD_i\cD_j\cf_{\Ga}
 = \Pi(\pt)~,
 \quad
 i,j = 1,2,3
\end{equation}
and we have
\begin{equation}
\label{eq:O(-3)cubic-sol}
 (-\lambda)^3 \lim_{\e\to0} \cf^D_\reg
 = \Pi_\reg(X_\bt) - \tfrac{\pi^2}{3}\Pi(\BP^1)~.
\end{equation}
The GV invariants $n_d$ can be read by matching \cref{eq:O(-3)quadratic-sol} to
\cref{eq:CY3-t2}, i.e.\
\begin{equation}
 \tfrac32 (G^{(0)})^2 - G^{(1)}
 = \sum_{d=1}^\infty d\, n_d(\BP^2) \Li_2(\td{z}^d)
\end{equation}
or equivalently by matching $\Pi_\reg(X_\bt)$ to
\cref{eq:localCY3-t3}, i.e.\
\begin{equation}
 \tfrac32 (G^{(0)})^3
 -\tfrac32 G^{(0)} G^{(1)}
 +\tfrac14 G^{(2)}
 -\tfrac{\pi^2}{4} G^{(0)}
 = \sum_{d=1}^\infty n_d(\BP^2) \Li_3(\td{z}^d)~,
\end{equation}
which give the same numbers as in ref.~\cite[Table~1]{Chiang:1999tz}.
The only singular contributions are the classical ones, therefore all GV
invariants are uniquely defined and the only ambiguity is in the regularization
of the classical intersection numbers.

\subsubsection{\texorpdfstring{$\co(-4)$}{O(-4)} over
\texorpdfstring{$\BP^{3}$}{P3}}

For $n=4$ the non-equivariant $\hat{I}$-operator can be expanded as
\begin{multline}
 \lim_{\e\to0} \hat{I} = \sum_{d=0}^\infty z^d
 \frac{(4\theta)_{4d}}{(1+\theta)_d^4}
 = 1 + 4 G_0 \theta + 4 G_1 \theta^2
 + 2 (G_2 - 2 \pi^2 G_0) \theta^3 \\
 + \left( \tfrac23 G_3 - 4 \pi^2 G_1 + 80 \zeta(3) G_0 \right)~,
 \theta^4 + \dots
\end{multline}
where we define the functions
\begin{equation}
 G^{(i)}(z) := \sum_{d=1}^\infty z^d \left(\tfrac{\partial}{\partial d}\right)^i
 \frac{\Ga(4d)}{\Ga(d+1)^4}~.
\end{equation}
The solutions to the non-equivariant PF equations are
\begin{equation}
\begin{aligned}
 \Pi(\pt) &= 1 \\
 \Pi(\BP^1) &= \log z + 4G^{(0)} = \log\td{z} \\
 \Pi(\BP^2) &= \tfrac12\log^2\td{z} + 4 G^{(1)} - 8 (G^{(0)})^2 \\
 \Pi(\BP^3) &= \tfrac16\log^3\td{z} + \left(4G^{(1)}-8(G^{(0)})^2\right) 
\log\td{z}
 +\tfrac{64}{3}(G^{(0)})^3
 -16 G^{(0)} G^{(1)}
 -4 \pi^2 G^{(0)}
 +2 G^{(2)}
\end{aligned}
\end{equation}
The modified PF operator
\begin{equation}
 \cD_1 \cD_2 \cD_3 \cD_4 \cD_5 - z
 \cD_5(\cD_5+\lambda)(\cD_5+2\lambda)(\cD_5+3\lambda)(\cD_5+4\lambda)
\end{equation}
admits an additional quartic solution associated
to the regularized fundamental cycle
\begin{multline}
 \Pi_\reg(X_\bt) = -\tfrac1{96}\log^4\td{z}
 +\left((G^{(0)})^2-\tfrac12 G^{(1)}\right) \log^2\td{z} \\
 +\left(-\tfrac{16}{3}(G^{(0)})^3+4G^{(0)}G^{(1)}+\pi^2G^{(0)}
 -\tfrac12 G^{(2)}\right) \log\td{z}\\
 +\left(8(G^{(0)})^4-8(G^{(0)})^2G^{(1)}-4\pi^2(G^{(0)})^2+2G^{(0)}G^{(2)}
 +\pi^2G^{(1)}-\tfrac16 G^{(3)}-20G^{(0)}\zeta(3)\right)~.
\end{multline}
Moreover,
\begin{equation}
 (-\lambda)^4 \lim_{\e\to0} \cf^D_\reg
 = \Pi_\reg(X_\bt) -\tfrac{5\pi^2}{12}\Pi(\BP^2) + 5\zeta(3)\Pi(\BP^1)~,
\end{equation}
\begin{equation}
 (-\lambda)^3 \lim_{\e\to0} \hat{I}_{X_\bt}\cD_5\cf_{\Ga}
 = \Pi(\BP^3) + \tfrac{5\pi^2}{3} \Pi(\BP^1) - 20 \zeta (3) \Pi(\pt)~,
\end{equation}
\begin{equation}
\label{eq:O(-4)quadratic-sol}
 (-\lambda)^2 \lim_{\e\to0} \hat{I}_{X_\bt}\cD_5\cD_i\cf_{\Ga}
 = \Pi(\BP^2) + \tfrac{5\pi^2}{3} \Pi(\pt)~,
 \quad
 i=1,2,3,4~.
\end{equation}
The GV invariants $n_d$ can be read by matching \cref{eq:O(-4)quadratic-sol} to
\cref{eq:CY4-t2}, i.e.\
\begin{equation}
 2(G^{(0)})^2 - G^{(1)} = \sum_{d=1}^\infty n_d(\BP^2) \Li_2(\td{z}^d)~,
\end{equation}
which gives the same numbers as in ref.~\cite[Table~1]{Klemm:2007in}.
Instantons are non-singular in this case and
GV invariants are uniquely defined.

\subsection{\texorpdfstring{$\co(-2,-2)$}{O(-2,-2)} over
 \texorpdfstring{$\BP^1 \times \BP^1$}{P1xP1}}

We consider $X_\bt = \tot K_{F_0}$, the canonical bundle of
the Hirzebruch surface $F_0$, realized as the quotient of
$\BC^5$ by $U(1)^2$ with charge matrix
\begin{equation}
Q =
\left(
\begin{array}{ccccc}
 1 & 1 & 0 & 0 & -2 \\
 0 & 0 & 1 & 1 & -2 \\
\end{array}
\right)~.
\end{equation}
The chamber is chosen so that $t^1,t^2>0$.
The K-theoretic disk function is defined as
\begin{equation}
 \cz^{D}(\bT,q;\q) =
 \oint_\qjk \frac{\dif w_1 \dif w_2} {(2\pi \ii)^2 w_1 w_2}
 \frac{w_1^{-T^1}w_2^{-T^2}}
 {(q_1w_1;\q)_\infty (q_2w_1;\q)_\infty (q_3w_2;\q)_\infty (q_4w_2;\q)_\infty
  (q_5w_1^{-2}w_2^{-2};\q)_\infty}
\end{equation}
with poles for $(w_1,w_2)$ in the set
\begin{equation}
\left\{
(q_1^{-1}\q^{-d_1},q_3^{-1}\q^{-d_2}),
(q_1^{-1}\q^{-d_1},q_4^{-1}\q^{-d_2}),
(q_2^{-1}\q^{-d_1},q_3^{-1}\q^{-d_2}),
(q_2^{-1}\q^{-d_1},q_4^{-1}\q^{-d_2}) | d_1, d_2 \in \BN \right\}~.
\end{equation}
The function $\cz^D$ satisfies the quantum K-theory relations
\begin{equation}
\begin{aligned}
 \left[ (1-\Delta_1)(1-\Delta_2) - \q^{T^1} (1-\Delta_5)(1-\q\Delta_5) \right]
 \cz^{D}(\bT,q;\q) = 0~, \\
 \left[ (1-\Delta_3)(1-\Delta_4) - \q^{T^2} (1-\Delta_5)(1-\q\Delta_5) \right]
 \cz^{D}(\bT,q;\q) = 0~,
\end{aligned}
\end{equation}
whose formal solution is
\begin{equation}
 \cz^D = \left[
 \sum_{d_1,d_2=0}^\infty \q^{d_1T^1+d_2T^2}
 \frac{(\Delta_5;\q)_{2d_1+2d_2}}
 { 
(\q^{-d_1}\Delta_1;\q)_{d_1}(\q^{-d_1}\Delta_2;\q)_{d_1}(\q^{-d_2}\Delta_3;\q)_{
d_2}(\q^{-d_2}\Delta_4;\q)_{d_2} }
 \right] \cz_{\Ga_\q}
\end{equation}
with
\begin{equation}
 \cz_{\Ga_\q} (\bT,q) =
  c_{1,3}q_1^{T^1}q_3^{T^2}
 +c_{1,4}q_1^{T^1}q_4^{T^2}
 +c_{2,3}q_2^{T^1}q_3^{T^2}
 +c_{2,4}q_2^{T^1}q_4^{T^2}
\end{equation}
and initial data
\begin{equation}
 c_{1,3} = \frac{1}{(\q;\q)_\infty(q_2q_1^{-1};\q)_\infty(q_4q_3^{-1};\q)_\infty
 (q_5q_1^2q_3^2;\q)_\infty}~,
\end{equation}
\begin{equation}
 c_{1,4} = \frac{1}{(\q;\q)_\infty(q_2q_1^{-1};\q)_\infty(q_3q_4^{-1};\q)_\infty
 (q_5q_1^2q_4^2;\q)_\infty}~,
\end{equation}
\begin{equation}
 c_{2,3} = \frac{1}{(\q;\q)_\infty(q_1q_2^{-1};\q)_\infty(q_4q_3^{-1};\q)_\infty
 (q_5q_2^2q_3^2;\q)_\infty}~,
\end{equation}
\begin{equation}
 c_{2,4} = \frac{1}{(\q;\q)_\infty(q_1q_2^{-1};\q)_\infty(q_3q_4^{-1};\q)_\infty
 (q_5q_2^2q_4^2;\q)_\infty}~.
\end{equation}
The total space $\co(-2,-2) \to \BP^1 \times \BP^1$ has a compact divisor $D_5$
corresponding to the zero section (i.e.\ the base $\BP^1 \times \BP^1$)
and therefore we have that
\begin{equation}
 (1-\Delta_5) \cz^D(\bT,q;\q) \quad\text{is analytic at $q_i=1$}~.
\end{equation}

In the cohomological limit $\hbar\to0$ we have
\begin{equation}
 \cf^D = \left[
 \sum_{d_1,d_2=0}^\infty z_1^{d_1}z_2^{d_2}
 \frac{\left(\tfrac{\cD_5}{\lambda}\right)_{2d_1+2d_2}}
 {\left(1-\tfrac{\cD_1}{\lambda}\right)_{d_1}
  \left(1-\tfrac{\cD_2}{\lambda}\right)_{d_1}
  \left(1-\tfrac{\cD_3}{\lambda}\right)_{d_2}
  \left(1-\tfrac{\cD_4}{\lambda}\right)_{d_2}}
 \right] \cf_{\Ga}
\end{equation}
with
\begin{equation}
 \cf_{\Ga} (\bt,\e) =
  c_{1,3}\eu^{-\e_1 t^1-\e_3 t^2}
 +c_{1,4}\eu^{-\e_1 t^1-\e_4 t^2}
 +c_{2,3}\eu^{-\e_2 t^1-\e_3 t^2}
 +c_{2,4}\eu^{-\e_2 t^1-\e_4 t^2}
\end{equation}
and semi-classical data
\begin{equation}
 c_{1,3} = \lambda^{-3}\Ga\left(\tfrac{\e_2-\e_1}{\lambda}\right)
 \Ga\left(\tfrac{\e_4-\e_3}{\lambda}\right)
 \Ga\left(\tfrac{\e_5+2\e_1+2\e_3}{\lambda}\right)~,
\end{equation}
\begin{equation}
 c_{1,4} = \lambda^{-3}\Ga\left(\tfrac{\e_2-\e_1}{\lambda}\right)
 \Ga\left(\tfrac{\e_3-\e_4}{\lambda}\right)
 \Ga\left(\tfrac{\e_5+2\e_1+2\e_4}{\lambda}\right)~,
\end{equation}
\begin{equation}
 c_{2,3} = \lambda^{-3}\Ga\left(\tfrac{\e_1-\e_2}{\lambda}\right)
 \Ga\left(\tfrac{\e_4-\e_3}{\lambda}\right)
 \Ga\left(\tfrac{\e_5+2\e_2+2\e_3}{\lambda}\right)~,
\end{equation}
\begin{equation}
 c_{2,4} = \lambda^{-3}\Ga\left(\tfrac{\e_1-\e_2}{\lambda}\right)
 \Ga\left(\tfrac{\e_3-\e_4}{\lambda}\right)
 \Ga\left(\tfrac{\e_5+2\e_2+2\e_4}{\lambda}\right)~.
\end{equation}
The disk function $\cf^D$ is annihilated by the equivariant PF
operators
\begin{equation}
\begin{aligned}
 \PF^\eq_{(1,0)} &= \cD_1\cD_2-z_1\cD_5(\cD_5+\lambda)~, \\
 \PF^\eq_{(0,1)} &= \cD_3\cD_4-z_2\cD_5(\cD_5+\lambda) ~,
\end{aligned}
\end{equation}
which shows that the instanton operators
\begin{equation}
 \sfp_{d_1,d_2} = z_1^{d_1}z_2^{d_2}
 \frac{\left(\frac{\cD_5}{\lambda}\right)_{2d_1+2d_2}}
 {\left(1-\frac{\cD_1}{\lambda}\right)_{d_1}
 \left(1-\frac{\cD_2}{\lambda}\right)_{d_1}
 \left(1-\frac{\cD_3}{\lambda}\right)_{d_2}
 \left(1-\frac{\cD_4}{\lambda}\right)_{d_2}}
\end{equation}
are proportional to the compact divisor $\cD_5$ if $(d_1,d_2)\neq(0,0)$,
and therefore that the instantons are all regular in the non-equivariant limit.

The solutions to the equivariant PF equations are
\begin{equation}
 \Pi(p_{ij},\e) = \sum_{d_1,d_2=0}^\infty z_1^{d_1}z_2^{d_2}
 \frac{\left(\frac{\e_5+2\e_i+2\e_j}{\lambda}\right)_{2d_1+2d_2}}
 {\prod_{k=1}^2\left(1-\frac{\e_k-\e_i}{\lambda}\right)_{d_1}
  \prod_{l=3}^4\left(1-\frac{\e_l-\e_j}{\lambda}\right)_{d_2}}
 \eu^{-\e_i t^1-\e_j t^2},
 \quad
 i=1,2\quad j=3,4.
\end{equation}
The non-equivariant $\hat{I}$-operator expands as
\begin{multline}
 \lim_{\e\to0}\hat{I}_{X_\bt} = 1
 + 2G^{(00)}(\theta_1+\theta_2)
 + 2G^{(10)}\theta_1^2
 + 2G^{(01)}\theta_2^2
 + 2(G^{(10)}+G^{(01)})\theta_1\theta_2 \\
 + (2G^{(11)}+G^{(20)}-\tfrac{5\pi^2}3 G^{(00)}) \theta_1^2\theta_2
 + (2G^{(11)}+G^{(02)}-\tfrac{5\pi^2}3 G^{(00)}) \theta_1\theta_2^2 \\
 + (G^{(20)}-\tfrac{\pi^2}3 G^{(00)}) \theta_1^3
 + (G^{(02)}-\tfrac{\pi^2}3 G^{(00)}) \theta_2^3
 + \dots~,
\end{multline}
where we define
\begin{equation}
 G^{(ij)}(z_1,z_2) = \sum_{(d_1,d_2)\neq(0,0)}z_1^{d_1}z_2^{d_2}
 \partial_{d_1}^i\partial_{d_2}^j
 \frac{\Ga(2d_1+2d_2)}{\Ga(d_1+1)^2\Ga(d_2+1)^2}~.
\end{equation}
The solutions to the non-equivariant PF equations are
\begin{equation}
\begin{aligned}
 \Pi(\pt) &= 1 ~,\\
 \Pi(C^1) &= \log z_1 + 2 G^{(00)} = \log\td{z}_1 ~,\\
 \Pi(C^2) &= \log z_2 + 2 G^{(00)} = \log\td{z}_2 ~,\\
 \Pi(\BP^1\times\BP^1) &= \log\td{z}_1 \log\td{z}_2
 - 4 (G^{(00)})^2 + 2 G^{(01)} + 2 G^{(10)}~,
\end{aligned}
\end{equation}
where $C^1$ and $C^2$ are the homology two-cycles corresponding to the two 
$\BP^1$'s.
The modified PF operators
\begin{equation}
\begin{aligned}
 &\cD_1\cD_2\cD_5-z_1\cD_5(\cD_5+\lambda)(\cD_5+2\lambda) ~,\\
 &\cD_3\cD_4\cD_5-z_2\cD_5(\cD_5+\lambda)(\cD_5+2\lambda)
\end{aligned}
\end{equation}
allow us to define the regularized cubic solution
\begin{equation}
\begin{aligned}
 \Pi_\reg(X_\bt) =& \tfrac1{24}\log^3\td{z}_1
 -\tfrac18\log^2\td{z}_1\log\td{z}_2
 -\tfrac18\log\td{z}_1\log^2\td{z}_2
 +\tfrac{1}{24}\log^3\td{z}_2 \\
 & +\left((G^{(00)})^2-G^{(01)}\right) \log\td{z}_1
 +\left((G^{(00)})^2-G^{(10)}\right) \log\td{z}_2 \\
 & -2\left(
 \tfrac43 (G^{(00)})^3
 - G^{(00)} G^{(01)}
 - G^{(00)} G^{(10)}
 - \tfrac{\pi^2}{3} G^{(00)}
 + \tfrac12 G^{(11)}
 \right)~.
\end{aligned}
\end{equation}
Then we have,
\begin{equation}
\label{eq:FregLocalP1P1}
 \lim_{\e\to0} (-\lambda)^3 \cf^D_\reg
 = \Pi_\reg(X_\bt)
 + \tfrac{\alpha}{12} \left(\Pi(C^1)-\Pi(C^2)\right)
 \left(8\pi^2+\left(16\alpha^2-3\right)\left(\Pi(C^1)-\Pi(C^2)\right)^2\right)~,
\end{equation}
where $\Pi(C^1)-\Pi(C^2)$ is annihilated by the compact divisor operator $\cD_5$
and $\alpha$ parametrizes the intrinsic ambiguity in the choice of
left-inverse $R^j_a$,
\begin{equation}
 R=
\left(
\begin{array}{ccc}
 \alpha-1/4 & -\alpha-1/4
\end{array}
\right)~.
\end{equation}
Changing the value of $\alpha$ changes the semi-classical data in $\cf^D_\reg$
but it leaves the instanton part of the solution unchanged, therefore the GV
invariants do not depend on this choice.

Observing that
\begin{equation}
 \lim_{\e\to0} (-\lambda)^2 \hat{I}_{X_\bt}\cD_5\cf_{\Ga}
 = \Pi(\BP^1\times\BP^1) + \tfrac{2\pi^2}{3}\Pi(\pt)
\end{equation}
we can match with \cref{eq:CY3-t2} to read the GV invariants $n_{d_1,d_2}$, 
namely
\begin{equation}
 -4 (G^{(00)})^2 + 2 G^{(01)} + 2 G^{(10)}
 = \sum_{(d_1,d_2)\neq(0,0)} (-2d_1-2d_2)
 n_{\bd}(\BP^1\times\BP^1) \Li_2(\td{z}_1^{d_1}\td{z}_2^{d_2})~,
\end{equation}
which reproduce the results of ref.~\cite[Table~9]{Chiang:1999tz}.
We can also match \cref{eq:FregLocalP1P1,eq:localCY3-t3}
\begin{equation}
 \tfrac43 (G^{(00)})^3
 - G^{(00)} G^{(01)}
 - G^{(00)} G^{(10)}
 - \tfrac{\pi^2}{3} G^{(00)}
 + \tfrac12 G^{(11)}
 = \sum_{(d_1,d_2)\neq(0,0)}
 n_{\bd}(\BP^1\times\BP^1) \Li_3(\td{z}_1^{d_1}\td{z}_2^{d_2})~,
\end{equation}
which gives the same GV numbers.
Comparing to ref.~\cite{Honma:2013hma}, the redefinition of Euler's constant 
$\gamma$
amounts in our setup to multiplying by a factor of $\eu^{\e_5(\gamma - h(z))}$
in the shift equation. A similar remark applies to all other cases.

\subsection{\texorpdfstring{$SU(3)_0$}{SU(3)level0} geometry}

Consider the Calabi-Yau three-fold $X_{\bt}$
given by the quotient of $\BC^6$ by $U(1)^3$ with
\begin{equation}
Q =
\begin{pmatrix}
 1 & 1 & 1 & -3 &  0 & 0 \\
 0 & 0 & 1 & -2 & 1  & 0 \\
 0 & 0 & 0 &  1 & -2 & 1
\end{pmatrix}
\end{equation}
and chamber $t^1 > t^2 > 0$, $t^3 > 0$. This CY geometry corresponds
to a 5d gauge theory with $SU(3)$ gauge group and zero Chern-Simons level.
This manifold has two compact toric divisors $D_4$ and $D_5$.
The disk function is defined as
\begin{multline}
 \cf^D (\bt,\e;\lambda) = \lambda^{-6} \oint_\qjk
 \frac {\dif\phi_1\dif\phi_2\dif\phi_3} {(2\pi\ii)^3}
 \eu^{\phi_1 t^1+\phi_2 t^2+\phi_3 t^3}
 \Ga\left(\tfrac{\e_1 + \phi_1}{\lambda}\right)
 \Ga\left(\tfrac{\e_2 + \phi_1}{\lambda}\right) \\
 \Ga\left(\tfrac{\e_3 + \phi_1+ \phi_2}{\lambda}\right)
 \Ga\left(\tfrac{\e_4 -3\phi_1 -2\phi_2 + \phi_3}{\lambda}\right)
 \Ga\left(\tfrac{\e_5 +\phi_2 - 2\phi_3}{\lambda}\right)
 \Ga\left(\tfrac{\e_6 + \phi_3}{\lambda}\right)
\end{multline}
and the poles are located at $(\phi_1,\phi_2,\phi_3)$ equal to
\begin{equation}
\begin{tabu}{r|l}
(-\e_1-k_1\lambda, -\e_5-2\e_6-(k_2+2k_3)\lambda, -\e_6-k_3\lambda) & (1,5,6) 
~,\\
(-\e_2-k_1\lambda, -\e_5-2\e_6-(k_2+2k_3)\lambda, -\e_6-k_3\lambda) & (2,5,6) 
~, \\
(-\e_1-k_1\lambda, \e_1 - \e_3-(-k_1+k_2)\lambda, - \e_1 - 
2\e_3-\e_4-(k_1+2k_2+k_3)\lambda) & (1,3,4) ~, \\
(-\e_2-k_1\lambda, \e_2 - \e_3-(-k_1+k_2)\lambda, - \e_2 - 
2\e_3-\e_4-(k_1+2k_2+k_3)\lambda) & (2,3,4) ~,\\
(-\e_1-k_1\lambda, -\e_3 +\e_1-(-k_1+k_2)\lambda, -\e_6-k_3\lambda) & (1,3,6) 
~,\\
(-\e_2-k_1\lambda, -\e_3 +\e_2-(-k_1+k_2)\lambda, -\e_6-k_3\lambda) & (2,3,6) 
~. \\
\end{tabu}
\end{equation}
The equivariant cohomology ring of $X_{\bt}$ is
\begin{equation}
 H^\bullet_{\sft}(X_{\bt}) \cong \BC[\phi_1,\phi_2,\phi_3,\e_1,\ldots,\e_6] /
 \langle x_1 x_2, x_3 x_5, x_3 x_6, x_4 x_6 \rangle
\end{equation}
and the quantum cohomology relations are encoded in the equivariant PF operators
\begin{equation}
\begin{aligned}
 \PF^\eq_{(1,-1,0)} &= \cD_1\cD_2 - \eu^{-\lambda(t^1-t^2)}\cD_4\cD_5 ~,\\
 \PF^\eq_{(0,1,0)} &= \cD_3\cD_5 - \eu^{-\lambda t^2}\cD_4(\cD_4+\lambda) ~, \\
 \PF^\eq_{(0,1,1)} &= \cD_3\cD_6 - \eu^{-\lambda(t^2+t^3)}\cD_4\cD_5 ~,\\
 \PF^\eq_{(0,0,1)} &= \cD_4\cD_6 - \eu^{-\lambda t^3}\cD_5(\cD_5+\lambda) ~,
\end{aligned}
\end{equation}
whose generic solution is
\begin{equation}
 \label{eq:SU3-solution}
 \cf^{D} =
 \sum_{ \substack{ d_1\geq0,\\d_2\geq-d_1,\\d_3\geq0 } }
 z_1^{d_1}z_2^{d_2}z_3^{d_3} \,
 \left(\tfrac{\cD_1}{\lambda}\right)_{-d_1}
 \left(\tfrac{\cD_2}{\lambda}\right)_{-d_1}
 \left(\tfrac{\cD_3}{\lambda}\right)_{-d_1-d_2}
 \left(\tfrac{\cD_4}{\lambda}\right)_{3d_1+2d_2-d_3}
 \left(\tfrac{\cD_5}{\lambda}\right)_{-d_2+2d_3}
 \left(\tfrac{\cD_6}{\lambda}\right)_{-d_3}
 \, \cf_{\Ga}
\end{equation}
with semi-classical data
\begin{equation}
\begin{aligned}
 \cf_{\Ga}(\bt,\e) =&
  c_{1,5,6} \, \eu^{ - \e_1 t^1 - (\e_5+2\e_6) t^2 - \e_6 t^3 }
  +c_{2,5,6} \, \eu^{ - \e_2 t^1 - (\e_5+2\e_6) t^2 - \e_6 t^3 } \\
 &+c_{1,3,4} \, \eu^{ - \e_1 t_1 - (-\e_1+\e_3) t^2 - (\e_1+2\e_3+\e_4) t^3 }
  +c_{2,3,4} \, \eu^{ - \e_2 t_1 - (-\e_2+\e_3) t^2 - (\e_2+2\e_3+\e_4) t^3 } \\
 &+c_{1,3,6} \, \eu^{ - \e_1 t^1 - (-\e_1+\e_3) t^2 - \e_6 t^3 }
  +c_{2,3,6} \, \eu^{ - \e_2 t^1 - (-\e_2+\e_3) t^2 - \e_6 t^3 }~.
\end{aligned}
\end{equation}
The instanton sum in \cref{eq:SU3-solution} contains only positive powers of
$z_1$, $z_3$ but also negative powers of $z_2$. This is consistent with the
choice of chamber for the Kähler moduli $t^1 > t^2$. After a
change of coordinates in the Kähler cone given by the unimodular matrix
\begin{equation}
\begin{pmatrix}
 1 & -1 & 0 \\
 0 & 1 & 0 \\
 0 & 0 & 1
\end{pmatrix}
\in SL(3,\BZ)
\end{equation}
we can bring back the instanton sum to the standard cone $d_1,d_2,d_3\geq0$.
This choice of Kähler coordinates corresponds to the choice of transformed
charge matrix
\begin{equation}
Q =
\begin{pmatrix}
 1 & 1 & 0 & -1 & -1 & 0 \\
 0 & 0 & 1 & -2 & 1  & 0 \\
 0 & 0 & 0 &  1 & -2 & 1
\end{pmatrix}
\end{equation}
and the chamber is mapped to the region $t^1,t^2,t^3>0$, where we use the same
symbols for the new coordinates on the transformed Kähler cone.
This geometry corresponds to two $\BP^2$ connected by a
$\BP^1$ in one of the phases related by a flop transition
as described in ref.~\cite{Chiang:1999tz}.

The Givental $\hat{I}$-operator is
\begin{multline}
 \hat{I}_{X_\bt} = \sum_{ \substack{ d_1\geq0,\\ d_2\geq0,\\d_3\geq0 } }
 z_1^{d_1}z_2^{d_2}z_3^{d_3} \,
 \left(\tfrac{{\cD}_1}{\lambda}\right)_{-d_1}
 \left(\tfrac{{\cD}_2}{\lambda}\right)_{-d_1}
 \left(\tfrac{{\cD}_3}{\lambda}\right)_{-d_2}
 \left(\tfrac{{\cD}_4}{\lambda}\right)_{d_1+2d_2-d_3}
 \left(\tfrac{{\cD}_5}{\lambda}\right)_{d_1-d_2+2d_3}
 \left(\tfrac{{\cD}_6}{\lambda}\right)_{-d_3} \\
 = 1
 + \sum_{ \substack{ d_1+2d_2-d_3>0,\\ d_1-d_2+2d_3\leq0 } }
 (-z_1)^{d_1}z_2^{d_2}(-z_3)^{d_3} \,
 \tfrac{
 \left(\tfrac{{\cD}_4}{\lambda}\right)_{d_1+2d_2-d_3}
 }{
 \left(1-\tfrac{{\cD}_1}{\lambda}\right)_{d_1}
 \left(1-\tfrac{{\cD}_2}{\lambda}\right)_{d_1}
 \left(1-\tfrac{{\cD}_3}{\lambda}\right)_{d_2}
 \left(1-\tfrac{{\cD}_5}{\lambda}\right)_{-d_1+d_2-2d_3}
 \left(1-\tfrac{{\cD}_6}{\lambda}\right)_{d_3}
 } \\
 + \sum_{ \substack{ d_1+2d_2-d_3\leq0,\\ d_1-d_2+2d_3>0 } }
 (-z_1)^{d_1}(-z_2)^{d_2}z_3^{d_3} \,
 \tfrac{
 \left(\tfrac{{\cD}_5}{\lambda}\right)_{d_1-d_2+2d_3}
 }{
 \left(1-\tfrac{{\cD}_1}{\lambda}\right)_{d_1}
 \left(1-\tfrac{{\cD}_2}{\lambda}\right)_{d_1}
 \left(1-\tfrac{{\cD}_3}{\lambda}\right)_{d_2}
 \left(1-\tfrac{{\cD}_4}{\lambda}\right)_{-d_1-2d_2+d_3}
 \left(1-\tfrac{{\cD}_6}{\lambda}\right)_{d_3}
 } \\
 + \sum_{ \substack{ d_1+2d_2-d_3>0,\\ d_1-d_2+2d_3>0 } }
 z_1^{d_1}(-z_2)^{d_2}(-z_3)^{d_3} \,
 \tfrac{
 \left(\tfrac{{\cD}_4}{\lambda}\right)_{d_1+2d_2-d_3}
 \left(\tfrac{{\cD}_5}{\lambda}\right)_{d_1-d_2+2d_3}
 }{
 \left(1-\tfrac{{\cD}_1}{\lambda}\right)_{d_1}
 \left(1-\tfrac{{\cD}_2}{\lambda}\right)_{d_1}
 \left(1-\tfrac{{\cD}_3}{\lambda}\right)_{d_2}
 \left(1-\tfrac{{\cD}_6}{\lambda}\right)_{d_3}
 }~,
\end{multline}
where all instanton operators are proportional to at least one of the two
compact divisor operators $\cD_4,\cD_5$ except for $\sfp_{(0,0,0)}=1$,
hence the only singular contribution to the disk function comes from the
semi-classical part $\cf_{\Ga}$.

If we define the functions
\begin{equation}
\begin{aligned}
 L_1^{(ijk)}&:= \sum_{ \substack{ d_1+2d_2-d_3>0,\\d_1-d_2+2d_3\leq0 } }
 (-z_1)^{d_1}z_2^{d_2}(-z_3)^{d_3}
 \partial_{d_1}^i\partial_{d_2}^j\partial_{d_3}^k
 \tfrac{\Ga(d_1+2d_2-d_3)}
 {\Ga(d_1+1)^2\Ga(d_2+1)\Ga(-d_1+d_2-2d_3+1)\Ga(d_3+1)} ~, \\
 L_2^{(ijk)}&:= \sum_{ \substack{ d_1+2d_2-d_3\leq0,\\d_1-d_2+2d_3>0 } }
 (-z_1)^{d_1}(-z_2)^{d_2}z_3^{d_3}
 \partial_{d_1}^i\partial_{d_2}^j\partial_{d_3}^k
 \tfrac{\Ga(d_1-d_2+2d_3)}
 {\Ga(d_1+1)^2\Ga(d_2+1)\Ga(-d_1-2d_2+d_3+1)\Ga(d_3+1)} ~, \\
 L_3^{(ijk)}&:= \sum_{ \substack{ d_1+2d_2-d_3>0,\\d_1-d_2+2d_3>0 } }
 z_1^{d_1}(-z_2)^{d_2}(-z_3)^{d_3}
 \partial_{d_1}^i\partial_{d_2}^j\partial_{d_3}^k
 \tfrac{\Ga(d_1+2d_2-d_3)\Ga(d_1-d_2+2d_3)}
 {\Ga(d_1+1)^2\Ga(d_2+1)\Ga(d_3+1)}
\end{aligned}
\end{equation}
and $L_n\equiv L_n^{(000)}$,
we can write the solutions to the non-equivariant PF equations as
\begin{equation}
\begin{aligned}
 \Pi(\pt) =& 1 ~,\\
 \Pi(C^1) =& \log z_1 + L_1 + L_2 = \log \td{z}_1 ~,\\
 \Pi(C^2) =& \log z_2 + 2 L_1 - L_2 = \log \td{z}_2 ~,\\
 \Pi(C^3) =& \log z_3 - L_1 + 2 L_2 = \log \td{z}_3  ~, \\
 \Pi(D_4) =& \tfrac12\log^2\td{z}_2 + \log\td{z}_1\log\td{z}_2
 - 4L_1^2 + L_1 L_2 + \tfrac12 L_2^2
 + 2L_1^{(100)} -L_2^{(100)} +3L_1^{(010)} - L_3 ~, \\
 \Pi(D_5) =& \tfrac12\log^2\td{z}_3 + \log\td{z}_1\log\td{z}_3
 + \tfrac12 L_1^2 + L_1 L_2 - 4 L_2^2
 + 2L_2^{(100)} -L_1^{(100)} +3L_2^{(001)} - L_3 ~,
\end{aligned}
\end{equation}
where $C^a\in H_2(X_\bt)$ are such that $\int_{C^a}\phi_b=\delta^a_b$
and $D_4,D_5\in H_4(X_\bt)$ are the compact divisors.
Matching with \cref{eq:CY3-t2} we can read the GV invariants $n_{d_1,d_2,d_3}$
and we obtain the same result as ref.~\cite[Table~6]{Chiang:1999tz}.

The additional solution to the non-equivariant modified PF equations is
\begin{multline}
 \Pi_\reg(X_\bt) =
 - \tfrac13 \left( \log\td{z}_1\log^2\td{z}_2
 + \log\td{z}_1\log\td{z}_2\log\td{z}_3
 + \log\td{z}_1\log^2\td{z}_3 \right) \\
 + \tfrac16 \left( \log^2\td{z}_2\log\td{z}_3
 + \log\td{z}_2\log^2\td{z}_3 \right) \\
 + \left\{ L_1^2 - L_1 L_2 + L_2^2 + L_3
 - L_1^{(010)} - L_2^{(001)} \right\} \log\td{z}_1 \\
 + \left\{ \tfrac32 L_1^2- L_1^{(100)} - L_1^{(010)} \right\} \log\td{z}_2
 + \left\{ \tfrac32 L_2^2 - L_2^{(100)} - L_2^{(001)} \right\} \log\td{z}_3 \\
 + \Big\{ L_1^2 L_2 + L_1 L_2^2 - \tfrac83 (L_1^3+L_2^3)
 + (L_1+L_2) \left(\tfrac{2\pi^2}{3} - L_3\right) \\
 + L_1 (2L_1^{(100)}-L_2^{(100)}+3L_1^{(010)}-L_3)
 + L_2 (2L_2^{(100)}-L_1^{(100)}+3L_2^{(001)}-L_3) \\
 - \tfrac12 L_1^{(020)} - \tfrac12 L_2^{(002)}
 - L_1^{(110)} - L_2^{(101)} + L_3^{(100)}
 \Big\}~,
\end{multline}
which corresponds to
\begin{multline}
 (-\lambda)^3\lim_{\e\to0}\cf^D_\reg =
 \Pi_\reg(X_\bt) - \tfrac{2\pi^2}{3}\Pi(C^2) - \tfrac{2\pi^2}{3}\Pi(C^3) \\
 +\tfrac16 \left(8 \alpha ^3-3 \alpha ^2 \beta +6 \alpha ^2-3 \alpha
 \beta ^2-6 \alpha  \beta +8 \beta ^3+6 \beta ^2\right)
 \left( \Pi(C^1)-\Pi(C^2)-\Pi(C^3) \right)^3 \\
 +\tfrac{2\pi^2}{3} (\alpha+\beta) \left( \Pi(C^1)-\Pi(C^2)-\Pi(C^3) \right)~,
\end{multline}
where the combination $\Pi(C^1)-\Pi(C^2)-\Pi(C^3)$ is in the kernel of the
operators $\td{\cD}_4,\td{\cD}_5$, as differential operators in the mirror
variables $\td{t}^a$. Here $\alpha,\beta$ are arbitrary numbers that parametrize
the choice of left-inverse $R^j_a$,
\begin{equation}
 R=
\left(
\begin{array}{ccc}
 \alpha & -\alpha -2/3 & -\alpha -1/3 \\
  \beta &  -\beta -1/3 &  -\beta -2/3 \\
\end{array}
\right)~.
\end{equation}

The K-theoretic uplift of the disk function is
\begin{multline}
 \cz^D(\bT,q;\q) = -\oint_\qjk
 \frac {\dif w_1\dif w_2\dif w_3} {(2\pi\ii)^3w_1w_2w_3} \\
 \frac {w_1^{-T^1}w_2^{-T^2}w_3^{-T^3}}
 {(q_1w_1;\q)_\infty (q_2w_1;\q)_\infty (q_3w_1w_2;\q)_\infty
  (q_4w_1^{-3}w_2^{-2}w_3;\q)_\infty (q_5w_2w_3^{-2};\q)_\infty
  (q_6w_3;\q)_\infty}
\end{multline}
satisfying the quantum K-theory relations
\begin{equation}
\begin{aligned}
 \left[ (1-\Delta_1)(1-\Delta_2) - \q^{T^1-T^2} (1-\Delta_4)(1-\Delta_5) \right]
 \cz^D(\bT,q;\q) = 0 ~,\\
 \left[ (1-\Delta_3)(1-\Delta_5) - \q^{T^2} (1-\Delta_4)(1-\q\Delta_4) \right]
 \cz^D(\bT,q;\q) = 0 ~,\\
 \left[ (1-\Delta_3)(1-\Delta_6) - \q^{T^2+T^3} (1-\Delta_4)(1-\Delta_5) \right]
 \cz^D(\bT,q;\q) = 0 ~,\\
 \left[ (1-\Delta_4)(1-\Delta_6) - \q^{T^3} (1-\Delta_5)(1-\q\Delta_5) \right]
 \cz^D(\bT,q;\q) = 0 ~.
\end{aligned}
\end{equation}
The K-theoretic $I$-function operator
\begin{multline}
 \hat{I}^K_{X_\bt} =
 \sum_{ \substack{ d_1\geq0,\\d_2\geq-d_1,\\d_3\geq0 } }
 \q^{d_1T^1+d_2T^2+d_3T^3}
  \times
 (\Delta_1;\q)_{-d_1}
 (\Delta_2;\q)_{-d_1}
 (\Delta_3;\q)_{-d_1-d_2}
 \times \\
 \times
 (\Delta_4;\q)_{3d_1+2d_2-d_3}
 (\Delta_5;\q)_{-d_2+2d_3}
 (\Delta_6;\q)_{-d_3}~,
\end{multline}
creates a solution to PF equations when acting on semi-classical data
\begin{equation}
\begin{aligned}
 \cz_{\Ga_\q}(\bT,q) =&
  +c_{1,5,6} \, q_1^{T^1}q_5^{T^2}q_6^{2T^2+T^3}
  +c_{2,5,6} \, q_2^{T^1}q_5^{T^2}q_6^{2T^2+T^3} \\
 &+c_{1,3,4} \, q_1^{T^1-T^2+T^3}q_3^{T^2+2T^3}q_4^{T^3}
  +c_{2,3,4} \, q_2^{T^1-T^2+T^3}q_3^{T^2+2T^3}q_4^{T^3} \\
 &+c_{1,3,6} \, q_1^{T^1-T^2}q_3^{T^2}q_6^{T^3}
  +c_{2,3,6} \, q_2^{T^1-T^2}q_3^{T^2}q_6^{T^3}~.
\end{aligned}
\end{equation}

\subsection{Local \texorpdfstring{$F_2$}{F2}} \label{ex:f2}

We consider the toric quotient $X_\bt = \tot K_{F_2}$
corresponding to the canonical bundle of the Hirzebruch surface $F_2$. 
This local CY geometry is defined by the charge matrix
\begin{equation}
Q =
\begin{pmatrix}
 1 & 1 & -2 &  0 & 0 \\
 0 & 0 &  1 & -2 & 1
\end{pmatrix}
\end{equation}
and chamber $t^1,t^2> 0$. The total space of the line bundle has one compact
toric divisor $D_4$ corresponding to the base $F_2$.

The disk function is defined by the integral
\begin{equation}
 \cf^D = \lambda^{-5} \oint_\qjk
 \frac {\dif\phi_1 \dif\phi_2} {(2\pi\ii)^2}
 \eu^{\phi_1 t^1+\phi_2 t^2}
 \Ga \left( \tfrac {\e_1 + \phi_1} {\lambda} \right)
 \Ga \left( \tfrac {\e_2 + \phi_1} {\lambda} \right)
 \Ga \left( \tfrac {\e_3 -2\phi_1 + \phi_2} {\lambda} \right)
 \Ga \left( \tfrac {\e_4 -2\phi_2} {\lambda} \right)
 \Ga \left( \tfrac {\e_5 +\phi_2} {\lambda} \right)
\end{equation}
with classical poles in the set
\begin{equation}
 \jk = \{(1,3),(1,5),(2,3),(2,5)\}
\end{equation}
and quantum poles at the values of $(\phi_1,\phi_2)$ equal to
\begin{equation}
\begin{tabu}{l}
(-\e_1 - k_1, -\e_3 - 2\e_1 - 2k_1 - k_2)~, \\
(-\e_2 - k_1, -\e_3 - 2\e_2 - 2k_1 - k_2)~, \\
(-\e_1 - k_1, -\e_5 - k_2)~, \\
(-\e_2 - k_1, -\e_5 - k_2)~.
\end{tabu}
\end{equation}
The equivariant cohomology ring is given by the quotient of
$\BC[\phi_1,\phi_2,\e_1,\dots,\e_5]$ by the ideal
\begin{equation}
 I_\mathrm{SR} = \langle x_1 x_2, x_3 x_5 \rangle~.
\end{equation}
The equivariant PF operators are then defined as
\begin{equation}
\begin{aligned}
 \PF^\eq_{(1,0)} = \cD_1 \cD_2 - \eu^{-\lambda t^1} (\lambda + \cD_3) \cD_3~, \\
 \PF^\eq_{(0,1)} = \cD_3 \cD_5 - \eu^{-\lambda t^2} (\lambda+ \cD_4) \cD_4~.
\end{aligned}
\end{equation}
The $\hat{I}$-operator is
\begin{equation}
\begin{aligned}
 \hat{I}_{X_\bt}=&
 \sum_{d_1,d_2\geq0}^\infty z_1^{d_1}z_2^{d_2}
 \left(\tfrac{\cD_{1}}{\lambda}\right)_{-d_1}
 \left(\tfrac{\cD_{2}}{\lambda}\right)_{-d_1}
 \left(\tfrac{\cD_{3}}{\lambda}\right)_{2d_1-d_2}
 \left(\tfrac{\cD_{4}}{\lambda}\right)_{2d_2}
 \left(\tfrac{\cD_{5}}{\lambda}\right)_{-d_2} \\
 =& 1 + \sum_{d_1=1}^\infty z_1^{d_1}
 \frac
 {\left(\frac{\cD_3}{\lambda}\right)_{2d_1}}
 {\left(1-\frac{\cD_1}{\lambda}\right)_{d_1}
  \left(1-\frac{\cD_2}{\lambda}\right)_{d_1}} \\
 &
 +\sum_{\substack{2d_1-d_2\leq0\\d_2>0}} z_1^{d_1}z_2^{d_2}
 \frac
 {\left(\frac{\cD_4}{\lambda}\right)_{2d_2}}
 {\left(1-\frac{\cD_1}{\lambda}\right)_{d_1}
  \left(1-\frac{\cD_2}{\lambda}\right)_{d_1}
  \left(1-\frac{\cD_3}{\lambda}\right)_{-2d_1+d_2}
  \left(1-\frac{\cD_5}{\lambda}\right)_{d_2}} \\
 &
 +\sum_{\substack{2d_1-d_2>0\\d_2>0}} z_1^{d_1}(-z_2)^{d_2}
 \frac
 {\left(\frac{\cD_3}{\lambda}\right)_{2d_1-d_2}
  \left(\frac{\cD_4}{\lambda}\right)_{2d_2}}
 {\left(1-\frac{\cD_1}{\lambda}\right)_{d_1}
  \left(1-\frac{\cD_2}{\lambda}\right)_{d_1}
  \left(1-\frac{\cD_5}{\lambda}\right)_{d_2}} ~,
\end{aligned}
\end{equation}
All instanton operators are proportional to the compact divisor
operator $\cD_4$, except for those of the form $\sfp_{d_1,0}$. These span
the singular cone of the disk function, which is non-trivial in this example.
It follows that infinitely many terms in the partition function are singular in 
the
non-equivariant limit and a regularization is necessary to get a cubic PF 
solution.

The regular solutions to the PF equations are in correspondence with the four
generators of the homology lattice and in the non-equivariant limit can be 
written
as
\begin{equation}
\begin{aligned}
 \Pi(\pt) &= 1 ~,\\
 \Pi(C^1) &= \log z_1 +2G = \log\td{z}_1 ~, \\
 \Pi(C^2) &= \log z_2 -G+2H_1 = \log\td{z}_2~, \\
 \Pi(D_4) &= \log\td{z}_1\log\td{z}_2 + \log^2\td{z}_2
 -2(2H_1^2-H_1^{(10)}-2H_1^{(01)}) ~,
\end{aligned}
\end{equation}
where we define the functions
\begin{equation}
 G^{(i)} := \sum_{d_1=1}^{\infty} z_1^{d_1} \partial_{d_1}^i
 \frac {\Ga(2d_1)} {\Ga(d_1+1)^2}
\end{equation}
\begin{equation}
 B^{(i)} := \sum_{d_1=1}^{\infty}z_1^{d_1}\partial_{d_1}^i
 \frac{\Ga(2d_1)}{\Ga(d_1+1)^2} \psi^{(0)}(2d_1)
\end{equation}
\begin{equation}
 H_1^{(ij)} := \sum_{\substack{2d_1-d_2\leq0\\d_2>0}}
 z_1^{d_1} z_2^{d_2}
 \partial_{d_1}^i \partial_{d_2}^j
 \frac {\Ga(2d_2)} {\Ga(d_1+1)^2 \Ga(-2d_1+d_2+1) \Ga(d_2+1)}
\end{equation}
\begin{equation}
 H_2^{(ij)} := \sum_{\substack{2d_1-d_2>0\\d_2>0}} z_1^{d_1}(-z_2)^{d_2}
 \partial_{d_1}^i \partial_{d_2}^j
 \frac{\Ga(2d_1-d_2)\Ga(2d_2)} {\Ga(d_1+1)^2 \Ga(d_2+1)}
\end{equation}
and it is understood that
where we do not write superscripts we mean that they are all zero.
The quadratic solution corresponding to the compact divisor $D_4$ then satisfies
\begin{equation}
 \lim_{\e\to0} (-\lambda)^2 \hat{I}_{X_\bt}\cD_4\cf_{\Ga}
 = \Pi(D_4) + \tfrac{2\pi^2}{3}\Pi(\pt)~.
\end{equation}

Using the regularization scheme in \cref{sec:regularization} with
\begin{equation}
 R=
\left(
\begin{array}{ccc}
 \alpha & -1/2
\end{array}
\right)
\end{equation}
we can compute the regularized disk function
\begin{multline}
 (-\lambda)^3\lim_{\e\to0}\cf^D_\reg =
 - \tfrac14\log\td{z}_1\log^2\td{z}_2 - \tfrac16\log^3\td{z}_2
 - \tfrac{\pi^2}{3} \log\td{z}_2 \\
 + \log\td{z}_1 \left( (\tfrac12 G-H_1)^2 + H_2 - H_1^{(01)} \right)
 + \log\td{z}_2 \left( 2 H_1^2 -H_1^{(10)} - 2 H_1^{(01)} \right) \\
 - \tfrac16 G^3 + G^2 H_1 - \tfrac83 H_1^3
 - G H_1^{(10)} - 2 G H_2 - \tfrac{\pi^2}{3} G \\
 + 4 H_1 H_1^{(01)} + 2 H_1 H_1^{(10)}
 + \tfrac{2\pi^2}{3} H_1 - H_1^{(02)} - H_1^{(11)} + H_2^{(10)}
 + \tfrac{\alpha}{3} (\log\td{z}_1-2G) \times \\
 \left((4\alpha(4\alpha+3)+3)G^2 - (4\alpha(4\alpha+3)+3) G \log\td{z}_1
 + \alpha(4\alpha+3) \log^2\td{z}_1 + 2\pi^2\right)~,
\end{multline}
which is a cubic solution to modified PF equations, obtained by operators
\begin{equation}
\begin{aligned}
 &\cD_1\cD_2\cD_4-z_1\cD_3(\cD_3+\lambda)\cD_4~, \\
 &\cD_3\cD_4\cD_5-z_2\cD_4(\cD_4+\lambda)(\cD_4+2\lambda) ~.
\end{aligned}
\end{equation}
The regularized cubic solution associated to the fundamental cycle of $X_\bt$ is
\begin{multline}
\label{eq:Pi-reg-F2}
 \Pi_\reg(X_\bt) =
 - \tfrac14 \log\td{z}_1 \log^2\td{z}_2
 -\tfrac16\log^3\td{z}_2
 +\log\td{z}_1\left(-\tfrac12 B-\tfrac{\gamma}{2} G
 +\left(\tfrac12 G-H_1\right)^2-H_1^{(01)}+H_2\right) \\
 +\log\td{z}_2\left(2H_1^2-H_1^{(10)}-2H_1^{(01)}\right)
 -2\Big(
 \tfrac14 B^{(1)}
 -\tfrac12 B G
 -\tfrac{\gamma}{4} G^2
 +\tfrac1{12} G^3
 -\tfrac12 G^2 H_1
 +\tfrac43 H_1^3 \\
 -2H_1 H_1^{(01)}
 +\tfrac12 H_1^{(02)}
 +\tfrac12 G H_1^{(10)}
 -H_1 H_1^{(10)}
 +\tfrac12 H_1^{(11)}
 +G H_2
 -\tfrac12 H_2^{(10)}
 -\tfrac{\pi^2}{12} G
 -\tfrac{\pi^2}{3} H_1
 \Big) ~,
\end{multline}
and it differs from $\cf^D_\reg$ by a lower-degree term proportional to the 
period
$\Pi(C^2)=\log\td{z}_2$ and also by a correction term $\delta$ that only
depends on $z_1$ (and not $z_2$). As $\lim_{\e\to0}\cD_4=\tfrac{\partial}
{\partial t^1}$, it follows that $\delta$ is in the kernel of the compact
divisor operator, as in \cref{eq:Freg-GV-CY3}.

The GV invariants $n_{d_1,d_2}$ can be read by matching \cref{eq:Pi-reg-F2} with
\cref{eq:localCY3-t3},
\begin{multline}
 \sum_{(d_1,d_2)\neq(0,0)} n_{d_1,d_2} \log(\td{z}_1^{d_1}\td{z}_2^{d_2})
 \Li_2 (\td{z}_1^{d_1}\td{z}_2^{d_2}) = \\
 = \log\td{z}_1\left(-\tfrac12 B-\tfrac{\gamma}{2} G
 +\left(\tfrac12 G-H_1\right)^2-H_1^{(01)}+H_2\right) \\
 + \log\td{z}_2\left(2H_1^2-H_1^{(10)}-2H_1^{(01)}\right)
\end{multline}
or
\begin{multline}
 \sum_{(d_1,d_2)\neq(0,0)} n_{d_1,d_2} 
 \Li_3 (\td{z}_1^{d_1}\td{z}_2^{d_2}) =
 \tfrac14 B^{(1)}
 -\tfrac12 B G
 -\tfrac{\gamma}{4} G^2
 +\tfrac1{12} G^3
 -\tfrac12 G^2 H_1
 +\tfrac43 H_1^3 \\
 -2H_1 H_1^{(01)}
 +\tfrac12 H_1^{(02)}
 +\tfrac12 G H_1^{(10)}
 -H_1 H_1^{(10)}
 +\tfrac12 H_1^{(11)}
 +G H_2
 -\tfrac12 H_2^{(10)}
 -\tfrac{\pi^2}{12} G
 -\tfrac{\pi^2}{3} H_1~,
\end{multline}
which give the same results as those in ref.~\cite[Table~11]{Chiang:1999tz},
including $n_{1,0}=-1/2$.
We should remark, however, that from $\Pi(D_4)$ one can read all $n_{\bd}$ with
$d_2\neq0$ and since $D_4$ is compact these numbers are uniquely defined.
The GV invariants $n_{d_1,0}$ instead only appear in the
expansion of $\Pi_\reg(X_\bt)$, which is regularization-dependent, hence they
are not guaranteed to be integers, as it is clear from the result
$n_{1,0}=-1/2$. If we were to read $n_{d_1,0}$ from
the non-equivariant limit of $\cf^D_\reg$ instead, we would get different
results (precisely because of the correction term $\delta$).
This signals that when instantons are singular
then some of the GW invariants (as computed from PF solutions) need
regularization and no canonical choice exists.
The discussion can be easily generalized to the K-theoretic case and also there
we observe that the instantons of charges $(d_1,0)$ are singular in the
$q\to1$ limit.

\subsection{Local \texorpdfstring{$A_2$}{A2}}
\label{ex:localA2}

We consider the CY manifold corresponding to the charge matrix
\begin{equation}
Q =
\begin{pmatrix}
 1 & 1 & -2 &  0 &  0  & 0 \\
 0 & 0 &  1 & -2 &  1  & 0 \\
 0 & 0 &  0 &  1 & -2  & 1
\end{pmatrix}
\end{equation}
with chamber $t^1,t^2,t^3>0$. By geometric engineering arguments this geometry
corresponds to a 5d gauge theory with gauge group $SU(3)$
and Chern-Simons level 3.
This manifold has two compact toric divisors $D_4$ and $D_5$.
We define the disk function
\begin{multline}
 \cf^D (\bt,\e;\lambda) = \lambda^{-6} \oint_\qjk
 \frac {\dif\phi_1 \dif\phi_2 \dif\phi_3} {(2\pi \ii)^3}
 \eu^{\phi_1 t^1+\phi_2 t^2+\phi_3 t^3}
 \Ga\left(\tfrac{\e_1 + \phi_1}{\lambda}\right)
 \Ga\left(\tfrac{\e_2 + \phi_1}{\lambda}\right) \\
 \Ga\left(\tfrac{\e_3 -2\phi_1 + \phi_2}{\lambda}\right)
 \Ga\left(\tfrac{\e_4 -2\phi_2 + \phi_3}{\lambda}\right)
 \Ga\left(\tfrac{\e_5 +\phi_2 - 2\phi_3}{\lambda}\right)
 \Ga\left(\tfrac{\e_6 + \phi_3}{\lambda}\right)
\end{multline}
with poles in $(\phi_1, \phi_2, \phi_3)$ located at (minus)
\begin{equation}
\begin{tabu}{l}
(\e_1 + k_1, \e_3 + 2\e_1 + 2k_1 + k_2, \e_6 + k_3)~, \\
(\e_2 + k_1, \e_3 + 2\e_2 + 2k_1 + k_2, \e_6 + k_3) ~, \\
(\e_1 + k_1, 2\e_1 + \e_3 + 2k_1 + k_2, 4\e_1 + 2\e_3 + \e_4 + 4k_1 + 2k_2 + 
k_3) ~, \\
(\e_2 + k_1, 2\e_2 + \e_3 + 2k_1 + k_2, 4\e_2 + 2\e_3 + \e_4 + 4k_1 + 2k_2 + 
k_3) ~, \\
(\e_1 + k_1, \e_5 + 2\e_6 + k_2 + 2k_3, \e_6 + k_3) ~,\\
(\e_2 + k_1, \e_5 + 2\e_6 + k_2 + 2k_3, \e_6 + k_3) ~.
\end{tabu}
\end{equation}
The equivariant cohomology ring is the quotient by the ideal
\begin{equation}
 I_\mathrm{SR} = \langle x_1 x_2, x_3 x_5, x_3 x_6, x_4 x_6 \rangle~.
\end{equation}
The quantum cohomology relations / equivariant PF operators are
\begin{equation}
\begin{aligned}
 \PF^\eq_{(1,0,0)} &= \cD_1\cD_2 - \eu^{-\lambda t^1}(\lambda+\cD_3)\cD_3 ~, \\
 \PF^\eq_{(0,1,0)} &= \cD_3\cD_5 - \eu^{-\lambda t^2}(\lambda+\cD_4)\cD_4 ~,\\
 \PF^\eq_{(0,1,1)} &= \cD_3\cD_6 - \eu^{-\lambda(t^2+t^3)}\cD_4\cD_5 ~, \\
 \PF^\eq_{(0,0,1)} &= \cD_4\cD_6 - \eu^{-\lambda t^3}(\lambda+\cD_5)\cD_5~,
\end{aligned}
\end{equation}
from which we can derive the Givental $\hat{I}$-operator
\begin{multline}
 \hat{I}_{X_\bt} = 
 1 + \sum_{d_1=1}^\infty z_1^{d_1}
 \frac
 {\left(\frac{\cD_3}{\lambda}\right)_{2d_1}}
 {\left(1-\frac{\cD_1}{\lambda}\right)_{d_1}
  \left(1-\frac{\cD_2}{\lambda}\right)_{d_1}} \\
 +  \sum_{\substack{2d_1-d_2\leq0\\2d_2-d_3>0\\-d_2+2d_3\leq0}}
 z_1^{d_1}z_2^{d_2}(-z_3)^{d_3}
 \frac
 {\left(\frac{\cD_4}{\lambda}\right)_{2d_2-d_3}}
 {\left(1-\frac{\cD_1}{\lambda}\right)_{d_1}
  \left(1-\frac{\cD_2}{\lambda}\right)_{d_1}
  \left(1-\frac{\cD_3}{\lambda}\right)_{-2d_1+d_2}
  \left(1-\frac{\cD_5}{\lambda}\right)_{d_2-2d_3}
  \left(1-\frac{\cD_6}{\lambda}\right)_{d_3}} \\
 + \sum_{\substack{2d_1-d_2\leq0\\2d_2-d_3\leq0\\-d_2+2d_3>0}}
 z_1^{d_1}(-z_2)^{d_2}z_3^{d_3}
 \frac
 {\left(\frac{\cD_5}{\lambda}\right)_{-d_2+2d_3}}
 {\left(1-\frac{\cD_1}{\lambda}\right)_{d_1}
  \left(1-\frac{\cD_2}{\lambda}\right)_{d_1}
  \left(1-\frac{\cD_3}{\lambda}\right)_{-2d_1+d_2}
  \left(1-\frac{\cD_4}{\lambda}\right)_{-2d_2+d_3}
  \left(1-\frac{\cD_6}{\lambda}\right)_{d_3}} \\
 + \sum_{\substack{2d_1-d_2\leq0\\2d_2-d_3>0\\-d_2+2d_3>0}}
 z_1^{d_1}(-z_2)^{d_2}(-z_3)^{d_3}
 \frac
 {\left(\frac{\cD_4}{\lambda}\right)_{2d_2-d_3}
  \left(\frac{\cD_5}{\lambda}\right)_{-d_2+2d_3}}
 {\left(1-\frac{\cD_1}{\lambda}\right)_{d_1}
  \left(1-\frac{\cD_2}{\lambda}\right)_{d_1}
  \left(1-\frac{\cD_3}{\lambda}\right)_{-2d_1+d_2}
  \left(1-\frac{\cD_6}{\lambda}\right)_{d_3}} \\
 + \sum_{\substack{2d_1-d_2>0\\2d_2-d_3>0\\-d_2+2d_3\leq0}}
 z_1^{d_1}(-z_2)^{d_2}(-z_3)^{d_3}
 \frac
 {\left(\frac{\cD_3}{\lambda}\right)_{2d_1-d_2}
  \left(\frac{\cD_4}{\lambda}\right)_{2d_2-d_3}}
 {\left(1-\frac{\cD_1}{\lambda}\right)_{d_1}
  \left(1-\frac{\cD_2}{\lambda}\right)_{d_1}
  \left(1-\frac{\cD_5}{\lambda}\right)_{d_2-2d_3}
  \left(1-\frac{\cD_6}{\lambda}\right)_{d_3}} \\
 + \sum_{\substack{2d_1-d_2>0\\2d_2-d_3\leq0\\-d_2+2d_3>0}}
 z_1^{d_1}z_2^{d_2}z_3^{d_3}
 \frac
 {\left(\frac{\cD_3}{\lambda}\right)_{2d_1-d_2}
  \left(\frac{\cD_5}{\lambda}\right)_{-d_2+2d_3}}
 {\left(1-\frac{\cD_1}{\lambda}\right)_{d_1}
  \left(1-\frac{\cD_2}{\lambda}\right)_{d_1}
  \left(1-\frac{\cD_4}{\lambda}\right)_{-2d_2+d_3}
  \left(1-\frac{\cD_6}{\lambda}\right)_{d_3}} \\
 + \sum_{\substack{2d_1-d_2>0\\2d_2-d_3>0\\-d_2+2d_3>0}}
 z_1^{d_1}z_2^{d_2}(-z_3)^{d_3}
 \frac
 {\left(\frac{\cD_3}{\lambda}\right)_{2d_1-d_2}
  \left(\frac{\cD_4}{\lambda}\right)_{2d_2-d_3}
  \left(\frac{\cD_5}{\lambda}\right)_{-d_2+2d_3}}
 {\left(1-\frac{\cD_1}{\lambda}\right)_{d_1}
  \left(1-\frac{\cD_2}{\lambda}\right)_{d_1}
  \left(1-\frac{\cD_6}{\lambda}\right)_{d_3}} ~.
\end{multline}
The instanton operators are regular except for those of the form
\begin{equation}
 \sfp_{(d_1,0,0)} = z_1^{d_1}
 \frac
 {\left(\frac{\cD_3}{\lambda}\right)_{2d_1}}
 {\left(1-\frac{\cD_1}{\lambda}\right)_{d_1}
  \left(1-\frac{\cD_2}{\lambda}\right)_{d_1}}
\end{equation}
which are not proportional to any of the compact divisor operators 
$\cD_4,\cD_5$.
It follows that the $z_1$ instantons are
singular in the non-equivariant limit, similarly to the local $F_2$ case.
All other instanton operators either contain $\cD_4$ or $\cD_5$ in the numerator
and the corresponding instanton contributions are regular.

Observe that for $z_3=0$ the $\hat{I}$-operator reduces to that of local $F_2$,
since the two charge matrices are equal once we remove the last line
from the one of local $A_2$.
Similarly, for $z_1=0$ the $\hat{I}$-operator reduces to that of the $A_2$ case,
which corresponds to removing the first line of the charge matrix.

The solutions to the non-equivariant PF equations are
\begin{equation}
\begin{aligned}
 \Pi(\pt) =& 1 ~,\\
 \Pi(C^1) =& \log z_1 + 2M_0 = \log\td{z}_1 ~, \\
 \Pi(C^2) =& \log z_2 - M_0 + 2M_1 - M_2 = \log\td{z}_2 ~,\\
 \Pi(C^3) =& \log z_3 - M_1 + 2M_2 = \log\td{z}_3 ~,\\
 \Pi(D_4) =& (\log\td{z}_1+\log\td{z}_2)\log\td{z}_2
 - 4 M_1^2 + 4 M_1 M_2 - M_2^2 - 4M_3 \\
 & + 2M_1^{(100)} + 4M_1^{(010)} - M_2^{(100)} - 2M_2^{(010)} ~,\\
 \Pi(D_5) =& (\log\td{z}_1+2\log\td{z}_2+2\log\td{z}_3)\log\td{z}_3
 + 2M_1^2 - 2M_1 M_2 - 4M_2^2 + 2M_3 \\
 & - M_1^{(100)} - 2M_1^{(010)} + 2M_2^{(100)} + 4M_2^{(010)} + 6M_2^{(001)}~,
\end{aligned}
\end{equation}
where we define the functions
\begin{equation}
 M_0 := \sum_{d_1=1}^\infty z_1^{d_1} \frac{\Ga(2d_1)}{\Ga(d_1+1)^2}~,
 \quad
 B^{(i)} := \sum_{d_1=1}^{\infty}z_1^{d_1}\partial_{d_1}^i
 \frac{\Ga(2d_1)}{\Ga(d_1+1)^2} \psi^{(0)}(2d_1)~,
\end{equation}
\begin{multline}
 M_1^{(ijk)} := \sum_{\substack{2d_1-d_2\leq0\\2d_2-d_3>0\\-d_2+2d_3\leq0}}
 z_1^{d_1} z_2^{d_2} (-z_3)^{d_3} \times \\
 \times\partial_{d_1}^i \partial_{d_2}^j \partial_{d_3}^k
 \frac{\Ga(2d_2-d_3)}{\Ga(d_1+1)^2 \Ga(-2d_1+d_2+1) \Ga(d_2-2d_3+1) 
\Ga(d_3+1)}~,
\end{multline}
\begin{multline}
 M_2^{(ijk)} := \sum_{\substack{2d_1-d_2\leq0\\2d_2-d_3\leq0\\-d_2+2d_3>0}}
 z_1^{d_1} (-z_2)^{d_2} z_3^{d_3}
 \partial_{d_1}^i \partial_{d_2}^j \partial_{d_3}^k \times \\
 \times
 \frac{\Ga(-d_2+2d_3)}{\Ga(d_1+1)^2 \Ga(-2d_1+d_2+1) \Ga(-2d_2+d_3+1) 
\Ga(d_3+1)}~,
\end{multline}
\begin{equation}
 M_3^{(ijk)} := \sum_{\substack{2d_1-d_2\leq0\\2d_2-d_3>0\\-d_2+2d_3>0}}
 z_1^{d_1}(-z_2)^{d_2}(-z_3)^{d_3}
 \partial_{d_1}^i\partial_{d_2}^j\partial_{d_3}^k
 \frac{\Ga(2d_2-d_3)\Ga(-d_2+2d_3)}{\Ga(d_1+1)^2\Ga(-2d_1+d_2+1)\Ga(d_3+1)}~,
\end{equation}
\begin{equation}
 M_4^{(ijk)} := \sum_{\substack{2d_1-d_2>0\\2d_2-d_3>0\\-d_2+2d_3\leq0}}
 z_1^{d_1}(-z_2)^{d_2}(-z_3)^{d_3}
 \partial_{d_1}^i\partial_{d_2}^j\partial_{d_3}^k
 \frac{\Ga(2d_1-d_2) \Ga(2d_2-d_3)}{\Ga(d_1+1)^2 \Ga(d_2-2d_3+1) \Ga(d_3+1)}~.
\end{equation}

The GV invariants $n_{d_1,d_2,d_3}$ can be read from $\Pi(D_4)$ or $\Pi(D_5)$
if $-2d_2+d_3\neq0$ or $d_2-2d_3\neq0$, respectively. If $d_2=d_3=0$, then
$n_{d_1,0,0}$ cannot be read from either of these regular solutions and a
regularization for $\Pi(X_\bt,\e)$ is needed.
The regularized cubic solution of the modified PF equations is
\begin{multline}
 \Pi_\reg(X_\bt)
 =
 - \tfrac13 \log\td{z}_1 \log^2\td{z}_2
 - \tfrac13 \log\td{z}_1 \log^2\td{z}_3
 - \tfrac13 \log^2\td{z}_2 \log\td{z}_3
 - \tfrac23 \log\td{z}_2 \log^2\td{z}_3 \\
 - \tfrac13 \log\td{z}_1 \log^2\td{z}_2 \log\td{z}_3
 - \tfrac29 \log^3\td{z}_2
 - \tfrac49 \log^3\td{z}_3 \\
 + \log\td{z}_1 \left\{ - \tfrac23 (B+\gamma M_0-\tfrac{M_0^2}{2})
 - M_0 M_1 + M_1^2 - M_1 M_2 + M_2^2 + M_3 + M_4 - M_1^{(010)} - M_2^{(001)}
 \right\} \\
 + \log\td{z}_2 \left\{ 2M_1^2 - 2 M_1 M_2 - 2 M_1^{(010)} - M_1^{(100)}
 + 2 M_2^2 - 2 M_2^{(001)} + 2 M_3
 \right\} \\
 + \log\td{z}_3 \left\{ 3 M_2^2 - 4 M_2^{(001)} - 2 M_2^{(010)} - M_2^{(100)}
 \right\} \\
 + \Big\{
 - \tfrac23 ( B^{(1)} - \gamma M_0^2 - 2 M_0 B
 + \tfrac13 M_0^3 - \tfrac{\pi^2}{3} M_0 )
 + M_0^2 M_1 - 2 M_0 M_4 + 4 M_1^2 M_2 \\
 - 2 M_1 M_2^2 - 4 M_1 M_3 + 2 M_2 M_3
 - \tfrac83 (M_1^3 + M_2^3) + \tfrac{2\pi^2}{3} (M_1+M_2)
 - M_0 M_1^{(100)} \\
 + M_1 ( 4 M_1^{(010)} + 2 M_1^{(100)} - 2 M_2^{(010)} - M_2^{(100)} )
 + M_2 ( 4 M_2^{(010)} + 2 M_2^{(100)} - 2 M_1^{(010)} - M_1^{(100)}
 + 6 M_2^{(001)} ) \\
 - M_1^{(020)} - M_1^{(110)} - 2 M_2^{(002)} - 2 M_2^{(011)} - M_2^{(101)}
 + 2 M_3^{(010)} + M_3^{(100)} + M_4^{(100)}
 \Big\}
\end{multline}
and by matching against \cref{eq:localCY3-t3} we can read all GV invariants
and reproduce the results of ref.~\cite[Table~4]{Chiang:1999tz}
(modulo some typos); we also get
$n_{1,0,0}=-2/3$ as observed in ref.~\cite[Section~4.1.8]{Honma:2018qcr}.
%
%
%
The numbers affected by typos are
$n_{1,d_2,d_3} = - 2 (d_2-1) d_3 + d_2 (d_2-1)$
for $d_3 > d_2$, as well as the bold entries in the tables
\begin{equation}
d_1=2: \quad
\begin{tabu}{c|ccccccc}
d_2 \diagdown d_3 & 0 & 1 & 2 & 3 & 4 & 5 & 6 \\
\hline
3 &   -6 &  -10 &  -12 &  -12 &  -10 & \bold{-14} & \bold{-18} \\
4 &  -32 &  -70 &  -96 & -110 & -112 & -126       & -192 \\
5 & -110 & -270 & -416 & -518 & -576 & -630       & -784 
\end{tabu}
\end{equation}
\begin{equation}
d_1=3: \quad
\begin{tabu}{c|ccccccc}
d_2 \diagdown d_3 & 0 & 1 & 2 & 3 & 4 & 5 & 6 \\
\hline
4 &   -8 &  -14 &  -18 &  -20 &  -20 & \bold{-18} & \bold{-24} \\
5 & -110 & -270 & -416 & -518 & -576 & -630       & -784 
\end{tabu}
\end{equation}
Since $n_{d_1,0,0}$ can only be read from $\Pi_\reg(X_\bt)$, it is not 
surprising that the obtained GV invariants are not all integer.

The regularized disk function $\cf^D_\reg$ can be obtained from \cref{eq:Freg}
by using the left-inverse matrix
\begin{equation}
 R=
\left(
\begin{array}{ccc}
 \alpha & -2/3 & -1/3 \\
  \beta & -1/3 & -2/3 \\
\end{array}
\right)~,
\end{equation}
where $\alpha,\beta$ parametrize the ambiguity in the choice regularization.

\section{Examples without compact divisors}
\label{sec:examples2}

In this section we present three examples with empty $H^2_\cmp(X_\bt)$ and 
non-empty $H^4_\cmp(X_\bt)$. The elements of $H^4_\cmp(X_\bt)$ are in
one-to-one correspondence with compact double intersections
of non-compact toric divisors.

\subsection{Resolved conifold}

The resolved conifold
$X_\bt = \tot \co(-1) \oplus \co(-1) \to \BP^1$ is defined by the charge matrix
\begin{equation}
Q =
\begin{pmatrix}
 1 & 1 & -1 & -1\\
\end{pmatrix}
\end{equation}
and the chamber $t>0$.
The equivariant symplectic volume is
\begin{equation}
 \cf(t,\e) = \oint_\jk  \frac{\dif\phi}{2\pi\ii} \,
\frac{\eu^{\phi t}} {(\e_1 + \phi) (\e_2 +\phi) (\e_3 - \phi) (\e_4-\phi)}~,
\end{equation}
where we take poles at $\phi=-\e_1$ and $\phi=-\e_2$.
We have the classical cohomology relation
\begin{equation}
 \cD_1\cD_2 \cf (t,\e) = 0
\end{equation}
so that the equivariant cohomology ring is
\begin{equation}
 H^\bullet_{\sft}(X_{\bt})
 \cong\BC[\phi,\e_1,\e_2,\e_3,\e_4] / \langle (\e_1+\phi)(\e_2+\phi) \rangle~.
\end{equation}

The K-theoretic disk function is defined as
\begin{equation}
 \cz^D(T,q;\q) =
 -\oint_\qjk \frac{\dif w}{2\pi\ii w} \,
 \frac {w^{-T}}
 {(q_1w;\q)_\infty(q_2w;\q)_\infty(q_3w^{-1};\q)_\infty(q_4w^{-1};\q)_\infty}
\end{equation}
with two towers of poles at $w=q_1^{-1}\q^{-d}$ and $w=q_2^{-1}\q^{-d}$
for $d\geq0$. The quantum K-theory is encoded in the difference equation
\begin{equation}
 \left[ (1-\Delta_1) (1-\Delta_2) - \q^T (1-\Delta_3) (1-\Delta_4) \right]
 \cz^D (T,q;\q) = 0
\end{equation}
with solution
\begin{equation}
 \cz^D(T,q;\q) = \sum_{d=0}^\infty \q^{dT}
 \frac {(\Delta_3;\q)_{d}(\Delta_4;\q)_{d}}
 {(\q^{-d}\Delta_1;\q)_{d}(\q^{-d}\Delta_2;\q)_{d}}
 \cz_{\Ga_\q} (T,q;\q)~,
\end{equation}
where the function
\begin{equation}
 \cz_{\Ga_\q} (T,q;\q) = c_1 q_1^T + c_2 q_2^T
\end{equation}
with
\begin{equation}
\begin{aligned}
 c_1 &= \frac{1}
 {(\q;\q)_\infty(q_2q_1^{-1};\q)_\infty(q_3q_1;\q)_\infty(q_4q_1;\q)_\infty}~,\\
 c_2 &= \frac{1}
 {(\q;\q)_\infty(q_1q_2^{-1};\q)_\infty(q_3q_2;\q)_\infty(q_4q_2;\q)_\infty}
\end{aligned}
\end{equation}
satisfies the classical K-theory relation
\begin{equation}
 (1-\Delta_1) (1-\Delta_2) \cz_{\Ga_\q} (T,q;\q) = 0~.
\end{equation}
The resolved conifold has no compact divisors, but the intersection of $D_3$ and
$D_4$ is the base of the fibration $\BP^1$ that generates $H_2 (X_{\bt})$.
It follows that the disk partition function satisfies
a generalization of the compact divisor equation, namely
\begin{equation}
 (1-\Delta_3)(1-\Delta_4)\cz^D(T,q;\q) \quad\text{is analytic at $q_i=1$}~.
\end{equation}
To see why this is the case, we rewrite
\begin{equation}
 (1-\Delta_3) (1-\Delta_4) \cz^D (T,q;\q) = \sum_{n^3\geq0} \sum_{n^4\geq0}
 \frac {(\q q_3)^{n^3}} {(\q;\q)_{n^3}} \frac {(\q q_4)^{n^4}} {(\q;\q)_{n^4}}
 \sum_{\Lambda_{3,4} (T,n^3,n^4)}
	\frac {q_1^{n^1}q_2^{n^2}} {(\q;\q)_{n^1}(\q;\q)_{n^2}}
\end{equation}
with
\begin{equation}
 \Lambda_{3,4}(T,n^3,n^4) =
	\left\{ (n^1,n^2) \in \BN^2\,\middle|\,n^1+n^2=T+n^3+n^4 \right\}
\end{equation}
so that each term in the $\q$ expansion is finite and polynomial in the $q_i$.
Sending all the $q_i$ to 1 is therefore a well-defined limit.

The cohomological limit $\hbar\to0$ is straightforward to compute.
The disk function becomes
\begin{equation}
 \cf^{D}(t,\e;\lambda) =
 \lambda^{-4} \oint_\qjk \frac{\dif\phi} {2\pi\ii} \,
 \eu^{\phi t}\,
 \Ga \left( \tfrac{\e_1 + \phi}{\lambda} \right)
 \Ga \left( \tfrac{\e_2 + \phi}{\lambda} \right)
 \Ga \left( \tfrac{\e_3 - \phi}{\lambda} \right)
 \Ga \left( \tfrac{\e_4 - \phi}{\lambda} \right)
\end{equation}
satisfying the quantum cohomology relation
\begin{equation}
 \left[ \cD_1 \cD_2 - \eu^{-\lambda t} \cD_3 \cD_4 \right]
 \cf^D (t,\e;\lambda) = 0~.
\end{equation}
We can write the instanton expansion
\begin{equation}
 \cf^D (t,\e;\lambda)
 = \sum_{d=0}^\infty z^d
 \frac{\left(\tfrac{\cD_3}{\lambda}\right)_d
 \left(\tfrac{\cD_4}{\lambda}\right)_d}
 {\left(1-\tfrac{\cD_1}{\lambda}\right)_d
 \left(1-\tfrac{\cD_2}{\lambda}\right)_d}
 \cf_{\Ga} (t,\e;\lambda)
\end{equation}
with
\begin{equation}
 \cf_{\Ga} (t,\e;\lambda) =
 \frac {\eu^{-\e_1 t}} {\lambda^3}
 \Ga \left( \tfrac{\e_2-\e_1}{\lambda} \right)
 \Ga \left( \tfrac{\e_1+\e_3}{\lambda} \right)
 \Ga \left( \tfrac{\e_1+\e_4}{\lambda} \right)
 + \frac{\eu^{-\e_2 t}}{\lambda^3}
 \Ga \left( \tfrac{\e_1-\e_2}{\lambda} \right)
 \Ga \left( \tfrac{\e_1+\e_3}{\lambda} \right)
 \Ga \left( \tfrac{\e_1+\e_4}{\lambda} \right)~.
\end{equation}

The instanton operators
\begin{equation}
 \sfp_d = z^d \frac{\left(\frac{\cD_3}{\lambda}\right)_d
 \left( \frac{\cD_4}{\lambda} \right)_d}
 {\left( 1-\frac{\cD_1}{\lambda} \right)_d
 \left( 1-\frac{\cD_2}{\lambda} \right)_d}
\end{equation}
are proportional to the product $\cD_3\cD_4$, which corresponds to the 
intersection
of divisors $D_3$, $D_4$. Since the intersection is compact, by
\cref{eq:singular-instantons}, the instanton corrections are non-singular.
Hence we can compute
\begin{equation}
 \lim_{\e\to0}\left[
 \cf^D (t,\e;\lambda)-\cf_{\Ga} (t,\e;\lambda)\right]
 = \frac1{(-\lambda)^3}
 \left[
 \log z\Li_2(z)
 -2\Li_3(z)
 \right]~.
\end{equation}
The equivariant Givental $I$-function is
\begin{equation}
 I_{X_{\bt}} = \sum_{d=0}^\infty \eu^{-\lambda d t + \phi t}
 \frac{\left(\frac{x_3}{\lambda}\right)_{d}\left(\frac{x_4}{\lambda}\right)_{d}}
 {\left(1-\frac{x_1}{\lambda}\right)_{d}\left(1-\frac{x_2}{\lambda}\right)_{d}}
\end{equation}
and the solutions to the equivariant PF equations are
\begin{equation}
 \Pi(p_i) = z^{\frac{\e_i}{\lambda}} \sum_{d=0}^\infty z^d
 \frac{\prod_{j=3}^4\left(\frac{\e_j-\e_i}{\lambda}\right)_{d}}
 {\prod_{j=1}^2\left(1-\frac{\e_j-\e_i}{\lambda}\right)_{d}}~,
 \quad
 i=1,2~.
\end{equation}
The regular periods that survive the non-equivariant limit are
\begin{equation}
\begin{aligned}
 \Pi(\pt) &= 1 ~,\\
 \Pi(\BP^1) &= \log z ~,
\end{aligned}
\end{equation}
which correspond to the two generators of the homology lattice.
The modified PF operator
\begin{equation}
 \cD_1\cD_2\cD_3\cD_4 -z (\cD_3+\lambda)\cD_3(\cD_4+\lambda)\cD_4
\end{equation}
admits the following quadratic and cubic solutions
\begin{equation}
\begin{aligned}
 \Pi_\reg(D_i) &= -\tfrac12\log^2 z - \Li_2(z)~,\quad i=3,4~, \\
 \Pi_\reg(X_\bt) &= \tfrac16\log^3 z + \log z\Li_2(z) - 2\Li_3(z) ~,
\end{aligned}
\end{equation}
corresponding to the non-compact cycles of $X_\bt$.
By compactness of $D_3\cap D_4$ we have
\begin{equation}
 \lim_{\e\to0}(-\lambda)\hat{I}_{X_\bt}\cD_3\cD_4\cf_{\Ga}
 = \Pi(\BP^1)~.
\end{equation}
Since there are no compact divisors, we cannot use \cref{eq:CY3-t2} to read
the GV invariants and we cannot apply the regularization procedure to $\cf^D$.
What we can do in this case is restrict to a non-compact divisor and
regularize the restricted disk function. The non-compact divisor has itself
a compact divisor corresponding to the $\BP^1$. Define
\begin{equation}
 \cf^D|_{D_4} := \hat{I}_{X_\bt}\cD_4\cf_{\Ga}~,
\end{equation}
which is still singular but can be regularized via the procedure in
\cref{sec:regularization}, namely
\begin{equation}
 \cf^D_\reg(t)|_{D_4} := \cf^D(t)|_{D_4}-\eu^{\e_3 t}\cf^D(0)|_{D_4}~,
\end{equation}
where we used the fact that $D_3$ is a compact divisor inside of $D_4$.
This function is regular
\begin{equation}
 \lim_{\e\to0}(-\lambda)^2\cf^D_\reg(t)|_{D_4} = \Pi_\reg(D_4)
\end{equation}
and from \cref{eq:CY3-t2} we can read the GV invariants $n_d=\delta_{d,1}$.
The same can be done upon exchanging the divisors $D_3$ and $D_4$.

\subsection{\texorpdfstring{$\co(-1,-1)\oplus\co(-1,-1)$
over $\BP^1 \times \BP^1$}{O(-1,-1)⊕O(-1,-1) over P1xP1}}

The charge matrix is
\begin{equation}
Q =
\begin{pmatrix}
 1 & 1 & 0 & 0 & -1 & -1 \\
 0 & 0 & 1 & 1 & -1 & -1
\end{pmatrix}
\end{equation}
with the chamber defined by $t^1,t^2>0$ and the disk function is
\begin{multline}
 \cf^D = \frac1{\lambda^6}
 \oint_{\qjk}\frac{\dif\phi_1\dif\phi_2}{(2\pi\ii)^2} 
 \eu^{\phi_1 t^1+\phi_2 t^2}
 \Ga\left(\tfrac{\e_1+\phi_1}{\lambda}\right)
 \Ga\left(\tfrac{\e_2+\phi_1}{\lambda}\right)
 \Ga\left(\tfrac{\e_3+\phi_2}{\lambda}\right)
 \Ga\left(\tfrac{\e_4+\phi_2}{\lambda}\right)\\
 \Ga\left(\tfrac{\e_5-\phi_1-\phi_2}{\lambda}\right)
 \Ga\left(\tfrac{\e_6-\phi_1-\phi_2}{\lambda}\right)~,
\end{multline}
which is annihilated by the equivariant PF operators
\begin{equation}
\begin{aligned}
 &\cD_1\cD_2-z_1\cD_5\cD_6~, \\
 &\cD_3\cD_4-z_2\cD_5\cD_6~.
\end{aligned}
\end{equation}
Similarly to the resolved conifold case, we have two non-compact divisors
$D_5$, $D_6$ that intersect to a compact four-cycle corresponding to the base
$\BP^1\times\BP^1$. The instanton operators are
\begin{equation}
 \sfp_{d_1,d_2} = z_1^{d_1} z_2^{d_2}
 \frac{\left(\frac{\cD_5}{\lambda}\right)_{d_1+d_2}
 \left(\frac{\cD_6}{\lambda}\right)_{d_1+d_2}}
 {\left(1-\frac{\cD_1}{\lambda}\right)_{d_1}
 \left(1-\frac{\cD_2}{\lambda}\right)_{d_1}
 \left(1-\frac{\cD_3}{\lambda}\right)_{d_2}
 \left(1-\frac{\cD_4}{\lambda}\right)_{d_2}}
\end{equation}
so that instanton corrections of degree $(d_1,d_2) \neq (0,0)$ are
regular in the non-equivariant limit.
The non-equivariant $\hat{I}$-operator expands as
\begin{equation}
 \lim_{\e\to0}\hat{I}_{X_\bt} = 1 + G(z_1,z_2) 
 (\theta_1+\theta_2)^2 + \dots~,
\end{equation}
where
\begin{equation}
 G(z_1,z_2) = \sum_{(d_1,d_2)\neq(0,0)} z_1^{d_1} z_2^{d_2}
 \frac{\Ga(d_1+d_2)^2}{\Ga(d_1+1)^2\Ga(d_2+1)^2}~.
\end{equation}
Since there is no linear term in the expansion, it follows that the
mirror map is trivial,
\begin{equation}
 \td{z}_i = z_i ~.
\end{equation}
The solutions to the non-equivariant PF equations are
\begin{equation}
\begin{aligned}
 \Pi(\pt) &= 1 ~,\\
 \Pi(C_1) &= \log z_1~, \\
 \Pi(C_2) &= \log z_2~, \\
 \Pi(\BP^1\times\BP^1) &= \log z_1 \log z_2 + 2 G(z_1,z_2) ~,
\end{aligned}
\end{equation}
where $C_1$ and $C_2$ are the homology two-cycles
corresponding to the two $\BP^1$'s.

We can compute the following regular solution to the PF equations
\begin{equation}
 \lim_{\e\to0}(-\lambda)^2\hat{I}_{X_\bt}\cD_5\cD_6\cf_{\Ga}
 = \Pi(\BP^1\times\BP^1) + \tfrac{\pi^2}{3}\Pi(\pt)
\end{equation}
from which we can read the GV invariants $n_{d_1,d_2}(\BP^1\times\BP^1)$ by
using \eqref{eq:CY4-t2}. It follows that
\begin{equation}
 G(z_1,z_2) = \sum_{(d_1,d_2) \neq (0,0)}
 n_{d_1,d_2}(\BP^1\times\BP^1) \Li_2 (\td{z}_1^{d_1} \td{z}_2^{d_2})
\end{equation}
and the $n_{d_1,d_2}(\BP^1\times\BP^1)$ match those in
ref.~\cite[Section 3.3]{Klemm:2007in}.

In this case there are no singular instantons and we can read all GV invariants
from the period $\Pi(\BP^1\times\BP^1)$. Similarly to the resolved conifold 
case,
one could also compute a regularized cubic solution and read the same GV 
invariants
from that solution.

\subsection{\texorpdfstring{$\co(-1)\oplus\co(-2)$
over $\BP^2$}{O(-1)⊕O(-2) over P2}}

The charge matrix is
\begin{equation}
Q =
\begin{pmatrix}
 1 & 1 & 1 & -1 & -2
\end{pmatrix}
\end{equation}
with the chamber defined by $t>0$ and the disk function is
\begin{equation}
 \cf^D = \frac1{\lambda^5}
 \oint_{\qjk}\frac{\dif\phi}{2\pi\ii} 
 \eu^{\phi t}
 \Ga\left(\tfrac{\e_1+\phi}{\lambda}\right)
 \Ga\left(\tfrac{\e_2+\phi}{\lambda}\right)
 \Ga\left(\tfrac{\e_3+\phi}{\lambda}\right)
 \Ga\left(\tfrac{\e_4-\phi}{\lambda}\right)
 \Ga\left(\tfrac{\e_5-2\phi}{\lambda}\right)~,
\end{equation}
which is annihilated by the equivariant PF operator
\begin{equation}
 \cD_1\cD_2\cD_3-z\cD_4\cD_5(\cD_5+\lambda)~.
\end{equation}
Similarly to the resolved conifold case, we have two non-compact divisors
$D_4$, $D_5$ that intersect to a compact four-cycle corresponding to the base
$\BP^2$. The instanton operators are
\begin{equation}
 \sfp_d = (-z)^d \frac{\left(\frac{\cD_4}{\lambda}\right)_d
 \left(\frac{\cD_5}{\lambda}\right)_{2d}}
 {\left(1-\frac{\cD_1}{\lambda}\right)_d
 \left(1-\frac{\cD_2}{\lambda}\right)_d
 \left(1-\frac{\cD_3}{\lambda}\right)_d},
\end{equation}
so that instanton corrections are regular in the non-equivariant limit.

The non-equivariant $\hat{I}$-operator expands as
\begin{equation}
 \lim_{\e\to0}\hat{I}_{X_\bt} = 1 + G(z) \theta^2 + \dots~,
\end{equation}
where
\begin{equation}
 G(z) = \sum_{d=1}^\infty (-z)^{d}
 \frac{\Ga(2d+1)}{d^4\Ga(d)^2}
 = 2\Li_2\left(\tfrac12\left(1-\sqrt{1+4z}\right)\right)
 - \Li_1^2\left(\tfrac12\left(1-\sqrt{1+4z}\right)\right)~.
\end{equation}
Since there is no linear term in the expansion, it follows that the
mirror map is trivial,
\begin{equation}
 \td{z} = z~.
\end{equation}
The solutions to the non-equivariant PF equations are
\begin{equation}
\begin{aligned}
 \Pi(\pt) &= 1~, \\
 \Pi(\BP^1) &= \log z_1 ~,\\
 \Pi(\BP^2) &= \tfrac12\log^2 z + G(z)~.
\end{aligned}
\end{equation}
We can compute the following regular solution to the PF equations
\begin{equation}
 \lim_{\e\to0}(-\lambda)^2\hat{I}_{X_\bt}\cD_4\cD_5\cf_{\Ga}
 = \Pi(\BP^2) + \tfrac{2\pi^2}{3}\Pi(\pt)
\end{equation}
from which we can read the GV invariants $n_{d}(\BP^2)$ by using
\cref{eq:CY4-t2}. It follows that
\begin{equation}
 G(z) = 2\sum_{d=1}^\infty
 n_{d}(\BP^2) \Li_2 (z^d)
\end{equation}
and the numbers $n_d(\BP^2)$ match those
in ref.~\cite[Section 3.2]{Klemm:2007in}.
In this case too there are no singular instantons and all GV invariants can be
read from the period $\Pi(\BP^2)$.

%% file: conclusion.tex
\section{Conclusions} \label{sec:conclusions} 
 
In this work we study the disk partition function $\cf^D (\bt,\e; 
\lambda)$ and its K-theoretic 
generalization $\cz^D (\bT,q;\q)$ for toric non-compact Kähler manifolds. 
We concentrate on 
structural issues related to the dependence of $\cf^D (\bt,\e; \lambda)$ 
on equivariant parameters $\e$'s and the ability to extract a non-equivariant 
answer. 
For non-compact manifolds the singularities in $\cf^D (\bt,\e; \lambda)$ at 
$\e=0$ are controlled 
by compact divisors (if $H^2_\cmp (X_{\bt})$ is non-empty). The nature of 
singularities depends on how compact divisors appear in 
the equivariant quantum 
cohomology relations. Using the formalism of Givental's equivariant 
$I$ and $J$ functions, we discuss
the nature of singularities in $\e$'s, the possibility to extract 
a non-equivariant answer (as well as the ambiguities involved),
and its impact on the enumerative geometry of the 
corresponding non-compact 
toric manifolds. We explain the relation between equivariant 
and modified PF equations, which are a natural generalization
of PF for non-compact manifolds. 
We perform a similar analysis for the K-theoretic function $\cz^D (\bT,q;\q)$.

Physically, $\cf^D (\bt,\e; \lambda)$ is a GLSM disk partition function
with a space-filling brane (all boundary conditions are Neumann)
\cite{Hori:2013ika, Sugishita:2013jca, Honda:2013uca}.
Our considerations on Givental's equivariant function, 
operators and the contours and formalism
extend to a more general setup
\begin{multline} \label{eq:cont-gamma-end}
 \cf^D_\al (\bt,\e; \lambda) =
 \lambda^{-N} \oint_\qjk \prod_{a=1}^r
 \frac {\dif \phi_a} {2\pi \ii} \,
 \eu^{\sum_a \phi_a t^a } \prod_{i=1}^N \Ga \left (\frac{x_i} {\lambda} \right) 
 \al(x) \\
 = \lambda^{-N} \oint_\jk \prod_{a} \frac{\dif\phi_a}{2\pi\ii}
 I_{X_{\bt}} \prod_{i} \Ga\left (\frac{x_i}{\lambda} \right ) \al(x)
 = \hat{I}_{X_\bt} \cdot \cf_{\al \Ga} (\bt, \e )~,
\end{multline}
where $\al(x)$ is a periodic function
in all its variables with period $\lambda$.
This object satisfies the equivariant PF equation, and
semiclassically it can be identified \cite{Hori:2013ika}
with the central charge of a brane $\mathcal B$,
with $\al$ being the Chern character of $\mathcal B$
\begin{equation}
 \cf^D_\al (\bt,\e; \lambda) = \int_{X_\bt} \eu^{\varpi_{\bt}-H_{\e}}
 \hat \Ga_\eq \ch (\mathcal B) + O(\eu^{-\lambda t})~.
\end{equation}
For example, if we split the set $\{ 1,2,\ldots, N \}$ into 
two disjoint subsets that we denote $\mathrm{Neu}$ (for Neumann 
directions) and $\mathrm{Dir}$ (for Dirichlet), then 
we can define the periodic function
\begin{equation} \label{con-function}
 \al(x) = \prod_{i \in \mathrm{Dir}} (1 - \eu^{2\pi\ii x_i/\lambda})~.
\end{equation}
The corresponding object
\begin{equation}\label{con-brane-PF}
 \cf^D_\al (\bt,\e; \lambda) = \lambda^{-N} (-2\pi\ii)^{|\mathrm{Dir}|} 
 \oint_\qjk  \prod_{a=1}^r
 \frac {\dif \phi_a} {2\pi \ii} \, \eu^{\sum_a \phi_a t^a }
 \frac
 {\prod_{i\in \mathrm{Neu}} \Ga \left (\frac{x_i}{\lambda} \right)}
 {\prod_{j\in \mathrm{Dir}} \Ga \left (1-\frac{x_j}{\lambda} \right)} 
 \eu^{\ii \pi \sum_{j\in \mathrm{Dir}} \frac{x_j}{\lambda}}
\end{equation}
is the GLSM disk partition function
with mixed boundary conditions \cite{Honda:2013uca}.
We use the identity \cref{gamma-sin} and the same contour $\qjk$ as before
but, due to the property that
the function in \cref{con-function} vanishes at some towers of poles,
these disappear from the final answer.

All our considerations are applicable to these objects,
and depending on the choice of boundary conditions the result 
may (or may not) contain singular terms in $\e$'s at $\e=0$.
It's worth noting that, even when such objects are non-singular,
for example for branes with a compact support,
they cannot be used to fix (regularization scheme dependent)
ambiguities in the GV numbers, as they are blind to such sectors.
The semiclassical part of \cref{con-brane-PF} 
can be interpreted as an integral
\begin{equation}
 \int_M \eu^{\varpi_{\bt}-H_{\e}}
 \frac{\hat \Ga_\eq (TM)}{\hat \Ga_\eq (NM)}
 \eu^{\frac{\ii\pi}{\lambda} c_1 (NM)} + O(\eu^{-\lambda t})
\end{equation}  
over the submanifold $M =\bigcap_{i \in \mathrm{Dir}} D_i$,
where we denote by the same symbol
$\varpi_{\bt}-H_{\e}$ and its pull-back to $M$,
$TM$ stands for tangent bundle and $NM$ for normal bundle of $M$ in $X_\bt$.
This is the equivariant extension of the curvature terms
of the D-brane effective action \cite{Bachas:1999um},
with the $\hat\Ga$-class replaced by \emph{some} square root of $\hat{A}$. 
The story can be generalized to K-theory \cite{Yoshida:2014ssa}.

The disk partition function $\cf^D (\bt,\e; \lambda)$ is well-defined only 
when equivariant parameters are turned on, and for non-compact spaces some 
non-canonical choices are always 
involved when we try to extract the non-equivariant part of the answer. 
Since $\cf^D (\bt,\e; \lambda)$ satisfies the equivariant PF equation,
we can think of it as a generalized period on the mirror \cite{Hori:2000kt}.
We think that equivariant parameters should be taken seriously and one needs
to understand their role in mirror symmetry.
We hope to come back to these issues in the future.

%% file: gamma-formulas.tex
\section{Useful formulas}

We collect useful formulas that we refer
to in the main body of the paper.

The Gamma-class of a complex vector bundle $E$
(whenever $E$ is omitted, it is understood that $E=TX$)
is defined in terms of its Chern roots $x_i$ as the power series
\begin{multline} \label{eq:gamma-class}
 \hat \Ga(E) := \prod_{i} \Ga \left( 1+ \tfrac{x_i}{\lambda} \right) =
 1 - \gamma c_1 \lambda^{-1}
 + \left[ \left( \tfrac{\gamma^2} {2} + \tfrac {\pi^2} {12} \right) c_1^2
 - \tfrac {\pi^2} {6} c_2 \right] \lambda^{-2} \\
 + \left[
	\left( \zeta(3) + \tfrac {\gamma \pi^2} {6} \right) c_1 c_2
	- \left( \tfrac {\zeta(3)} {3} + \tfrac {\gamma^3} {3}
		+ \tfrac {\gamma \pi^2} {12} \right) c_1^3
	- \zeta(3) c_3
 \right] \lambda^{-3} \\
 +\left[
 \left(\tfrac{\pi^4}{90}+\gamma\zeta(3)\right) c_1 c_3
 -\left(\tfrac{\gamma^2\pi^2}{12}+\tfrac{\pi^4}{40}
 +\gamma\zeta(3)\right) c_1^2 c_2
 \right. \\
 +\left.
 \left(\tfrac{\gamma^4}{24}+\tfrac{\gamma^2\pi^2}{24}+\tfrac{\pi^4}{160}
 +\tfrac{\gamma\zeta(3)}{3}\right) c_1^4
 -\tfrac{\pi^4}{90} c_4
 +\tfrac{7\pi^4}{360} c_2^2
 \right]\lambda^{-4}
 + O (\lambda^{-5})~,
\end{multline}
where $\gamma$ is the Euler--Mascheroni constant and the r.h.s.\ is
expanded over a basis of Chern classes $c_i$.
The equivariant version $\hat \Ga_\eq$ is obtained by replacing
Chern roots with equivariant Chern roots.
(For $E=TX$, it is thus a function of the equivariant curvature.)
From the expansion, it follows that for $X=CY_d$
\begin{equation} \label{eq:gamma-class-scl}
\begin{aligned}
 d=2 \Longrightarrow &
 \int_X \eu^{\varpi} \hat \Ga(TX) =
 \tfrac12\int_X \varpi^2
 -\tfrac{\pi^2}{6\lambda^2} \int_X c_2 ~,\\
 d=3 \Longrightarrow &
 \int_X \eu^{\varpi} \hat \Ga(TX) =
 \tfrac1{6}\int_X \varpi^3
 -\tfrac{\pi^2}{6\lambda^2} \int_X \varpi c_2
 -\tfrac{\zeta(3)}{\lambda^3} \int_X c_3 ~,\\
 d=4 \Longrightarrow &
 \int_X \eu^{\varpi} \hat \Ga(TX) =
 \tfrac1{24}\int_X \varpi^4
 -\tfrac{\pi^2}{12\lambda^2} \int_X \varpi^2 c_2
 -\tfrac{\zeta(3)}{\lambda^3} \int_X \varpi c_3
 -\tfrac{\pi^4}{90\lambda^4}  \int_X \left(c_4-\tfrac{7}{4}c_2^2\right)
\end{aligned}
\end{equation}
with the caveat that equivariant versions should be used for non-compact $X$.

The Pochhammer symbol is defined as the function
\begin{equation}
\label{eq:pochhammer}
 (z)_n := \frac{\Ga(z+n)}{\Ga(z)}
\end{equation}
for $n \in \BZ$. It satisfies the following useful identities
\begin{equation}
\label{eq:pochhammer-identity1}
 (z)_n =
 \begin{dcases*}
  \prod_{i=0}^{n-1} (z+i) & if $n>0$ \\
  1 & if $n=0$ \\
  \prod_{i=n}^{-1} \frac1{(z+i)} & if $n<0$
 \end{dcases*}
\end{equation}
and
\begin{equation}
\label{eq:pochhammer-identity2}
 (z)_{-n} = \frac1{(z-n)_n} = \frac{(-1)^n}{(1-z)_n}~.
\end{equation}

The $\q$-analog of the Pochhammer symbol is known as the $\q$-Pochhammer symbol
$(w;\q)_n$. For $n\in\BZ$ it is defined as
\begin{equation} \label{eq:q-pochhammer}
 (w;\q)_n :=
 \begin{dcases*}
  \prod_{i=0}^{n-1} (1-\q^i w) & if $n>0$ \\
  1 & if $n=0$ \\
  \prod_{i=n}^{-1} \frac1{(1-\q^i w)} & if $n<0$
 \end{dcases*}
\end{equation}
and it satisfies the following identity
\begin{equation}
 (w;\q)_{-n} = \frac1{(\q^{-n}w;\q)_n}
 = \frac{(-\q w^{-1})^{n}\q^{\frac{n(n-1)}{2}}}{(\q w^{-1};\q)_n}~.
\end{equation}

Then one can introduce the Jackson $q$-Gamma function
\begin{equation}
\label{eq:q-gamma}
 \Ga_{\q}(z) := \frac{(\q;\q)_{\infty}
 (1-\q)^{1-z}}{\left(\q^z;\q\right)_{\infty}}~,
\end{equation}
which we regard as the $\q$-analog of the Euler Gamma function.
Similarly to \cref{eq:gamma-class} one can use the $\q$-Gamma
function to define a $\q$-Gamma-class in K-theory as
\begin{equation} \label{eq:q-gamma-class}
 \hat \Ga_\q(E) :=
 \prod_{i} \Ga_{\q} \left( 1 + \tfrac{x_i}{\lambda} \right)
 = (\q;\q)_\infty^{\operatorname{rk} E} (1-\q)^{-c_1(E)/\lambda}
 \prod_i \frac{1}{\left(\q L_i;\q\right)_{\infty}}~,
\end{equation}
where $L_i=\eu^{-\hbar x_i}$ are the K-theoretic
Chern roots of $E$ and $\q=\eu^{-\hbar\lambda}$.

The $\q$-Gamma function satisfies the recurrence relation
\begin{equation}
\label{eq:recurrence-q-gamma}
 \frac{1-\q^z}{1-\q}\Ga_{\q}(z)=\Ga_{\q}(z+1)~,
\end{equation}
which is the $\q$-analogue of the standard identity $z\Ga(z)=\Ga(z+1)$.

The infinite $\q$-Pochhammer satisfies the $\q$-difference equation
\begin{equation} \label{eq:q-pochhammer-difference}
 \frac{(1-z)}{(z;\q)_\infty} = \frac{1}{(\q z;\q)_\infty}
\end{equation}
as well as the $\q$-binomial theorem
\begin{equation} \label{eq:qbinomial}
 \frac{1}{(z;\q)_\infty} = \sum_{n=0}^\infty \frac{z^n}{(\q;\q)_{n}}~,
\end{equation}
where we can write the coefficient
as a sum over integer partitions $\mu$ of length $\leq n$
\begin{equation} \label{eq:p-len}
 \frac{1}{(\q;\q)_n} = \sum_{\ell(\mu)\leq n} \q^{|\mu|}~.
\end{equation}

Finally, we recall Euler's reflection formula
\begin{equation} \label{gamma-sin}
 \Ga(1+z) \Ga(1-z)
 = \frac{\pi z}{\sin(\pi z)}
 = \frac{(-2\pi\ii z) \eu^{\pi\ii z}}{(1-\eu^{2\pi\ii z})}
\end{equation}
and its $\q$-analogue
\begin{equation}
 \Ga_\q (1+z) \Ga_\q(1-z) = \frac{(1-\q^z)(\q;\q)_\infty^2}{\theta(\q^z;\q)}~,
\end{equation}
where $\theta(w;\q):=(w;\q)_\infty(\q w^{-1};\q)_\infty$ is a theta function.